\newif\ifblind
\blindfalse

\newif\iflipics
\lipicsfalse

\iflipics
\documentclass[a4paper,UKenglish,cleveref, autoref, thm-restate]{lipics-v2021}
\usepackage[hyperpageref]{backref}

\EventEditors{Rahul Santhanam}
\EventNoEds{1}
\EventLongTitle{39th Computational Complexity Conference (CCC 2024)}
\EventShortTitle{CCC 2024}
\EventAcronym{CCC}
\EventYear{2024}
\EventDate{July 22--25, 2024}
\EventLocation{Ann Arbor, MI, USA}
\EventLogo{}
\SeriesVolume{300}
\ArticleNo{2}
\else
\documentclass[11pt]{article}
\usepackage{amsmath,amsthm}
\usepackage{fullpage}
\usepackage{thm-restate}

\usepackage[pagebackref]{hyperref}
\usepackage[capitalize,nameinlink]{cleveref}
\fi

\usepackage[T1]{fontenc}
\usepackage[utf8]{inputenc}
\usepackage{amsfonts,amssymb,bm,bbm}
\usepackage{mathtools}
\usepackage{braket}
\usepackage{physics}
\usepackage{thmtools}
\usepackage{graphicx}
\usepackage{enumitem}
\usepackage{xcolor}
\usepackage{boxedminipage}
\usepackage{tikz}
\usetikzlibrary{matrix}
\usepackage{xspace}
\usepackage{graphicx}
\usepackage{subcaption}
\usepackage{physics}
\usepackage[breakable]{tcolorbox}
\usepackage{crossreftools}

\makeatletter
\newcommand{\tcbline}{%
  \begin{tikzpicture}
  \path[use as bounding box] (0,0) -- (\linewidth,0);
  \draw[color=black!20!white,dashed,dash phase=2pt]
        (0-\kvtcb@leftlower-\kvtcb@boxsep,0)--
        (\linewidth+\kvtcb@rightlower+\kvtcb@boxsep,0);
  \end{tikzpicture}%
  }
\makeatother

\makeatletter
\def\@footnotecolor{red}
\define@key{Hyp}{footnotecolor}{%
 \HyColor@HyperrefColor{#1}\@footnotecolor%
}
\def\@footnotemark{%
    \leavevmode
    \ifhmode\edef\@x@sf{\the\spacefactor}\nobreak\fi
    \stepcounter{Hfootnote}%
    \global\let\Hy@saved@currentHref\@currentHref
    \hyper@makecurrent{Hfootnote}%
    \global\let\Hy@footnote@currentHref\@currentHref
    \global\let\@currentHref\Hy@saved@currentHref
    \hyper@linkstart{footnote}{\Hy@footnote@currentHref}%
    \@makefnmark
    \hyper@linkend
    \ifhmode\spacefactor\@x@sf\fi
    \relax
  }%
\makeatother

\hypersetup{
  linkcolor  = blue!80!black,
  citecolor  = green!50!black,
  footnotecolor = red!90!black,
  colorlinks = true,
}

\makeatletter
\newcommand{\subalign}[1]{%
  \vcenter{%
    \Let@ \restore@math@cr \default@tag
    \baselineskip\fontdimen10 \scriptfont\tw@
    \advance\baselineskip\fontdimen12 \scriptfont\tw@
    \lineskip\thr@@\fontdimen8 \scriptfont\thr@@
    \lineskiplimit\lineskip
    \ialign{\hfil$\m@th\scriptstyle##$&$\m@th\scriptstyle{}##$\hfil\crcr
      #1\crcr
    }%
  }%
}
\makeatother

\newlist{steps}{enumerate}{10}
\setlist[steps]{label=Step \arabic*:, ref=\arabic*}
\crefname{stepsi}{Step}{Steps}

\makeatletter
\newcommand{\envfootnote}{%
\addtocounter{footnote}{+1}%
\addtocounter{Hfootnote}{+1}%
\global\let\Hy@saved@currentHref\@currentHref%
\hyper@makecurrent{Hfootnote}%
\global\let\Hy@footnote@currentHref\@currentHref%
\global\let\@currentHref\Hy@saved@currentHref%
}
\makeatother

\iflipics
\else
\newtcolorbox[
  auto counter,
  number within=section,
  number freestyle={\noexpand\thesection.\noexpand\arabic{\tcbcounter}},
  crefname={Protocol}{Protocols}]%
{protocol}[2][]{
  breakable,
  colback=black!3!white,
  colframe=black!20!white,
  coltitle=black,
  code={\def\mytitle{\normalfont #2}},
  fonttitle=\bfseries,
  title=Protocol \thetcbcounter: \mytitle, #1,
  beforeafter skip=\baselineskip}

\newtcolorbox[
  auto counter,
  number within=section,
  number freestyle={\noexpand\thesection.\noexpand\arabic{\tcbcounter}},
  crefname={Algorithm}{Algorithms}]%
{algorithm}[2][]{
  breakable,
  colback=black!3!white,
  colframe=black!20!white,
  coltitle=black,
  code={\def\mytitle{\normalfont #2}},
  fonttitle=\bfseries,
  title=Algorithm \thetcbcounter: \mytitle, #1,
  beforeafter skip=\baselineskip}

\creflabelformat{equation}{#2#1#3}
\newtheorem{theorem}{Theorem}[section]
\newtheorem{lemma}[theorem]{Lemma}
\newtheorem{corollary}[theorem]{Corollary}
\counterwithin{corollary}{section}
\newtheorem{claim}[theorem]{Claim}
\newtheorem{proposition}[theorem]{Proposition}

\theoremstyle{definition}
\newtheorem{remark}[theorem]{Remark}
\newtheorem{definition}[theorem]{Definition}
\fi
\newtheorem{openproblem}{Open problem}

\title{Streaming Zero-Knowledge Proofs}
\ifblind\else
\iflipics
\titlerunning{Streaming Zero-Knowledge Proofs}

\author{Graham Cormode}{University of Warwick, UK \and \url{http://dimacs.rutgers.edu/~graham/}}{g.cormode@warwick.ac.uk}{https://orcid.org/0000-0002-0698-0922}{The work of Graham Cormode and Chris Hickey was supported by European Research Council grant ERC-2014-CoG 647557.}
\author{Marcel Dall'Agnol}{Princeton University, USA \and \url{https://marceldallagnol.github.io/}}{dallagnol@princeton.edu}{https://orcid.org/0000-0002-6060-1663}{}
\author{Tom Gur}{University of Cambridge, UK \and \url{https://www.cst.cam.ac.uk/people/tg508}}{tom.gur@cl.cam.ac.uk}{https://orcid.org/0000-0001-7864-7013}{Tom Gur is supported by UKRI Future Leaders Fellowship MR/S031545/1 and EPSRC New Horizons Grant EP/X018180/1.}
\author{Chris Hickey}{University of Manchester, UK}{cjahickey1994@gmail.com}{}{}

\authorrunning{G. Cormode, M. Dall'Agnol, T. Gur and C. Hickey}

\Copyright{Graham Cormode, Marcel Dall'Agnol, Tom Gur and Chris Hickey}



\else
\author{Graham Cormode\thanks{The work of Graham Cormode and Chris Hickey was supported by European Research Council grant ERC-2014-CoG 647557.}\\University of Warwick \and Marcel Dall'Agnol\\Princeton University \and Tom Gur\thanks{Tom Gur is supported by UKRI Future Leaders Fellowship MR/S031545/1 and EPSRC New Horizons Grant EP/X018180/1.}\\University of Cambridge \and Chris Hickey$^*$\\University of Manchester}
\fi

\newcommand{\zksip}{\textrm{zkSIP}}
\newcommand{\eat}[1]{}
\DeclareMathOperator{\view}{View}
\newcommand{\E}{\mathbb E}
\newcommand{\N}{\mathbb N}
\newcommand{\Z}{\mathbb Z}
\newcommand{\field}{\mathbb{F}}

\newcommand{\R}{\mathbb R}
\renewcommand{\P}{\mathbb P}
\renewcommand{\set}[1]{\left\{#1\right\}}
\newcommand{\poly}{\operatorname{poly}}
\newcommand{\polylog}{\operatorname{polylog}}
\newcommand{\eps}{\varepsilon}
\newcommand{\bitset}{\{0,1\}}
\newcommand{\indexproblem}{\normalfont\textsc{index}\xspace}
\newcommand{\searchindex}{\normalfont\textsc{search-index}\xspace}
\newcommand{\decisionindex}{\normalfont\textsc{decision-index}\xspace}
\newcommand{\pair}{\normalfont\textsc{pair}\xspace}
\newcommand{\reconstruct}{\normalfont\textsc{reconstruct}\xspace}
\newcommand{\pep}{\normalfont\textsc{pep}\xspace}
\newcommand{\pepprotocol}{\textsf{pep}\xspace}
\newcommand{\sumcheck}{\normalfont\textsc{sumcheck}\xspace}
\newcommand{\frequencymoment}{\normalfont\textsc{frequency-moment}\xspace}
\newcommand{\pointquery}{\normalfont\textsc{point-query}\xspace}
\newcommand{\rangecount}{\normalfont\textsc{range-count}\xspace}
\newcommand{\selection}{\normalfont\textsc{selection}\xspace}
\newcommand{\innerproductproblem}{\normalfont\textsc{inner-product}\xspace}
\newcommand{\median}{\normalfont\textsc{median}\xspace}
\newcommand{\vcommitlength}{v}

\newcommand{\vcommitstring}{z}
\newcommand{\pcommitlength}{p}

\newcommand{\pcommitstring}{y}
\newcommand{\dimension}{m}

\newcommand{\degree}{d}

\newcommand{\fieldsize}{q}

\newcommand{\evaluationpoint}{\bm{\fe2}}
\newcommand{\fingerprinttuple}{\bm{\rho}}
\newcommand{\Fingerprinttuple}{\bm{\sigma}}
\renewcommand{\line}{L}
\newcommand{\snapshot}{b}
\newcommand{\whiteboxoracle}{\mathcal{W}}
\newcommand{\cubeside}{H}

\newcommand{\ball}{\mathcal B}
\newcommand{\alphabet}{\Gamma}
\newcommand{\alphabetsize}{\gamma}
\newcommand{\chf}{\mathbbm 1}

\def\symbol#1{
    \ifnum#1=1
        \alpha
    \else
    \ifnum#1=2
        \beta
    \else
    \ifnum#1=3
        \gamma
    \fi\fi\fi
}

\def\fe#1{
    \ifnum#1=1
        \alpha
    \else
    \ifnum#1=2
        \beta
    \else
    \ifnum#1=3
        \gamma
    \fi\fi\fi
}

\def\rfe#1{
    \ifnum#1=1
        \rho
    \else
    \ifnum#1=2
        \sigma
    \else
    \fi\fi
}

\begin{document}
\pagenumbering{roman}

\date{}

\maketitle

\begin{abstract}
    Streaming interactive proofs (SIPs) enable a space-bounded algorithm with one-pass access to a massive stream of data to verify a computation that requires large space, by communicating with a powerful but untrusted prover.
    
    This work initiates the study of \emph{zero-knowledge} proofs for data streams. We define the notion of zero-knowledge in the streaming setting and construct zero-knowledge SIPs for the two main algorithmic building blocks in the streaming interactive proofs literature: the \emph{sumcheck} and \emph{polynomial evaluation} protocols. To the best of our knowledge \emph{all} known streaming interactive proofs are based on either of these tools, and indeed, this allows us to obtain zero-knowledge SIPs for central streaming problems such as index, point and range queries, median, frequency moments, and inner product.
    
    Our protocols are efficient in terms of time and space, as well as communication: the verifier algorithm's space complexity is $\polylog(n)$ and, after a non-interactive setup that uses a random string of near-linear length, the remaining parameters are $n^{o(1)}$.
    
    En route, we develop an algorithmic toolkit for designing zero-knowledge data stream protocols, consisting of an \emph{algebraic streaming commitment protocol} and a \emph{temporal commitment protocol}. Our analyses rely on delicate algebraic and information-theoretic arguments and reductions from average-case communication complexity.
\end{abstract}


\newpage
\tableofcontents
\newpage

\pagenumbering{arabic}
\section{Introduction}
\label{sec:intro}

The design and analysis of algorithms in the \emph{streaming model} is an exceptionally active area of research, particularly so in recent years (see, e.g., the surveys \cite{M05, M14, C23} and references therein). A streaming algorithm $A$ observes a long data stream $x = (x_1, \ldots, x_n)$, whose size far exceeds $A$'s limited memory, one symbol at a time, and computes some pre-specified information about $x$ (e.g., statistics such as the number of distinct elements). To successfully do so, $A$ maintains a summary of the stream that is both small and easily updateable. Algorithmic techniques and data structures that work under such constraints underpin both theoretical progress and real-world deployment of algorithms for massive datasets.   

However, it has long been known that many natural and important problems are hard in the streaming model \cite{AMS99}. This motivates the study of protocols for \emph{delegation of computation}, whereby a streaming algorithm offloads expensive operations to an untrusted party with large memory, but can still verify (in low space) that the purported result is correct.  
Accordingly, \emph{interactive proofs} in the data stream model received a great deal of attention in the last decade \cite{CTY11, CMT12, CMT13, T13, CCGT14, T14, CCMTV15, DTV15, CH18, CG19, CGT20}.

Streaming interactive proofs (SIPs) are delegation-of-computation protocols where the computationally bounded party is bounded \emph{not} in its time complexity, but rather in space and input access. More precisely, an SIP is an interactive protocol between a powerful but untrusted prover $P$ and a space-bounded streaming verifier $V$ which has sequential, one-pass access to a massive input as well as the prover's messages. We note that prover and verifier observe the same stream of bits, which only the former can store in its entirety.

Remarkably, SIPs allow low-space streaming algorithms to efficiently verify key problems in the data stream model that are completely intractable without the assistance of a prover. Indeed, the aforementioned sequence of works constructed SIPs with polylogarithmic-space verifiers for a large collection of problems, many of which require linear space for a streaming algorithm alone (such as the \indexproblem and frequency moment problems). The underlying power that enables exponential separations between streaming algorithms and SIPs essentially boils down to two powerful protocols: \emph{sumcheck} and \emph{polynomial evaluation}, which can in turn be applied to a plethora of problems.

Determining the extent to which SIPs can be augmented with extra features is the natural next step to a refined understanding of the complexity landscape around them. Our work focuses on \emph{zero-knowledge}: ensuring that the protocol reveals no information besides what it is designed to compute. We remark that this feature widens the array of computational tasks solvable by mutually distrusting parties, thus supporting numerous cryptographic protocols currently in use \cite{BCGGMTV14,BBHR18,BCRSVW19}.

Despite the fundamental role of zero-knowledge in theoretical computer science (see, e.g., \cite{V99, G02, V07, G08} and references therein) and the extensive study of SIPs over the last decade, no zero-knowledge SIPs were known prior to this work. Indeed, it is not obvious a priori whether they are at all possible: for instance, while traditional zero-knowledge prevents leakage of information to a polynomial-time adversary about some hard computation on an input $x$ (e.g., a witness that certifies $x$ is in a language), in the streaming setting a space-bounded verifier must learn no additional information \emph{about $x$ itself} -- even if its runtime is unbounded.

\subsection{Zero-knowledge in the streaming model}
Recall that in the traditional setting, which deals with \emph{polynomial-time} algorithms, a protocol is zero-knowledge if, for every (possibly malicious) verifier $\widetilde V$, there exists a simulator $S_{\widetilde V}$ whose output cannot be told apart (either computationally or statistically) from a real interaction between $P$ and $\widetilde V$ by any distinguisher $D$; and if this holds up to negligibly small error, the protocol can be safely repeated or composed.

In the streaming model, algorithms are restricted to \emph{one-pass sequential access} to their input and the primary resource is \emph{space}, rather than time. Accordingly, we say that an SIP is zero-knowledge if $\widetilde V$, $S$ and $D$ are streaming algorithms; when $\widetilde V$ has $s$ bits of memory, the simulator has roughly $s$ space and we allow the distinguisher $D$ to have an arbitrary $\poly(s)$ amount of memory. (See \cref{sec:streaming-zk} for formal definitions.) Albeit similar, this notion is distinct to its poly-time analogue in two fundamental ways.

Negligible distinguishing bias is a robust notion of security in the setting of polynomial-time computation because it prevents polynomial-time adversaries from boosting their advantage by repeating (polynomially) many executions. However, in the data stream model, the \emph{one-pass} restriction on input access precludes this strategy altogether; indeed, streaming problems often become trivial with a single additional pass. 
We therefore define secure protocols as those achieving $o(1)$ distinguishing bias, which ensures that the probability of information leakage tends to zero. (See \cref{rem:subconstant-bias} for a more detailed discussion of alternative ``hybrid'' models and security bounds.)

The second crucial distinction is that the notion of zero-knowledge for SIPs is \emph{unconditional}, i.e., does not rely on computational assumptions, faithfully to the nature of the data stream model. This differs markedly from past work on zero-knowledge protocols where the verifier is able to process incoming messages in a streaming fashion (e.g., \cite{GKR08,CMT12}), 
whose zero-knowledge property is still with respect to the standard setting: while the honest verifier is a streaming algorithm, the protocols are only secure against polynomial-time adversaries. In this work, \emph{adversaries are also streaming algorithms}.

\iflipics\else\paragraph*{}\fi This paper explores the extent to which zero knowledge streaming interactive proofs (\zksip{}s) can outperform streaming algorithms: does there exist a problem they solve more efficiently? If so, can they do so for a natural problem such as \indexproblem, or even more ambitiously, achieve an exponential reduction in the space complexity for key problems in the data stream model?

\subsection{Main results}

Our main contribution is a strong positive answer to the questions above, providing the tools to construct zero-knowledge streaming interactive proofs for essentially any problem within the reach of current (non-zero-knowledge) SIPs.

In more detail, our main results are zero-knowledge versions of the two building blocks underlying all known SIPs: the \emph{sumcheck} and \emph{polynomial evaluation} protocols, from which we derive \zksip{}s for central streaming problems in \cref{sec:applications}. In doing so, we obtain \emph{exponentially} smaller space complexity for the fundamental \indexproblem and frequency moment problems (among others) when compared to streaming algorithms alone. 

We remark that all our \zksip{}s are two-stage protocols with a \emph{setup} and an \emph{interactive} stage. The setup is non-interactive and consists merely of a random string (see \cref{sec-to:vp-commitment}), which can be reused in multiple interactive executions (of possibly different protocols).\footnote{We also remark that omitting the setup yields an honest-verifier (but not malicious-verifier) \zksip{} with $n^{o(1)}$ communication complexity.} With this simple preprocessing step, we achieve essentially optimal time and communication complexities (i.e., subpolynomial or even polylogarithmic -- as do the best non-zero-knowledge SIPs -- and dramatically smaller than the complexity of streaming the input) in the interactive stage.

\paragraph*{Sumcheck Zero-Knowledge SIP\iflipics\else.\fi}
In the $\sumcheck$ problem, the goal is to compute the sum of evaluations of a low-degree polynomial over a large structured set (a subcube). Protocols for \sumcheck are some of the most important building blocks for interactive proofs, and are extremely useful for SIPs in particular.

We state the following theorem in generality, but note that standard parameter settings imply space complexity $s = \polylog(n)$ as well as $O(n^{1 + \delta})$ (for any constant $\delta > 0$) and $n^{o(1)}$ communication in the setup and interactive stages, respectively. (The time complexity is of the same order as the communication in both stages.) This is the case in all of our applications.

\begin{theorem}[\cref{thm:sumcheck-correctness,thm:sumcheck-zk}, informally stated]
    \label{infthm:zk-sumcheck}
    There exists a \zksip{} for \sumcheck where, for $\dimension$-variate low-degree polynomials over $\field$, the verifier uses $s = O(\dimension^2 \log \abs{\field})$ bits of space. The SIP communicates $\tilde O(\abs{\field}^\dimension)$ bits in its setup and $\abs{\field}^{\log\log \abs{\field} + O(1)}$ bits in the interactive stage.
\end{theorem}

The round complexity (the number of messages sent or received by each party throughout the SIP) is $\dimension + O(1)$, a small constant larger than that of the standard sumcheck protocol.

We stress that while sumcheck is traditionally used (in the polynomial-time setting) to verify exponentially large sums in polynomial time, this is \emph{not} the goal of the streaming variant, as sums of evaluations over a large set can be obtained incrementally for functions computable in low space (a class that includes polynomials).

Nevertheless, the sumcheck protocol achieves exponential savings in space complexity for problems that require large space without interaction: it enables efficient verification of sums of polynomials implicitly defined  input defines implicitly, which require \emph{linear space} to compute otherwise.

\paragraph*{Polynomial Evaluation Zero-Knowledge SIP\iflipics\else.\fi}
We proceed to our second main result: a zero-knowledge SIP for the \emph{polynomial evaluation problem} \pep, which consists of computing a low-degree polynomial at a single point (revealed after the description of the polynomial). It allows a streaming algorithm to recover data that was seen but not stored, by saving a small \emph{fingerprint} of the stream.
Similarly to sumcheck, general-purpose \pep protocols are widely applicable to the design of SIPs.

\begin{theorem}[\cref{thm:pep-correctness,thm:pep-zk}, informally stated]
    \label{infthm:zk-pep}
    There exists a \zksip{} for \pep where, for $\dimension$-variate low-degree polynomials over $\field$, the verifier uses $O(\dimension \log \abs{\field})$ bits of space. The communication complexity is $\tilde O(\abs{\field}^\dimension)$ in the setup and $\poly(\abs{\field})$ bits in the interactive stage.
\end{theorem}
As in \cref{infthm:zk-sumcheck}, standard parameter settings imply \zksip{}s with polylogarithmic space, $n^{o(1)}$ time and communication complexity (in the interactive stage)\footnote{A nontrivial security guarantee still holds with $\polylog(n)$ communication, but with $n^{o(1)}$ the protocol becomes secure against arbitrary $\polylog(n)$-space adversaries; see \cref{rem:indistinguishability-pep}.} as well as near-linear communication in the setup. The round complexity is $O(1)$.

\subsubsection{Streaming commitment protocols}
En route to proving \cref{infthm:zk-sumcheck,infthm:zk-pep}, we construct tools for the design of \zksip{}s which we find of independent interest. Namely, we provide two types of \emph{commitment protocols} for streaming algorithms.

We remark that in the polynomial-time setting, the existence of secure commitment schemes is equivalent to the existence of one-way functions \cite{IL89, N91, HILL99}, so it may seem surprising that our results hold \emph{unconditionally}. However, in the incomparable model of streaming algorithms, which are not time-bounded, but are instead severely constrained with respect to space and input access, we show that no cryptographic assumption is needed.\footnote{We refer to commitment \emph{protocols} rather than schemes in the streaming model to avoid ambiguity with the polynomial-time analogue; see \cref{def:commitment}.}

\paragraph*{Streaming algebraic commitment protocol\iflipics\else.\fi}
The following result shows that not only does a streaming commitment protocol exist, but that it can be made \emph{linear}; that is, the sender may commit to a sequence of messages and decommit to a linear combination thereof, with linear coefficients of the receiver's choosing.

\begin{theorem}[\cref{thm:pv-algebraic-commitment}, informally stated]
    \label{infthm:pv-commitment}
    There exists a commitment protocol whereby an unbounded-space sender commits a tuple $\bm{\fe1} \in \field^\ell$ to a streaming receiver and decommits to a linear combination $\bm{\fe1} \cdot \evaluationpoint$, with linear coefficients $\bm{\fe2}$ chosen by the receiver. The receiver's space complexity is $O (\ell \log \abs{\field})$ and the protocol communicates $\tilde O\left(\abs{\field}^{3\ell}\right)$ bits.
\end{theorem}

\paragraph*{Temporal commitment protocol\iflipics\else.\fi}
The second component is a new notion of a streaming commitment, which we call \emph{temporal}.
This protocol allows a streaming \emph{verifier} to ``timestamp'' its message, providing evidence that it was chosen before streaming a particular input.

\begin{theorem}[\cref{thm:correct-set}, informally stated]
    \label{infthm:correct-set}
    Let $\alphabet$ be an alphabet and $A$ a space-$s$ streaming algorithm with $s = \polylog \abs{\alphabet}$. If $A$ streams $z \sim \alphabet^\vcommitlength$ and $\vcommitlength$ is large enough, the following holds: \emph{independently of its computation} after $\vcommitstring$, with high probability $A$ can output at most $s$ symbol-certificate pairs $(\symbol1, i) \in \alphabet \times [\vcommitlength]$ such that $\symbol1 = \vcommitstring_i$.
\end{theorem}

In other words, $A(\vcommitstring)$ cannot remember more than $s$ symbol-certificate pairs for the string $\vcommitstring$; and the bound is unchanged if $A$ obtains information uncorrelated with $\vcommitstring$ after reading the stream.

\subsection{Applications}
\label{sec:applications}

Recall that \cref{infthm:zk-pep,infthm:zk-sumcheck} provide zero-knowledge versions of the general tools that essentially underlie all known SIPs, namely, the \emph{sumcheck} and \emph{polynomial evaluation} protocols. We demonstrate the power and flexibility of our tools by deriving from them explicit \zksip{}s for streaming problems of fundamental importance: \indexproblem and \frequencymoment, as well as \pointquery, \rangecount, \selection and \innerproductproblem.

As mentioned in the previous section, while the following statements highlight space complexities, the communication complexities are $n^{o(1)}$ in the interactive stage and $O(n^{1 + \delta})$ for arbitrarily small $\delta$ in the setup stage.

In the \indexproblem problem, a streaming algorithm reads a length-$n$ string $x$ followed by an index $j \in [n]$, and its goal is to output $x_j$. \indexproblem is a hard problem for streaming algorithms, requiring \emph{linear} space to solve \cite{RY20}. By instantiating our \zksip{} for polynomial evaluation with respect to the low-degree extension of the input evaluated at the index $j$, we obtain the following.

\begin{corollary}[\cref{cor:index}, informally stated]
    \label{infcor:index}
    There exists a \zksip{} for $\indexproblem$ with logarithmic verifier space complexity.
\end{corollary}
\noindent Note that this matches the space complexity of the non-zero-knowledge SIP of \cite{CCMTV15, CCMTV19}.

In the $\frequencymoment_k$ (or $F_k$) problem, an algorithm streams $x \in [\ell]^n$ and its task is to compute $F_k(x) = \sum_{i \in [\ell]} \varphi_i^k$, the $k^\text{th}$ moment of the frequency vector $(\varphi_1,\ldots,\varphi_\ell)$, where $\varphi_i$ is the number of occurrences of $i$ in $x$. This is a central problem in the streaming literature, which is well known to require linear space to compute \cite{AMS99}; by instantiating our sumcheck protocol with respect to the low-degree extension \emph{of the frequency vector}, we obtain a zero knowledge protocol for the exact computation of $F_k$.

\begin{corollary}[\cref{cor:frequency-moment}, informally stated]
    \label{infcor:frequency-moment}
    For every $\ell \in [n]$ and $k$, there exists a \zksip{} that computes $F_k$ with $\polylog(n)$ verifier space complexity.
\end{corollary}

Lastly, we illustrate the flexibility of our protocols by constructing additional \zksip{}s for several other problems: \pointquery (where the input is a stream of integer updates to an $\ell$-dimensional vector $y$ followed by an index $j$ and the task is to output $y_j$); \rangecount (where the input is a sequence of points in $[\ell]$ followed by a range $R \subseteq [\ell]$ and the task is to output the number of occurrences in $R$); \selection (which generalises the computation of the median); and \innerproductproblem (where the task is to output the inner product between the frequency vectors of a pair of streams).

\begin{corollary}[\cref{cor:point-query,cor:range-count,cor:selection,cor:inner-product}, informally stated]
    \label{infcor:inner-product}
    There exist $\polylog(n)$-space \zksip{}s for \pointquery, \rangecount, \selection and \innerproductproblem.
\end{corollary}

\subsection{Related work}

    This work builds on the line of research on streaming interactive proofs, initiated by \cite{CCM09,CCMT14} and actively investigated over the last decade \cite{CMT13, CTY11, CMT12, CCGT14, T14, GR15, CCMTV15, DTV15, ADRV16, CH18, CCMTV19, G20, CGT20}. These sublinear interactive proofs are also closely related to proofs of proximity \cite{RVW13, G17, GR17, GR18, GGR18, RR20,CG18a, GG21, GLR21, DGRT22}.
    
    Indeed, our two main results can be seen as zero-knowledge versions of the main techniques in \cite{CCMTV19} and \cite{CMT12}: respectively, a polynomial evaluation and a sumcheck protocol. (We note that while \cite{CFGS22} construct a zero-knowledge sumcheck protocol via an algebraic commitment scheme, their model and techniques are completely different.)
    
    Past work has studied zero-knowledge protocols where the verifier is able to process incoming messages in a streaming fashion (e.g., \cite{GKR08,CMT12}), but their zero-knowledge property is with respect to the standard, \emph{polynomial-time}, setting; that is, while the honest verifier is a streaming algorithm, the security of the protocol holds against polynomial-time adversaries, whereas we consider adversaries that are also streaming algorithms.

    We note that while unconditional cryptographic primitives such as bit commitments and key agreement are achievable in the \emph{bounded-storage model} (see, e.g., \cite{GZ19} and references therein), the security guarantees are weaker, allowing at most a quadratic, rather than arbitrary polynomial, gap between honest and malicious parties. Recent work on the streaming variant of the model \cite{DQW22,DQW23} is more closely related to ours. However, they do not construct commitment schemes and, more importantly, these results assume bounds on the space of both parties; therefore, they do not immediately apply to \emph{statistically sound} proof systems such as those considered in this work.
    
    We also note that while zero-knowledge proofs within sublinear models of computation have been actively explored in the last decade (e.g., \cite{BRV18, IW14}), our work is the first to do so in the streaming model.

\subsection{Open problems}

    This work opens several avenues for future research; in this short section, we highlight four particularly compelling directions.
    
    Achieving zero-knowledge versions of the main building blocks in the SIP literature suggests a natural question: can \emph{all} SIPs be endowed with zero-knowledge? That is, denoting by \textsf{SIP} (respectively, \textsf{zkSIP}) the class of languages that admit SIPs (respectively, \zksip{}s) with $\polylog(n)$ space complexity, we raise the following problem.
    \begin{openproblem}
        Is $\mathsf{SIP}$ equal to $\mathsf{zkSIP}$?
    \end{openproblem}

    In our two-stage protocols, the communication complexity is dominated by the setup (a reusable random string of near-linear length); the remainder of the protocol is extremely efficient, with $n^{o(1)}$ (or even $\polylog n$) communication and time complexity. Making this parameter sublinear would be a major step towards practical applicability.

    \begin{openproblem}
        Can zero-knowledge SIPs achieve sublinear communication complexity?
    \end{openproblem}
    
    Lastly, recall that the notion of security in this work is (unconditional and) \emph{computational}, where streaming adversaries detect a simulation with at most $o(1)$ bias. It is natural to ask whether stronger notions are achievable -- both with respect to an adversary's capabilities and feasible security bounds.
    
    \begin{openproblem}
        Are there SIPs with statistical (or even perfect) zero-knowledge?
    \end{openproblem}

    \begin{openproblem}
        Can security bounds of $\frac1{\poly(n)}$ or $\frac1{n^{\omega(1)}}$ be obtained for computational \zksip{}s?
    \end{openproblem}
    
\subsection*{Organisation}
    
    The rest of the paper is organised as follows. In \cref{sec:technical-overview} we give a high-level overview of the challenges and the techniques we use to endow SIPs with zero-knowledge. We briefly discuss the preliminaries for the technical sections in \cref{sec:preliminaries}, and, in \cref{sec:streaming-zk}, formally define the notion of streaming zero-knowledge and discuss key conceptual points. In \cref{sec:commitment-protocols} we construct the two commitment protocols that comprise the main components for our polynomial evaluation and sumcheck protocols. We construct the protocols, prove their zero-knowledge property and show applications for them in \cref{sec:pep,sec:sumcheck}, respectively.

\section{Technical overview}
    \label{sec:technical-overview}
    
    We provide a high-level overview of the techniques we use and build upon in this paper. For concreteness, we illustrate our methodology by focusing on the construction of zero knowledge SIPs for one of the most fundamental problems in the data stream model: \indexproblem.
    
    We begin with a bird's eye view of our ideas and the challenges that arise in their implementation. The starting point of our efforts is \cref{sec-to:pep}, where we describe the \emph{polynomial evaluation protocol} (\pepprotocol), from which a (non zero-knowledge) SIP for the \indexproblem problem follows.
    An attempt to make this protocol zero-knowledge faces two fundamental challenges, which we address in \cref{sec-to:pv-commitment,sec-to:vp-commitment} via the construction of two types of \emph{streaming commitment protocols}.
    
    In \cref{sec-to:index-protocol}, we apply the foregoing protocols to obtain a streaming interactive proof for \indexproblem and provide an overview of the proof of its zero-knowledge property, which requires an involved simulator argument. Finally, \cref{sec-to:sumcheck} sketches another application of this framework that obtains an additional powerful and flexible tool: a \emph{zero-knowledge streaming sumcheck} protocol.

\subsection{A starting point: the polynomial evaluation protocol}
\label{sec-to:pep}

    Recall that in the \indexproblem problem, a streaming algorithm with $s$ bits of memory receives a length-$n$ string $x$ over an alphabet $\alphabet$, followed by a coordinate $j \in [n]$, and its goal is to output $x_j \in \alphabet$. It is well-known that \indexproblem is maximally hard for streaming algorithms, requiring $s = \Omega(n)$ space for the output to be correct with nontrivial probability.
    
    First, note that obtaining an efficient SIP for \indexproblem is non-trivial even without zero-knowledge. Indeed, the naive approach of having the prover $P$ reveal the index $j$ before $V$ streams $x$ (allowing the verifier to save $x_j$) fails: both parties observe \emph{the same} stream of information, so $P$ only learns $j$ long after $V$ has seen $x_j$. Any communication in an SIP before the input stream must therefore be \emph{independent} of it.
    
    Remarkably, an exponential reduction in space complexity is possible despite both prover and verifier not knowing the index $j$ before it appears in the stream. We recall the SIP in \cite{CCMTV19}, upon which we build, and argue why it is \emph{not} zero-knowledge to begin with. Their SIP is an application of \pepprotocol, the \emph{polynomial evaluation protocol}, which enables a small-space algorithm to recover any element that was streamed but not stored, using only a small fingerprint of the stream.

    We embed the input stream into an object with algebraic structure in a space of size much larger than $n$, namely, by viewing $x_i \in \field$, for a large enough finite field $\field$, and considering an $\dimension$-variate \emph{low-degree polynomial} $\hat x$ that interpolates across all $x_i$; we call the polynomial $\hat x: \field^\dimension \to \field$ of individual degree $\degree = \degree(\dimension, n)$ the low-degree extension (LDE) of $x$. (Usual parameter settings satsify $\degree, \dimension \leq \log n$ and $\abs{\field} = \polylog(n)$.)
    
    The protocol proceeds as follows. The verifier samples a random evaluation point $\fingerprinttuple \sim \field^\dimension$ and computes the \emph{fingerprint} $\hat x(\fingerprinttuple)$, which can be evaluated in low space via standard online Lagrange interpolation. After $V$ learns $j$, it enlists $P$ in the recovery of $x_j$: it sends $P$ a line $\line: \field \to \field^\dimension$ incident to $j$ (viewing this index as an element of $\field^\dimension$) and $\fingerprinttuple$, where $\line(0) = j$ and $\line(\rfe1) = \fingerprinttuple$ for a random $\rfe1 \sim \field$, whereupon $P$ replies with the (low-degree) univariate polynomial $\hat x_{|\line} = \hat x \circ \line$.
    
    If $P$ is honest, then $V$ can easily recover $x_j = \hat x(j) = \hat x_{|\line}(0)$. However, $P$ could easily cheat if $V$ made no further checks: the prover could just as well pick $\fe1 \in \field$ arbitrarily and send any low-degree polynomial $g$ such that $g(0) = \fe1$ to (falsely) convince $V$ that $x_j = \fe1$. By having $V$ only accept the prover's claim that $x_j = g(0)$ \emph{if $g$ also agrees with the fingerprint}, i.e., if $g(\rfe1) = \hat x_{|\line}(\rfe1) = \hat x(\fingerprinttuple)$, the verifier thwarts this (and any other) attack: since both $\fingerprinttuple$ and $\rfe1$ are unknown to the prover, to convince the verifier of an incorrect answer $g(0) \neq \hat x_{|\line}(0)$, the prover must send a polynomial $g \neq \hat x_{|\line}$ that agrees with $\hat x_{|\line}$ at a random point; and if $\field$ is sufficiently large, the probability of this event ($\rfe1$ being a root of the nonzero polynomial $g - \hat x_{|\line}$) is arbitrarily small.
    
\begin{figure}[ht]
	\centering
     \begin{subfigure}[t]{0.48\textwidth}
       \center
       \includegraphics[width=\textwidth]{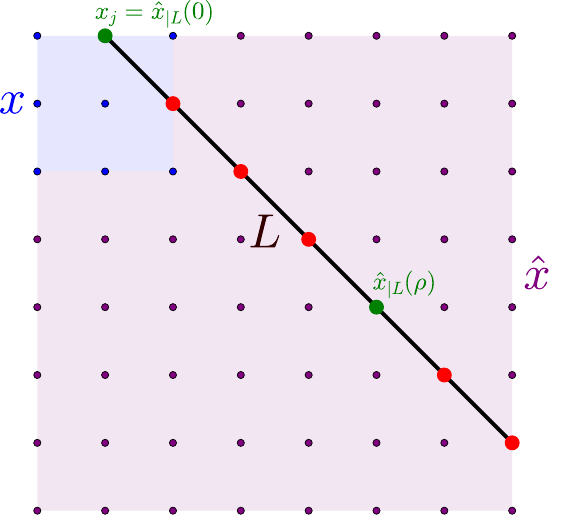}
	\caption{$V$ streams $x$ (in blue), learns $\hat x(\fingerprinttuple) = \hat x_{|\line}(\rfe1)$ and sends $\line$. The prover replies with $\hat x_{|\line}$, revealing $x_j$ and $\hat x(\fingerprinttuple)$ (in green) along with evaluations of $\hat x$ that $V$ cannot learn on its own (in red).}
	    \label{fig:pep-leakage}
     \end{subfigure}%
     \hfill
     \begin{subfigure}[t]{0.48\textwidth}
        \center
        \includegraphics[width=\textwidth]{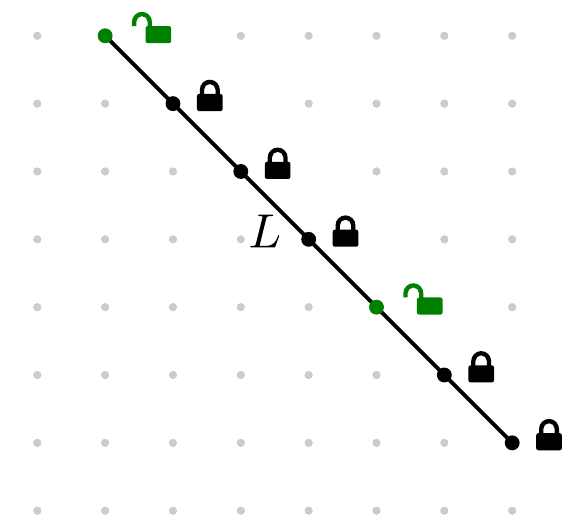}
        \caption{A first attempt at preventing leakage: sending the evaluation table of $\hat x_{|\line}$ in ``locked boxes'' and only unlocking the points checked by the verifier.}
        \label{fig:line-commitment}
     \end{subfigure}
     \caption{Leakage in the SIP for \indexproblem via evaluation of the bivariate polynomial $\hat x: \field^2 \to \field$, and an (unsuccessful) attempt to prevent it.}
\end{figure}
    
    The protocol outlined above is, however, \emph{not} zero-knowledge: after all, $V$ learns not only $x_j$, but the restriction of $\hat x$ to an entire line $\line$ \emph{through $j$} (see \cref{fig:pep-leakage}). Note that learning the restriction of $\hat x$ to (say) a random line $R$ does not necessarily constitute leakage: $V$ could simply compute a few evaluations (rather than only one) of $\hat x_{|R}$, which fully determine the polynomial. The issue is that $\line$ is a function of the coordinate $j$, which $V$ does not know prior to streaming $x$.
     
     In the next section we will take our first steps towards making the protocol zero-knowledge, i.e., ensuring that the verifier learns nothing beyond the value $x_j$.
     Note that the honest $V$ only evaluates $\hat x_{|\line}$ at two points, $\rfe1$ and $0$; what if $P$ could send the evaluations of $\hat x_{|\line}$ in ``locked boxes'' and only open the pair that the verifier needs?

\subsection{Curtailing leakage with commitments}
\label{sec-to:pv-commitment}

    To make the foregoing approach more precise, let us first assume the existence of a \emph{commitment protocol} that allows $P$ to transmit any field element $\fe1$ to $V$ in two steps: sending a string $\textsf{commit}(\fe1)$, from which $V$ is unable to extract any information about $\fe1$; and later, upon the verifier's request, revealing a field element $\fe2$ such that, if $\fe2 \neq \fe1$, then $V$ can detect that the $P$ is being dishonest.
    
    With such a commitment protocol in hand, a natural attempt to prevent the \pepprotocol protocol from leaking information is to have the prover $P$ send a commitment to $\hat x_{|\line}$, the restriction of the input's LDE to the line chosen by $V$ (rather than sending the polynomial in the clear). That is, the prover would commit to the evaluation table of $\hat x_{|\line}$, sending $\big(\textsf{commit}(\hat x_{|\line}(\rfe1')) : \rfe1' \in \field\big)$, after which $V$ can reveal its random evaluation point $\rfe1$ and $P$ decommits \emph{only} to the evaluations of $0$ and $\rfe1$  (see \cref{fig:line-commitment}). This does indeed reveal less information ($2$ rather than $\abs{\field}$ evaluations of $\hat x$), but is still far from what we set out for.
    
    There are two severe shortcomings with this idea; we shall tackle one now and defer the other to \cref{sec-to:vp-commitment}. First we need to ask: what is to prevent a cheating prover from committing to a function $g$ that is inconsistent with $\hat x_{|\line}$? 
    Indeed, since $V$ is (by design) unable to learn the field elements that were committed to, it cannot detect whether the function is a low-degree polynomial; then a cheating prover may commit to any $\fe1 \neq x_j = \hat x_{|\line}(0)$ as the claimed evaluation at $0$, while committing to the correct evaluations elsewhere. The resulting function is not a low-degree polynomial anymore, but $V$ is oblivious to this fact.

    Therefore, we require a scheme that allows not only to commit to a function, but to also ensure it is a low-degree polynomial. We solve this problem by constructing an \emph{algebraic} commitment protocol, whereby $P$ commits to a set of field elements and can decommit to \emph{any linear combination} of them. Then $P$ may commit to $\degree + 1$ points -- which uniquely determine a degree-$\degree$ polynomial $g$ -- and $V$ requests a decommitment to the linear combination that coincides with $g(\rfe1)$ (see \cref{fig:algebraic-commitment}). 
   We next present the basic commitment protocol, and then extend it to be algebraic.
    
    \begin{figure}[ht]
     \center
     \begin{subfigure}[t]{0.48\textwidth}
        \center
       \includegraphics[width=\textwidth]{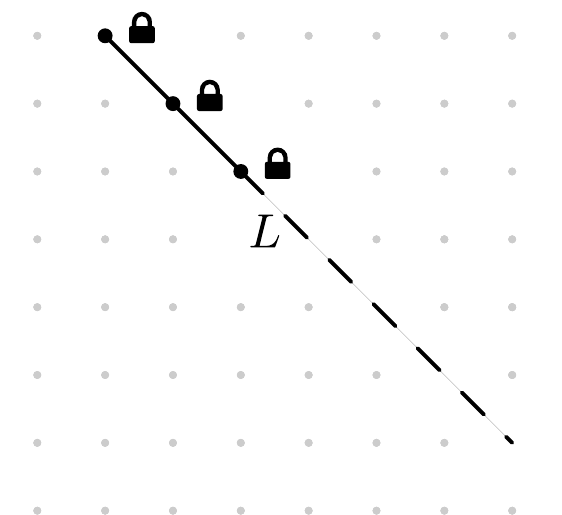}
      \caption{\small Commitments to an interpolating set of $\hat x_{|\line}$.}
     \end{subfigure}%
     \hfill
     \begin{subfigure}[t]{0.48\textwidth}
        \center
       \includegraphics[width=\textwidth]{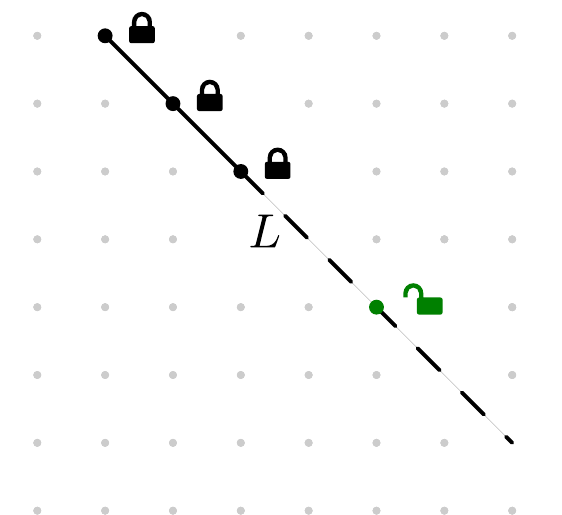}
      \caption{\small Decommiting to a point outside the interpolating set.}
     \end{subfigure}
     \caption{Preventing leakage by committing to $\hat x_{|\line}$ as an interpolating set for the polynomial. To decommit to an evaluation outside the set, the scheme must be algebraic.}
     \label{fig:algebraic-commitment}
    \end{figure}
    
    \paragraph*{The basic protocol\iflipics\else.\fi} Recall that our goal is to construct a commitment protocol between asymmetric parties, allowing a computationally unbounded $P$ to send and later reveal a message $\fe1 \in \field$ to a low-space verifier $V$. We focus on the first step, where $P$ sends a hidden message, and deal with how to reveal it later. A natural attempt is to play the prover's strength against the verifier's weakness: we know, from the hardness of \indexproblem, that the space limitation of $V$ prevents it from recalling an item from a long stream whose position is only revealed later; we can thus have $P$ send a long stream $\pcommitstring$ with the message hidden at a coordinate $k$ that is revealed at the end.
    
    While the idea seems intuitively sound, there are nontrivial issues to address. For example, the string-coordinate pair $(\pcommitstring,k)$ should not have any structure from which $V$ could extract information, which we can ensure by sampling both uniformly at random; but to prove security for this strategy, \indexproblem must be hard to solve \emph{on average}. Luckily, reductions from one-way communication complexity enable us to prove this fact: one-way protocols where Alice receives $x \sim \bitset^n$ and sends an $s$-bit message to Bob, who receives $j \sim [n]$ and attempts to output $x_j$, succeed with probability at most $\frac1{2} + O(\sqrt{s/n})$ \cite{RY20}. We show that the bound extends to larger alphabets, carrying over to space-$s$ streaming algorithms (see \cref{prop:index-hardness,lem:string-to-bit-index}).
    
    In short, we have $P$ encode its message $\fe1 \in \field$ \emph{as the solution to a random \indexproblem instance}, exploiting the problem's average-case hardness to ensure that $V$ is unable to extract $\fe1$; more precisely, $P$ sends a uniformly random string-coordinate pair $(\pcommitstring, k)$ and then the ``correction'' $\fe3 = \fe1 - \pcommitstring_k$.\footnote{We remark that while replacing $\pcommitstring_{ik}$ with $\fe1$ (rather than sending a random element and a correction later) looks simpler, then $(\pcommitstring,k)$ ceases to be a random \indexproblem instance, and it is not clear how to show a reduction from \indexproblem.} Of course, the discussion thus far only shows how $P$ can commit; but we also need a decommitment protocol whereby $V$ can check that $P$ is being honest when it reveals $\fe2$ (which may or may not coincide with the message $\fe1$). Fortunately, we already have a tool $V$ can use to solve \indexproblem with an untrusted prover's assistance! The decommitment thus consists of an execution of \pepprotocol by $P$ and $V$ with respect to the instance $(\pcommitstring,k)$: this allows $V$ to learn $\pcommitstring_k$ and check that $\fe3 + \pcommitstring_k = \fe2$, i.e., that the correction $\fe3$ sent earlier matches the (alleged) message.
    
    Recall that we are building technical tools towards a \zksip{} for \indexproblem, so we ultimately \emph{exploit the hardness of a problem to solve an instance of the same problem}. Should we not expect, then, that the same leakage issues should arise with respect to the ``virtual'' instance $(\pcommitstring, k)$ as they did with the ``real'' instance $(x,j)$? While this may appear to be circular reasoning, we stress that revealing evaluations of $\hat y$ leaks no information whatsoever about the input; indeed, $(\pcommitstring,k)$ is a uniform random variable that is independent of $(x,j)$. Put differently, $V$ only obtains information about uniformly random strings that are completely uncorrelated with the input. See \cref{sec:pv-commitment} for details.
    
    \paragraph*{Making the scheme algebraic\iflipics\else.\fi} We now extend the foregoing idea into an \emph{algebraic} protocol, which allows $P$ to commit to a \emph{tuple} of field elements $\bm{\fe1} = (\bm{\fe1}_1, \ldots, \bm{\fe1}_\ell)$ and decommit to a linear combination $\bm{\fe1} \cdot \bm{\fe2}$. (Committing to a polynomial and decommitting to an evaluation follows as a special case; see \cref{sec:lde-pep}.) Note that such an extension seems to follow if linear combinations ``commute'' with commitments; that is, by showing that linear combinations of a fingerprint (as defined in \cref{sec-to:pep}) match a fingerprint of the linear combinations, we should be able to use essentially the same strategy of the basic scheme: committing with a random \indexproblem instance and decommitting with \pepprotocol. Details follow.
    
    Consider a trivial extension of the scheme that allows $P$ to transmit a pair of messages $\fe1, \fe1' \in \field$: sending two independent commitments $(\pcommitstring, k, \fe1 - \pcommitstring_k)$ and $(\pcommitstring', k', \fe1' - \pcommitstring'_{k'})$. The key observation is that, if $V$ saves two fingerprints \emph{at the same evaluation point} $\fingerprinttuple$, then linear combinations and low-degree extensions do commute: for any $\fe2, \fe2' \in \field$, defining $z \coloneqq \fe2 \pcommitstring + \fe2' \pcommitstring'$, we have $\hat z(\fingerprinttuple) = \fe2 \hat \pcommitstring(\fingerprinttuple) + \fe2' \hat \pcommitstring'(\fingerprinttuple)$; in short, evaluating a low-degree extension is a linear operation.
    
    A problem still remains, however: since $k \neq k'$ with overwhelming probability, an execution of the \pepprotocol protocol enables $V$ to learn $z_k = \fe2 \pcommitstring_k + \fe2' \pcommitstring'_k$; but the correction for $\pcommitstring'$ refers to another coordinate $k' \neq k$ (with overwhelming probability). We address this issue by \emph{hiding both messages at the same index}, i.e., setting $k' = k$ and only revealing the coordinate after both $\pcommitstring$ and $\pcommitstring'$ are sent; see \cref{sec:pv-algebraic-commitment} for details.

\subsection{From honest to malicious verifiers: temporal commitments}
\label{sec-to:vp-commitment}

   Recall that a source of leakage in the \indexproblem protocol of \cref{sec-to:pep} is the prover $P$ sending the restriction of $\hat x$ (the LDE of the input) to a line $\line$ in the clear. In the previous section, we constructed a prover-to-verifier scheme that enables $P$ to commit to a low-degree polynomial and decommit to a single evaluation of it. We may then use it to modify the original protocol, having $P$ instead \emph{commit} to $\hat x_{|\line}$ and decommit to the points inspected by $V$.
   
   While this modification amounts to significant progress -- indeed, it achieves an \emph{honest-verifier} SIP for \indexproblem  -- there is a second major challenge to address. The issue is that \emph{if a verifier $\widetilde V$ cheats}, it can use the protocol to extract information that it could not have learned on its own, as we will see next. The goal of this section is to describe a strategy that prevents leakage of information \emph{without} requiring that $\widetilde V$ behave honestly; in other words, we would like to make the protocol \emph{malicious-verifier} zero-knowledge.
    
    Concretely, consider the (cheating) verifier $\widetilde V$ that ignores the input string $x$, reads $j$ and requests the line through $j$ and $j + 1$ from the prover. $P$ then commits to the restriction of $\hat x$ to this line and decommits to the evaluation of the LDE at both $j$ and $j + 1$. This reveals $x_j$ and $x_{j+1}$ to $\widetilde V$, which shows clearly that the modified protocol still leaks: $x_j$ is \emph{the only information the verifier should learn} that it could not have computed on its own, but the protocol also reveals $x_{j+1}$ (which is just as hard to compute as the $j^\text{th}$ coordinate).
    
    \paragraph*{An idealised scenario: $V$-to-$P$ commitments\iflipics\else.\fi} Let us assume, for the moment, that there also exists a commitment protocol in the reverse direction, allowing $V$ to commit and later reveal a message to $P$. We will show how, in this idealised setting, we can prevent information leakage altogether. Note that the difficulty posed by a malicious verifier $\widetilde V$ is the usage of an allegedly random evaluation point $\fingerprinttuple$ that is, in reality, a function of the input.
    
    If $\widetilde V$ \emph{proves that $\fingerprinttuple$ is indeed random}, however, we may conclude that $\widetilde V$ could have computed $\hat x(\fingerprinttuple)$ alone -- and thus that no leakage occurs. The idealised scheme allows $\widetilde V$ to do (almost) that, by having it commit to $\fingerprinttuple$ \emph{before reading the input stream} and decommit to it at a later step (after the prover's commitment). While this does not ensure $\fingerprinttuple$ is random, the fact that $\widetilde V$ cannot decommit to anything other than $\fingerprinttuple$ constrains its evaluation point to be \emph{chosen before the input stream}, so that it cannot be a function of the input.
    
    Of course, it is not at all clear that such a commitment protocol, allowing a weak computational party to commit to a computationally unbounded one, even exists; after all, the commitment step generally exploits their very difference to hide the message, as we did in the previous section. Is this just wishful thinking?
    
    \paragraph*{The solution: a temporal commitment\iflipics\else.\fi} We will now see that, perhaps surprisingly, we can once again exploit the space limitation of $\widetilde V$ to accomplish this goal. What we obtain in fact falls short of a full-fledged commitment protocol: roughly speaking, the \emph{temporal  commitment} will enable a space-$s$ verifier $\widetilde V$ to reveal not one, but $s$ messages. But this collection is still determined before the input, so that it remains fit for purpose (incurring a small overhead in the simulator algorithm that we discuss in the following section).
    
    As discussed above, we cannot expect $\widetilde V$ to be able to send a hidden message to $P$: however $\widetilde V$ may try to hide it, $P$ can simply store the entirety of the communication and extract the message itself. Since sending is out of the picture, could $\widetilde V$ instead commit by \emph{receiving} a message? Note that, while somewhat counter-intuitive, this would allow $\widetilde V$ to play what is essentially its only strength, its private randomness, against $P$. Recall, moreover, that there is a temporal aspect to the positions of a long stream $\vcommitstring$ that $\widetilde V$ can remember: if it remembers $\vcommitstring_i$, this can be seen as evidence that \emph{$i$ was determined no later than when $\vcommitstring$ was seen}.
    
    Let us now make the idea more precise, and construct our verifier-to-prover temporal commitment protocol. The main idea is to impose some cost onto the ability of $\widetilde V$ to ``unlock'' the decommitment from $P$, without overly constraining the honest verifier $V$. Note that after $P$ sends the commitment to a low-degree polynomial, having $V$ reveal the point $\fingerprinttuple = \line(\rfe1)$ at which it computed $\hat x$ is not a problem (as opposed to revealing $\fingerprinttuple$ \emph{before} $P$ sends the polynomial, which allows the prover to cheat easily). Therefore, we will have $\widetilde V$ reveal its alleged evaluation point $\fingerprinttuple$ \emph{along with a certificate} $c(\fingerprinttuple)$ that shows $\widetilde V$ selected the point before seeing the input stream. $P$ will only proceed with the protocol if the certificate is valid; if not, it aborts to prevent $\widetilde V$ from learning information beyond its reach.
    
    Given that the verifier's scarce resource is space, we design this certificate to require a number of bits that is not too large and yet not negligible; then the honest $V$ should have no trouble, as it only needs to remember one piece of information, whereas the malicious $\widetilde V$ described before would need to store a certificate for the evaluation point $j + 1$, which it does not know before reading $x$. 
    
    We thus prepend our \indexproblem protocol with a step where $P$ sends $\widetilde V$ a long string $\vcommitstring$ containing all possible evaluation points (i.e., the entire domain) of the low-degree extension $\hat x$.\footnote{In fact, any given point has a small probability of being absent from the string. We ignore this issue in the technical overview.} Now, if $\widetilde V$ wants the prover, in the future, to decommit to a polynomial evaluation at the point $\fingerprinttuple$, it must offer evidence that $\fingerprinttuple$ is uncorrelated with the input stream: $\widetilde V$ does so by revealing $\fingerprinttuple$ \emph{along with the coordinate $i$ that contains $\fingerprinttuple$ in $\vcommitstring$}; i.e., the certificate for $\fingerprinttuple$ is $c(\fingerprinttuple) = i$, the coordinate satisfying $\vcommitstring_i = \fingerprinttuple$.
    
    \iflipics\else\paragraph*{}\fi The temporal commitment indeed achieves what we set out for: regardless of what $\widetilde V$ does, as long as its space is bounded we are able to extract the points it may ask $P$ for \emph{in advance of its streaming of $x$} (see \cref{sec:vp-commitment}). Note that the commitment is non-interactive (consisting of a single message from $P$ to $V$) and need not be rerun if the verifier streams multiple inputs; we shall use it as the setup stage of our protocol. Its analysis is subtle and involved: it begins with a study of a variant of \indexproblem in the one-way communication model that we call \reconstruct, where, upon receipt of a message from Alice, Bob outputs a guess for every coordinate of the input string rather than for only one. Using tools from information theory, we obtain an upper bound on the expected number of correct coordinates, which we call the protocol's \emph{score}.
    
    Next, we use the expected score bound of \reconstruct to prove a related upper bound for a problem we call \pair: a variant of \indexproblem where Bob, rather than receiving the coordinate to be recovered as part of the input, is free to choose it. The implication is that any protocol for \pair has a small number $C$ of indices such that the output of the protocol is outside $C$ and yet correct (i.e., a pair $(i, \vcommitstring_i)$ with $i \notin C$) with arbitrarily small probability. This will underpin the \emph{simulator argument} that ultimately shows our protocol is zero-knowledge, which we sketch in the next section.

\subsection{A sketch of the zero-knowledge {\indexproblem} protocol}
\label{sec-to:index-protocol}

    We now have all of the components necessary to sketch a zero-knowledge streaming interactive proof for $\indexproblem$. Recall that we constructed a prover-to-verifier \emph{algebraic} commitment protocol in \cref{sec-to:pv-commitment} and a verifier-to-prover \emph{temporal} commitment in \cref{sec-to:vp-commitment}. We will now compose them in the appropriate order, using the temporal commitment to constrain $V$ to choose its inner randomness before reading the input stream; and the algebraic commitment to ensure $P$ only reveals what the verifier needs.
    The protocol follows.
    
    \paragraph*{Parameters\iflipics\else.\fi}
    Without loss of generality, we consider the alphabet over which the input string is defined to be a field of size $\abs{\field} = \fieldsize$; that is, $x \in \field^n$. We also fix two additional parameters, $\degree$ and $\dimension$, which characterise the low-degree extension $\hat x: \field^\dimension \to \field$ as an $\dimension$-variate polynomial of individual degree $\degree$. We assume all parameters are known to $P$ and $V$ in advance.
    
    \paragraph*{Setup: verifier-to-prover temporal commitment\iflipics\else.\fi} $P$ sends $V$ a permutation of $\field^\dimension$ as a string $\vcommitstring$ (of length $\vcommitlength = \fieldsize^\dimension$). Before receiving the string, $V$ samples $\fingerprinttuple \sim \field^\dimension$ and then streams $\vcommitstring$. When it sees $\fingerprinttuple$ at the $\ell^\text{th}$ coordinate of $z$, the verifier stores $\ell$.

    \paragraph*{Step 1: input streaming\iflipics\else.\fi} $V$ streams the input string $x$ and records the fingerprint $\hat x(\fingerprinttuple)$ as well as the target index $j$.

    \paragraph*{Step 2: prover-to-verifier algebraic commitment\iflipics\else.\fi} $V$ samples $\rfe1 \sim \field$ and sends $P$ the line $\line: \field \to \field^\dimension$ through $j$ and $\fingerprinttuple$ (satisfying $\line(0) = j$ and $\line(\rfe1) = \fingerprinttuple$).
    
    $P$ sends $x_j = \hat x_{|\line}(0)$ (in the clear) and an algebraic commitment $(y, \bm{\fe3}, k)$ to the remainder of an interpolating set of the degree-$\degree\dimension$ polynomial $\hat x_{|\line}: \field \to \field$, i.e., to the field elements $\hat x_{|\line}(i)$ for all $i \in [\degree\dimension]$. The commitment consists of a random matrix $\pcommitstring \sim \field^{\degree\dimension \times \pcommitlength}$ with $\degree\dimension$ rows and a large enough number $\pcommitlength$ of columns; a random (column) coordinate $k \sim [\pcommitlength]$; and the correction tuple $\bm{\fe3}$ satisfying $\bm{\fe3}_i = \hat x_{|\line}(i) - \pcommitstring_{ik}$.
    
    $V$ samples (another) evaluation point $\Fingerprinttuple$ and computes the fingerprint $\pcommitstring(\Fingerprinttuple, \evaluationpoint) = \sum_i \evaluationpoint_i \hat \pcommitstring_i(\Fingerprinttuple)$, where the tuple $\evaluationpoint$ satisfies $\sum_i \evaluationpoint_i \hat x_{|\line}(i) = \hat x(\fingerprinttuple)$;\footnote{Note that $\evaluationpoint_i$ is determined solely by $i$ and $\rfe1$: it is the evaluation $\chi_i(\rfe1)$ of the $i^\text{th}$ Lagrange polynomial.} it also computes $\fe3 = \sum_i \evaluationpoint_i \bm{\fe3}_i$ and stores $k$.

    \paragraph*{Step 3: temporal decommitment\iflipics\else.\fi} $V$ reveals its fingerprint's evaluation point $\fingerprinttuple$ along with the index $\ell$ where it appeared in $\vcommitstring$. The prover checks that $\vcommitstring_\ell = \fingerprinttuple$, and only continues to the final step if the check passes.

    \paragraph*{Step 4: algebraic decommitment\iflipics\else.\fi} $P$ and $V$ engage in the decommitment of the $k^\text{th}$ coordinate of the string $\pcommitstring' = \bm{\fe2} \cdot \pcommitstring$ (the linear combination of the rows $\pcommitstring_i$ with coefficients $\evaluationpoint_i$).\footnote{This requires $P$ to know the linear coefficients $\bm{\fe2}$, and, while we could have the verifier send them, this is not necessary: $P$ learns $\fingerprinttuple$ in step 3, which allows it to determine $\rfe1 = \line^{-1}(\fingerprinttuple)$ and thus $\bm{\fe2} = \bm{\fe2}(\rfe1)$ as well.} $V$ outputs the (alleged) $x_j$ if the decommitment is consistent with $\hat x(\fingerprinttuple)$, and rejects otherwise.

    \iflipics\else\paragraph*{}\fi In an honest execution of the above protocol, the final decommitment reveals
    \begin{align*}
        \pcommitstring'_k &= \sum_i \evaluationpoint_i y_{ik}\\
        &= \sum_i \evaluationpoint_i \big(\hat x_{|\line}(i) - \bm{\fe3}_i\big)\\
        &= \hat x(\fingerprinttuple) - \fe3,
    \end{align*}
    so that $V$, having stored $\hat x(\fingerprinttuple)$ and $\fe3$, can indeed perform this consistency check (which shows the protocol is complete). The protocol's soundness follows from that of \pepprotocol, noting that none of the mechanisms we add harm soundness (indeed, the last check relies, as does \pepprotocol, on a random evaluation of the low-degree extension), while zero-knowledge, which we discuss next, follows from the correctness of our commitment protocols.    
    
    \paragraph*{Proving the zero knowledge property\iflipics\else.\fi}
    We conclude with a discussion of the simulator argument for the protocol laid out in this section. Recall that proving zero-knowledge for the foregoing protocol entails the construction of a \emph{simulator} $S$, a streaming algorithm with knowledge of $x_j$ and roughly the same memory as $\widetilde V$, which is able to interact with $\widetilde V$ without it being able to tell whether it is communicating with $S$ or $P$.
    
    Roughly speaking, $S$ does the following: after the temporal commitment step, it inspects the memory state of $\widetilde V$ and records (almost) all the points to which $\widetilde V$ can decommit; as shown in the last section, this is a relatively small set $C$. It then streams the input and records $\hat x(\fingerprinttuple)$ \emph{for all $\fingerprinttuple \in C$}.\footnote{We note that storing $C$ is the most space-intensive task of $S$, which implies a small overhead to its space complexity as compared to $\widetilde V$; see \cref{thm:pep-zk}.} Upon receipt of a line $\line$ from $\widetilde V$, the simulator computes and commits to an arbitrary low-degree polynomial $g$ that interpolates across the points in $\line \cap C$. When $\widetilde V$ requests the algebraic decommitment to obtain an evaluation of $g$, the simulator checks that the evaluation point $\rfe1$ is contained in $C$ (in which case $g(\rfe1)$ matches a fingerprint $\hat x(\fingerprinttuple)$ known to $S$), proceeds with the decommitment if that is the case, and otherwise aborts.
    
    We note that implementing the strategy above raises yet another challenge, namely, extracting the set $C$ of evaluation points from the description and memory state of $\widetilde V$. This is accomplished via a form of \emph{white-box access} to $\widetilde V$, see \cref{sec:streaming-zk}.
    
     The simulator $S$ is thus able to generate the transcript of an interaction where the message $\hat x_{|\line}$ of the algebraic commitment is replaced with another low-degree polynomial $g$ whose evaluations match $\hat x_{|\line}$ \emph{at all points where $\widetilde V$ is able to temporally decommit}. Then, distinguishing between a real and a simulated transcript amounts to distinguishing an \indexproblem instance whose solution is $\hat x_{|\line}$ from one whose solution is $g$.
     
     We prove that any streaming algorithm that does so with nontrivial bias implies a one-way communication protocol for \indexproblem with a small message, contradicting the known hardness of the problem. We remark that the reduction is rather nontrivial, as we must insert an \indexproblem instance into the algebraic commitment $(\pcommitstring, \bm{\fe3}, k)$ while ensuring the decommitment can be simulated without any knowledge about the instance. See \cref{thm:pep-zk} for details.
     
     \begin{remark}[Superpolynomial to near-linear communication]
         We stress that, while we may prove zero-knowledge with the strategy above, the natural reduction from \indexproblem is over a large alphabet $\alphabet = \field^{\degree\dimension}$. But then, for indistinguishability to follow, the length $\pcommitlength$ of the temporal commitment must be $\fieldsize^{\degree\dimension}$, which implies \emph{superpolynomial} communication complexity.
         
         We avoid this blowup via \cref{lem:string-to-bit-index}, which shows that an \indexproblem (one-way) protocol for large alphabets implies another protocol for the binary alphabet with only a mild loss to its success probability; this restricts our ambient field to be an extension of $\field_2$, but reduces the superpolynomial complexity to \emph{barely superlinear}.
     \end{remark}

\subsection{A general-purpose zero-knowledge SIP: sumcheck}
\label{sec-to:sumcheck}

    Lastly, we briefly mention how the commitment protocols developed in \cref{sec-to:pv-commitment,sec-to:vp-commitment} can be used not only to solve \indexproblem (and, more generally, the polynomial evaluation problem), but also to construct another widely applicable tool: a streaming zero-knowledge \emph{sumcheck} protocol.
    
    As before, we start with an SIP that is clearly not zero-knowledge: the standard sumcheck protocol leaks hard-to-compute sums over subcubes. By carefully using the algebraic and temporal commitment protocols, we can also endow the sumcheck protocol with zero-knowledge in the data stream model. However, we note that doing so is considerably more involved than in the case of \indexproblem, owing to, among other reasons, several rounds of interaction with nontrivial dependencies of messages on past communication.
    
    More precisely, we consider a slight variation of the standard sumcheck protocol: while in the latter every round is followed by a (random) consistency check, we instead defer all such checks to the end. It is clear that this variant is equivalent to the standard protocol; however, without the modification, the zero-knowledge property seems to require a strengthening of the chained commit-decommit strategy we follow. Moreover, rather than a single algebraic commitment followed by a (single) decommitment, the sumcheck protocol requires many decommitments; indeed, for an $\dimension$-variate polynomial $f$, the prover commits to $\dimension$ partial sums of $f$, and each partial sum is involved in two decommitments (for a total of $\dimension + 1$ decommitments).
    
    Therefore, by extending the techniques that underpin our approach for the \indexproblem problem to a \emph{multi-round} setting, we are able to construct a zero-knowledge sumcheck SIP. Such a protocol can then be used to  compute frequency moments and inner products, problems known to require linear space without a prover's assistance \cite{AMS99}. See \cref{sec:sumcheck} for details.

\section{Preliminaries}
\label{sec:preliminaries}

    \paragraph*{General notation\iflipics\else.\fi} For an integer $k \geq1$, we denote by $[k]$ the set $\set{1,2, \ldots, k}$. Vectors are denoted with notation analogous to that of sets, i.e., $(\fe1_i : i \in [k])$ denotes the vector $(\fe1_1, \ldots, \fe1_k)$. We use $n$ to denote the length of a string that is the input to an algorithm, and $\poly(n)$ (respectively, $\polylog(n)$) to denote an arbitrary polynomial (respectively, polylogarithmic) function in $n$.
    
    We use lowercase Latin letters to denote positive integers (e.g., $\degree, i, j, k, \ell, \dimension, \pcommitlength, \vcommitlength$) or strings (e.g., $x, \pcommitstring, \vcommitstring$); $r$ and $t$ often (but not always) denote random strings. Lowercase Greek letters denote elements of a finite alphabet or field (e.g., $\fe1, \fe2, \fe3$), and we reserve $\rfe1, \rfe2$ for random elements. 
    Uppercase letters denote either algorithms (e.g., $A, B, P, S, V$) or sets (e.g., $C, K$), with $T$ used as the indeterminate of a polynomial.
    
    When $f$ and $g$ are functions, we sometimes use $\fe1 \in f$ as a shorthand for $\fe1 \in \Im f$ and $f_{|g}$ for $f \circ g$; if $f$ is a low-degree polynomial that is communicated in an interactive protocol, we assume it is sent in a canonical form (e.g., a line is communicated by a pair of points $f(0), f(1)$). We use $\chf[\cdot = x_0]$ to denote the delta function at $x_0$ (i.e., $\chf[x_0 = x_0] = 1$ and $\chf[x = x_0] = 0$ for $x \neq x_0$) and $\log$ to denote $\log_2$.
    
    As integrality issues do not substantially change any of our results, equality between an integer and an expression (that may not necessarily evaluate to one) is assumed to be rounded to the nearest integer.
    
    \paragraph*{Vectors and matrices\iflipics\else.\fi} The notation we use for matrices is the same as for strings (lowercase Latin letters), and it will be clear from context which is the case. When $x$ is a matrix, we use $x_i$ to refer to the $i^\text{th}$ row of $x$.
    
    We use vectors or tuples, interchangeably, to refer to elements of a vector space over a finite field $\field$. Such tuples are denoted with boldface (e.g., $\bm{\fe1}, \evaluationpoint, \bm{\fe3}$) and random tuples are (similarly to strings) denoted $\fingerprinttuple, \Fingerprinttuple$. We use $\bm{\fe1} \cdot \bm{\fe2}$ to denote the inner product between the two vectors, and, when the dimension of $\bm{\fe1}$ matches the number of rows of a matrix $x$, we use $\bm{\fe1} \cdot x$ to denote the vector corresponding to the linear combination of the rows of $x$ with coefficients $\bm{\fe1}$, i.e., $\sum_i \bm{\fe1}_i x_i$. (Equivalently, we assume vectors to be in row form.)

    \paragraph*{Probability\iflipics\else.\fi} We use $X \sim \mu$ to denote a random variable with distribution $\mu$, and, for the uniform distribution over a set $S$, we write $X \sim S$. We sometimes make the sources of randomness in a probabilistic expression explicit, and when we do they are assumed to be independent; e.g., only when $X$ and $Y$ are independent do we write $\P_{X \sim \mu, Y \sim \lambda}[E]$. The internal randomness of an algorithm is generally omitted; e.g., $\P[A(X) = 0]$ (if the distribution of $X$ is known from context) or $\P_{X \sim \mu}[A(X) = 0]$ are shorthand for $\P_{\subalign{X &\sim \mu\\r &\sim \bitset^m}}[A(X; r) = 0]$, where $r$ is $A$'s internal randomness.
    
    We will also make use of the following versions of the Chernoff and Hoeffding bounds.
    \begin{lemma}[Additive Chernoff-Hoeffding bound]
    \label{lem:additive-chernoff}
      Let $X_1, \ldots, X_k$ be independent Bernoulli random variables distributed as $X$. Then, for every $\delta \in [0,1]$,
      \begin{align*}
        \Pr\left[\frac{1}{k}\sum_{i = 1}^k X_i \leq \E[X] -\delta \right] &\leq e^{- 2\delta^2 k} \text{ and}\\
        \Pr\left[\frac{1}{k}\sum_{i = 1}^k X_i \geq \E[X] +\delta \right] &\leq e^{- 2\delta^2 k}.
      \end{align*}
    \end{lemma}
    \begin{lemma}[Hoeffding's inequality]
    \label{lem:hoeffding}
      Let $X_1, \ldots, X_k$ be independent random variables distributed as $X \in [a,b]$. Then, for every $\delta \in [0,1]$,
      \begin{align*}
        \Pr\left[\frac{1}{k}\sum_{i = 1}^k X_i \leq (1 - \delta) \E[X]\right] &\leq e^{-\left(\frac{\delta \E[X]}{b - a}\right)^2 k} \text{ and}\\
        \Pr\left[\frac{1}{k}\sum_{i = 1}^k X_i \geq (1 + \delta) \E[X]\right] &\leq e^{-\left(\frac{\delta \E[X]}{b - a}\right)^2 k}.
      \end{align*}
    \end{lemma}

    \paragraph*{Algorithms and protocols\iflipics\else.\fi} We use the same term to refer to computational problems and to protocols that solve them, but distinguish the two cases with different fonts (so that the \pepprotocol and \textsf{sumcheck} protocols solve the \pep and \sumcheck problems, respectively).
    
    We generally use $A$, $D$, $S$ and $V$ to denote streaming algorithms, while $P$ denotes an algorithm with unbounded computational resources (including space). $A(x)$ is the output of an algorithm that receives $x$ as input; when $A$ is a streaming algorithm, $x$ is read sequentially in one pass, from the first symbol ($x_1$) to the last. When $A(x, y, z)$ reads multiple inputs, $A(y)$ denotes the partial execution of $A$ after it has read $x$. When the entries of a length-$n$ string $x$ are taken over a finite alphabet $\alphabet$, we may also use $x$ for the equivalent bit string of length $n \log\abs{\alphabet}$.
    
    We shall often make use of the \emph{minimax principle}, and assume, without loss of generality, that a computationally unbounded algorithm $A$ whose goal is to maximise some value $\E_{x \sim \mu}[f(A(x))]$ (e.g., the probability that $A(x)$ equals $x$) can be assumed to be deterministic, and thus given by a function $x \mapsto a(x)$; equivalently, $A$ can be taken as the deterministic algorithm that maximises $\E[f \circ a(x)]$ for the distribution of inputs $\mu$.
    
    In a protocol, two algorithms $P$ and $V$ interact by exchanging messages in a predefined order; after all messages have been exchanged, $V$ chooses an output that we denote $\langle P, V\rangle$ and call the output of the protocol. When $V$ rejects or $P$ aborts midway through the interaction, we assume the algorithm proceeds until the end of the protocol with dummy messages (e.g., strings of zeroes).
    
    The \emph{snapshot} of an algorithm is synonymous to its memory state; when $A$ reads a sequence of more than one input, e.g., $A(x,y)$, the ``snapshot of $A$ after $x$'' is the snapshot immediately before the first symbol of $y$ is streamed (i.e., after $A$ has read and processed the last symbol of $x$). When $A$ is interacting in a protocol and sends a message between reading $x$ and $y$, the snapshot after $x$ is that immediately before sending the message.

    \paragraph*{Low-degree extensions\iflipics\else.\fi} For any field $\field$ and integer $k$ such that $\abs{\field} \geq k$, we consider $[k] \subseteq \field$ via a canonical injection (e.g., taking the image of $\ell \in [k]$ as the field element whose binary representation is the same as that of $\ell$). Accordingly, we write $\ell \in \field$ as shorthand for the field element corresponding to the image of $\ell \in [k]$ via this canonical injection.
    
    For a string $y \in \field^k$, the \emph{low-degree extension} (LDE) with \emph{degree} $\degree$ and \emph{dimension} $\dimension$ where $\abs{\field} \geq \degree + 1$ and $k \leq (\degree + 1)^\dimension$, denoted $\hat y$, is the unique $\dimension$-variate polynomial of individual degree $\degree$ that coincides with $y$ in $[k]$; more precisely, viewing $[k] \subseteq [\degree + 1]^\dimension \subseteq \field^\dimension$, the LDE $\hat y: \field^\dimension \to \field$ is the unique polynomial satisfying $\hat y(i) = y_i$ for all $i \in [k]$. Our notation for the polynomial $\hat y$ omits the degree and dimension, as they will be clear from context.
    
    When $y$ is a matrix, we use $\hat y(\bm{\fe1}, \bm{\fe2})$ to denote the linear combination of the LDEs of the rows with linear coefficients $\bm{\fe2}$, i.e., $\hat y(\bm{\fe1}, \bm{\fe2}) = \sum_i \bm{\fe2}_i \hat y_i(\bm{\fe1})$.

    \subsection{Information theory} We will make use of several notions of information theory and approximations of information-theoretic quantities. The \emph{$q$-ary entropy function} $H_q \colon [0,1] \to [0,1]$ is
    \begin{align}
    \label{eq:q-ary-entropy}
        H_q(t) &= t \log_q(q-1) -t \log_q t - (1 - t) \log_q(1 - t)\\
        &= \frac{1}{\log q} \big(t \log(q - 1) - t \log t - (1 - t) \log(1 - t)\big)\nonumber\\
        &= \frac{1}{\log q} \big(t \log(q - 1) + H_2(t)\big),\nonumber
    \end{align}
    where $H_q(0) = 0$; we also use the shorthand $H$ for $H_2$, which simplifies to
    \begin{equation}
    \label{eq:binary-entropy}
        H(t) = H(1 - t) = -t \log t - (1 - t) \log(1 - t).
    \end{equation}

    We will make use of the following approximation for the (natural) logarithm function: for $0 \leq t \leq 1/2$,
    \begin{equation}
    \label{eq:log-approximation}
        -t (1 + t) \leq \ln(1-t) \leq -t.
    \end{equation}
    
    The (relative) \emph{Hamming distance} between two strings $a, b \in \alphabet^k$ over a finite alphabet is the fraction of coordinates where they differ, i.e., $d(a,b) = \frac1k \abs{\set{i \in [k] : a_i \neq b_i}} \in [0,1]$. With $\alphabetsize = \abs{\alphabet}$, the volume of a \emph{Hamming ball} $\ball(b,\delta) \coloneqq \set{a \in \alphabet^k : d(a, b) \leq \delta}$ of radius $\delta = 1 - \eps$, when $k$ is large enough and $\eps = k^{-1} \polylog(k)$, satisfies\iflipics\footnote{The lower bound is a simplification of
    \begin{equation*}
        \abs{\ball(b,\delta)} \geq \alphabetsize^{H_\alphabetsize(\delta) k} \cdot \exp\left(\frac1{12k + 1} - \frac1{12\delta k} - \frac1{12 \eps k}\right)\big/\sqrt{2\pi \delta(1 - \delta) k};
    \end{equation*}
    since $1/\eps k = o(1)$, the numerator is $1 - o(1)$, and the denominator is of order $\Theta(\sqrt{\eps k}) =\polylog(k)$. (See, e.g., \cite{GRS12}.)}\fi
    \begin{equation}
    \label{eq:hamming-volume-approximation}
        \alphabetsize^{H_\alphabetsize(\delta) k} \geq \abs{\ball(b, \delta)} = \Omega \left(\frac{\alphabetsize^{H_\alphabetsize(\delta) k}}{\sqrt{\eps k}}\right) = \frac{\alphabetsize^{H_\alphabetsize(\delta) k}}{\polylog(k)}\enspace.\iflipics\else\footnote{The lower bound is a simplification of
    \begin{equation*}
        \abs{\ball(b,\delta)} \geq \alphabetsize^{H_\alphabetsize(\delta) k} \cdot \frac{\exp\left(\frac1{12k + 1} - \frac1{12\delta k} - \frac1{12 \eps k}\right)}{\sqrt{2\pi \delta(1 - \delta) k}};
    \end{equation*}
    since $\frac1{\eps k} = o(1)$, the numerator is $1 - o(1)$, and the denominator is of order $\Theta(\sqrt{\eps k}) =\polylog(k)$. (See, e.g., \cite{GRS12}.)}\fi
    \end{equation}

    The entropy of a discrete random variable $X$ taking values in $\alphabet$ is
    \begin{equation*}
      H(X) = - \sum_{\symbol1 \in \alphabet} \P[X = \symbol1] \log\big(\P[X = \symbol1]\big).
    \end{equation*}
    Every such random variable satisfies
\begin{equation}
\label{eq:entropy-bounds}
    H(X) \in [0,\log\abs{\alphabet}].
\end{equation}

    The conditional entropy $H(X|Y)$ is the entropy of the conditional random variable, which satisfies
    \begin{equation}
    \label{eq:entropy-conditioning}
        H(X|Y) \leq H(X).
    \end{equation}
    If $X,Y$ are independent, then
    \begin{equation}
    \label{eq:entropy-independence}
        H(X,Y) = H(X) + H(Y).
    \end{equation}
    The last property of entropy we will make use is the \emph{chain rule}: for random variables $X_1, \ldots, X_n$,
    \begin{equation}
        \label{eq:chain-rule-entropy}
        H(X_1, \ldots, X_n) = \sum_{i = 1}^n H(X_i|X_1, \ldots, X_{i-1}).
    \end{equation}

    For ease of notation, when $(X, Y)$ are jointly distributed over $\alphabet^2$ with marginals $\mu$ and $\lambda$, respectively, we denote the distribution of $Y$ conditioned on $X = x$ as $\lambda_x$. The \emph{KL divergence} between the distributions is
    \begin{equation}
        \label{eq:kl-divergence}
        \mathrm{KL}(\mu ~||~ \lambda) = \sum_{\symbol1 \in \alphabet} \mu(x) \log \frac{\mu(\symbol1)}{\lambda(\symbol1)},
    \end{equation}
    which upper bounds the Euclidean distance between probability vectors via \emph{Pinsker's inequality} (see, e.g., \cite{BLM13}):
    \begin{equation}
    \label{eq:pinsker}
    \norm{\mu - \lambda}^2 \leq \frac{\mathrm{KL}\big(\mu ~||~ \lambda\big)}{2 \ln 2}.
\end{equation}
    
    Finally, the \emph{mutual information} is defined as (and equivalent to)
    \begin{align}
        \label{eq:mutual-information}
        I(\mu:\lambda) \coloneqq&~ I(X:Y)\\
        =&~ I(Y:X) \nonumber\\
        =&~ H(Y) - H(Y|X) \nonumber\\
        =&~ \E_{X \sim \mu}[\mathrm{KL}(\lambda_X ~||~ \lambda)] \nonumber.
    \end{align}

\section{Zero-knowledge streaming interactive proofs}
\label{sec:streaming-zk}

    This section motivates and provides a definition of zero-knowledge proofs in the data stream model. We start by discussing the differences between the streaming and the traditional settings as well as establish necessary notation. We then we provide a formal definition in \cref{sec:def}.
    
    The notion of zero-knowledge proofs in a computational model should capture the intuition that, when engaged in an interactive protocol, a verifier algorithm $V$ should learn nothing but the truth of some hard-to-compute statement about its input $x$ (e.g., that $x$ is in a language $L$). For consistency with the general notion we define zero-knowledge for \emph{decision problems} in the streaming model, but remark that the definition extends to search problems in the standard way (i.e., the verifier $V$ learns nothing but a valid solution to the search problem).
    
    In the traditional setting, $V$ can easily store the entirety of $x$ and make polynomial-time computations without the assistance of a prover. This implies that the sensitive information a zero-knowledge proof in this setting must not leak is the result of a computation on $x$ beyond the verifier's reach, i.e., one that requires superpolynomial time to obtain from the information available to $V$. In the streaming setting, however, the notion of ``hard-to-compute'' changes dramatically: the model puts \emph{space} as the primary resource, so that computations within the reach of $V$ are those possible with a small amount of space and sequential one-pass access to the input (but arbitrarily large time complexity). Knowledge then essentially corresponds to all information that $V$ cannot compute in low space complexity using its streaming access. As a result, zero-knowledge streaming interactive proofs (\zksip{}s) must satisfy a much more stringent requirement: that they not leak any information \emph{about the input $x$ itself} (which in the traditional setting is fully known to the verifier).
    
    In order to capture such a stringent notion of sensitive information, we define \zksip{}s as protocols such that no \emph{streaming} algorithm can distinguish a real transcript of the protocol from one that is generated by a (streaming) simulator. To this end, we first recall the formalisation of \emph{streaming interactive proofs} (SIPs) \cite{CTY11} without any zero-knowledge requirement.

\begin{definition}
    \label{def:sip}
    A \emph{streaming interactive proof} (SIP) for a language $L$ is an interactive proof defined by a pair $(P,V)$ of algorithms: a computationally unbounded prover $P$ and streaming verifier $V$ with space $s = o(n)$. The verifier engages in an iteractive protocol with $P$ and streams, at a predetermined step, the bit string $x \in \bitset^n$, which $P$ also observes.\footnote{The definition could allow for alternating between streaming parts of $x$ and communicating with the prover, as well as adaptively choosing the round(s) on which to read the input. Our protocols do not require this flexibility, however, so we assume the entirety of $x$ is read at a fixed step along the communication protocol.} At the end of the protocol, $V$ outputs a binary decision $\langle P,V \rangle(x)$ satisfying
    \begin{itemize}
        \item (completeness) if $x \in L$, then $\P[\langle P,V \rangle(x) = 1] \geq 2/3$; and
        \item (soundness) if $x \notin L$, then then $\P[\langle P,V \rangle(x) = 1] \leq 1/3$.
    \end{itemize}
\end{definition}

We call $s$ the \emph{space complexity} (of the verifier).
Note that, while the constant $1/3$ is arbitrary, soundness amplification does not hold for streaming algorithms due to the need to reread the input; nevertheless, many SIPs (including all those considered in this paper) allow for improving soundness by a desired factor with a logarithmic increase to their space complexity (see \cref{sec:lde-pep}). We stress that \cref{def:sip} constrains the verifier \emph{only} in terms of space, which allows arbitrarily large time complexities for both prover and verifier. (This is similar to other settings such as communication complexity and property testing, where the primary resources are communication and queries, respectively.)

Loosely speaking, we capture the notion of zero-knowledge in the data stream model by saying that an SIP is zero-knowledge if there exists a streaming \emph{simulator algorithm} $S$, with roughly the same space as the verifier $V$, able to simulate a prover-verifier interaction that is indistinguishable from a real one; that is, $S$ generates a \emph{view} of the verifier (defined next) that no \emph{distinguisher} algorithm with power comparable to $V$ (i.e., a streaming algorithm with roughly the same space) can tell apart from a real interaction. We stress that while the distinguisher $D$ is reminiscent of computational zero-knowledge, the security of our protocols is information-theoretic and \emph{does not rely on computational assumptions}.

\begin{definition}\label{def:streaming-view}
    Let $(P,V)$ be an SIP with a space-$s$ verifier, where $P$ sends $k_1$ messages to $V$ before the verifier streams its input, and an additional $k_2$ messages afterwards. 
    Denote the prover's messages by $y_1 \in \bitset^{p_1},\ldots,y_{k_1 + k_2} \in \bitset^{p_{k_1 + k_2}}$; the input by $x$; and the verifier's and prover's internal randomness by $r$ and $t$, respectively.
    
    The \emph{view} of the verifier $\widetilde V$, denoted $\view_{P,\widetilde V}(x,r)$, is the random variable defined as
	\begin{equation*}
	    \view_{P,\widetilde V}(x,r ; t) = (r, y_1, \ldots, y_{k_1}, x, y_{k_1 + 1}, \ldots, y_{k_1 + k_2}).\footnote{We note that a more general definition allows the random bits $r$ to be partially streamed throughout the protocol, rather than only in the beginning. This simpler definition suffices to capture the honest $V$ in all of our protocols, but we assume the more general version when (a malicious) $\widetilde V$ consumes more randomness than it can store.}
	\end{equation*}
\end{definition}

While \cref{def:streaming-view} is similar to its polynomial-time analogue, we highlight an important distinction: to faithfully correspond to what $\widetilde V$ sees, the order in which the view is streamed must be preserved. Indeed, a step-by-step execution of $\widetilde V$ in an interaction with $P$ corresponds exactly to its streaming $\view_{P,\widetilde V}(x,r)$ one symbol at a time. Order preservation is also consistent with the input stream $x$ being observed by all parties simultaneously (which are, in a simulation, $\widetilde V$, the simulator $S$ and a distinguisher $D$).

\subsection{Definition}
\label{sec:def}
We now ready to give a formal definition of zero-knowledge streaming interactive proofs.

\begin{definition}[zkSIP]
\label{def:streaming-zk}
    Let $L$ be a language and $(P,V,S)$ be a triplet where $(P,V)$ is an SIP with a space-$s$ verifier $V$ and $S$ is a streaming $\poly(s)$-space simulator with \emph{white-box access} to the verifier, streaming access to the input $x$ and additional query access to a random bit string $t$.
    
    $(P,V,S)$ forms a \emph{zero-knowledge streaming interactive proof} (\zksip{}) for $L$ that is secure against space-$s'$ adversaries if, for any space-$s$ algorithm $\widetilde V$ and $x\in L$, the random variables $\view_{P,\widetilde V}(x, r)$ and $S(\widetilde V, x, r)$ are indistinguishable by any streaming space-$s'$ algorithm. That is, for every space-$s'$ streaming algorithm $D$,
    \begin{equation*}
        \abs{\P\left[D\big(\view_{P,\widetilde V}(x, r)\big) \text{ accepts}\right] - \P\left[D\big(S(\widetilde V, x, r)\big) \text{ accepts}\right]} = o(1).
    \end{equation*}
\end{definition}

We note that all our applications have $s = \polylog(n)$, and the protocols are secure against adversaries with any space $s' = \poly(s)$ (see \cref{rem:indistinguishability-pep}).

\begin{remark}
\label{rem:subconstant-bias}
    Recall that the analogue of \cref{def:streaming-zk} in the polynomial-time setting requires a much stronger notion of indistinguishability: \emph{negligible} (i.e., sub-inverse-polynomial), rather than $o(1)$, bias. This is necessary for the notion to be robust with respect to poly-time algorithms, as otherwise repeating polynomially many executions of $D$ would boost its success probability arbitrarily close to $1$.
    
    This raises a number of interesting questions on the achievable notions of security for \zksip{}s: can we obtain tighter bounds, such as $1/\poly(n)$ or negligible? (Perhaps even in the statistical case?) An answer to each such question ensures security against one type of adversary (i.e., distinguisher): we will study the natural threat model where all parties are streaming algorithms and argue why $o(1)$ is a sufficient bound in this case. Before doing so, however, we briefly discuss an important alternative.
    
    As explained above, streaming verifiers secure against polynomial-time adversaries require negligible distinguishability. This has been previously studied, most notably for zero-knowledge interactive proofs that reduce to evaluating low-degree polynomials defined by the input and allow for it to be processed in a streaming fashion, such as \cite{GKR08}. (We stress, however, that such protocols rely on computational assumptions.) An interesting question that we leave to future work is whether \zksip{}s can \emph{simultaneously} achieve security against different adversaries -- e.g., with negligible bias for poly-time distinguishers (under cryptographic assumptions) in addition to subconstant bias for streaming distinguishers.
    
    Recall that a key distinction between the poly-time and streaming settings is the \emph{one-pass} restriction of the latter, which prevents even a single repetition of (a streaming) $D$ -- indeed, \indexproblem trivialises with $2$ passes (as do many fundamental streaming problems). In other words, as the common technique of \emph{amplification} is unavailable in the streaming model, $o(1)$ bias is a sufficiently robust requirement that guarantees the probability of information leakage tends to $0$. (We note that the weaker requirement of arbitrarily small constant bias would also suffice, i.e., the existence of $(P_\eps, V_\eps, S_\eps)$ achieving $\eps$ bias for every $\eps > 0$. We adopt the simpler and stronger subconstant version, which our protocols satisfy.)
    
\end{remark}

\paragraph*{The streaming simulator\iflipics\else.\fi}
For technical reasons, the simulator is given white-box access to the verifier and explicit access to a random string. We stress that this auxiliary information is completely independent of the input. This can viewed as allowing the verifier to obtain some computation about auxiliary information (about its own strategy, or a uniformly chosen random string), but learn absolutely \emph{zero information about the input stream $x$}. 

While white-box access gives the simulator $S$ knowledge of any function of the verifier's strategy, we do not require such generality; indeed, we will only be interested in questions about the most likely messages that $\widetilde V$ may send at a single point of the protocol. As such, the weaker definition that follows is sufficient.

\begin{definition}
    \label{def:whitebox-max}
    Let $A$ be a space-$s$ streaming algorithm that reads an $n$-bit string $y$ and outputs an $m$-bit string $z$. We define \emph{white-box access} to $A$ as oracle access to a function $\whiteboxoracle$ with two inputs, a snapshot $\snapshot \in \bitset^s$ and a candidate output $z \in \bitset^m$; the oracle returns the maximum probability \emph{over all inputs $y$} with which $A$, starting with memory state $\snapshot$, outputs $z$; that is,
    \begin{equation*}
        \whiteboxoracle(\snapshot,z) = \max_{y \in \bitset^n}\set{\P[A(y) \text{ outputs } z \text{ when its initial snapshot is } \snapshot]}.
    \end{equation*}
\end{definition}

\begin{remark}
    \label{rem:streaming-vs-memory-bottleneck}
    While the honest verifier $V$ does not use a large random string, malicious verifiers $\widetilde V$ with this additional resource can readily be simulated by $S$ as above. We assume hereafter that $\widetilde V$ has the same resources as the honest verifier, but note that the simulations extend straightforwardly to verifiers with both white-box access (to their strategies) and query access to a random string.    
\end{remark}

\section{Algebraic and temporal commitments}
\label{sec:commitment-protocols}

A commitment protocol is a two-party protocol (or, more accurately, a pair of protocols) that allows the transmission of a message from one party to another to be split into two parts: a \emph{commitment}, where the message is transmitted in a form that cannot be interpreted by the recipient; followed, at some point in the future, by a \emph{decommitment}, where the sender transmits additional information with which the recipient can read the message. (A useful analogy is that the commitment amounts to sending a locked box containing the message, and the decommitment to sending the key.)

In the standard setting \cite{B83} we have two parties: a sender and a receiver, which we will refer to as prover and verifier, respectively. The prover wishes to communicate a symbol $\symbol1$, and does so by first choosing a random \emph{key} $k$ and sending another string $c=\textsf{commit}(\symbol1,k)$. Then, at some point in the future, prover and verifier engage in a protocol at the end of which the receiver obtains $\fe1 = \textsf{decommit}(c)$. (We will refer to the streaming analogue as a commitment \emph{protocol}, rather than scheme, to avoid ambiguity with the polynomial-time analogue.)

Commitment protocols are extremely useful components for the construction of interactive protocols, and should satisfy two properties: \emph{hiding}, i.e., the commitment alone should prevent the verifier from obtaining a non-negligible amount of information about the message $\symbol1$; and \emph{binding}, i.e., the prover should not be able to decommit to a message that differs from the one it committed to. 
We will construct a commitment protocol whose hiding property follows from the average-case hardness of \searchindex for streaming algorithms, while binding follows from the soundness of the \pepprotocol protocol (which we introduce formally in \cref{sec:lde-pep}).

We first formally define streaming commitment protocols. We note that while the definition that follows can be generalised,\footnote{A natural generalisation is to parameterise the bias in the hiding property as well as the completeness and soundness in binding by $\eps_b, \eps_c, \eps_s \in (0,1)$; our definition has $\eps_b, \eps_s = o(1)$ and $\eps_c = 0$.} it suffices to capture our constructions.

\begin{definition}
\label{def:commitment}
A \emph{streaming commitment protocol} for alphabet $\alphabet$ (with security parameter $\pcommitlength$) and space bound $s$ consists of a function $\textsf{commit}: \alphabet \times K \to C$, where $K \subseteq \bitset^\pcommitlength$ is the set of keys and $C$ is the set of commitments, and a space-$s$ SIP $(P,V)$ which satisfy the following conditions.

\begin{itemize}
    \item \emph{Hiding:} Fix any pair of distinct messages $\symbol1, \symbol2 \in\alphabet$ and sample $k \sim K$. Set $c = \textsf{commit}(\symbol1) = \textsf{commit}(\symbol1, k)$ and $c' = \textsf{commit}(\symbol2) = \textsf{commit}(\symbol2, k)$. Every (streaming) space-$s$ distinguisher $D$ tells the two commitments apart with at most subconstant bias (with respect to the parameter $\pcommitlength$); that is,

    \begin{equation*}
        \abs{\P[D(c) \text{ accepts}] - \P[D(c') \text{ accepts}]} = o(1).
    \end{equation*}

    \item \emph{Binding:} Fix $k \in K$ and $\alpha \in \alphabet$. Then
    \begin{equation*}
        \P\big[\langle P, V\rangle\big(\textsf{commit}(\alpha, k), \alpha\big) = 1\big] = 1,
    \end{equation*}
    and for any $\beta \neq \alpha$,
    \begin{equation*}
        \P\big[\langle P, V\rangle\big(\textsf{commit}(\alpha, k), \beta\big) = 1\big] = o(1).
    \end{equation*}
\end{itemize}
\end{definition}

Note that, with some abuse of notation, the binding condition corresponds to $(P,V)$ being an SIP for the language $L = \set{(\textsf{commit}(\alpha, k), \alpha) : \alpha \in \alphabet, k \in K}$.

The next sections introduce the commitment protocols we will use to build our protocols. \cref{sec:lde-pep} begins by defining the concepts and tools we build upon: low-degree extensions and the polynomial evaluation protocol (\pepprotocol). In \cref{sec:pv-commitment}, we use them to construct a basic scheme that allows for the communication of a single symbol (which we use as a stepping stone), based on the hardness of \indexproblem (or, more accurately, \searchindex); in it, the keys are simply long strings paired with a coordinate, i.e., $K = \alphabet^\pcommitlength \times [\pcommitlength]$, and commitments are keys appended with a single extra symbol (i.e., $C \subset \alphabet^{\pcommitlength + 1} \times [\pcommitlength]$).

\cref{sec:pv-algebraic-commitment} then extends the construction of \cref{sec:pv-commitment} into an \emph{algebraic} commitment protocol, which allows for the commitment of low-degree polynomials. In both the basic and algebraic schemes, hiding is achieved by overwhelming $V$ with ``too much information'', and can only be broken if a malicious verifier is lucky enough to retain a critical fragment of the information stream; indeed, as we will see, breaking it amounts to solving \indexproblem. Binding, on the other hand, relies on the \pepprotocol protocol, which we introduce in the next section.

While commitment protocols are not a prerequisite for a zero-knowledge protocol, they also serve as inspiration for our second main component: \cref{sec:vp-commitment} shows how the verifier can perform a \emph{temporal commitment} to show its alleged internal randomness is uncorrelated with its input, and thus that it is not behaving maliciously.

\subsection{Low-degree extensions and polynomial evaluation}
\label{sec:lde-pep}

Fingerprinting is a technique that enables streaming algorithms to approximately verify an arbitrary coordinate of a long string in small space. It exploits \emph{low-degree extensions} (LDEs), extremely useful objects in the design of interactive proofs more broadly.

Given a data set $x$, viewed as a string of $n$ elements in a finite field $\field = \field_q$, an LDE is a low-degree polynomial that interpolates every data point. More precisely, we may view $x$ as a function $x: [n] \to \field$; given a \emph{dimension} $\dimension$ and defining the \emph{degree}  $\degree$ as the smallest (positive) integer such that $n \leq (\degree + 1)^\dimension$, we can also view $x: [\degree + 1]^\dimension \to \field$ by some canonical injection $[n] \hookrightarrow [\degree + 1]^\dimension$ (padding with zeroes if $n < (\degree + 1)^\dimension$). Then, as long as $\fieldsize > \degree$, we can also view (via another canonical injection $[\degree + 1] \hookrightarrow \field$) the data set as the restriction of a function from $\field^\dimension$ to $\field$.

Standard properties of polynomials imply that if this function is an $\dimension$-variate polynomial of individual degree $\degree$, then the extension is unique; we thus denote by $\hat x: \field^\dimension \to \field$ the unique degree-$\degree$ polynomial whose restriction to $[n]$ is equal to $x$. Explicitly, with $(i_1, \ldots, i_\dimension)$ as the image of $i$ by $[n] \hookrightarrow \field^\dimension$,
\begin{equation*}
    \hat x = \sum_{i = 1}^n x_i\chi_i = \sum_{i_1, \ldots, i_\dimension \in [\degree + 1]} x_{i_1, \ldots, i_\dimension} \chi_{i_1, \ldots, i_\dimension}
\end{equation*}
where the $\chi_i$ are the Lagrange basis polynomials, given by
\begin{equation*}
    \chi_i(\fe1_1, \ldots, \fe1_\dimension) \coloneqq \prod_{j = 1}^\dimension \prod_{\subalign{k & =1\\ k &\neq i_j}}^{\degree + 1} \frac{\fe1_j-k}{i_j-k}
\end{equation*}
(viewing $k \in [\degree + 1]$ as an element of $\field$); equivalently, the Lagrange polynomials are the unique $\dimension$-variate degree-$\degree$ polynomials satisfying $\chi_i(j) = \chf[i = j]$ when $i, j \in [\degree + 1]$. We note that LDEs and Lagrange polynomials can equivalently be defined with an injection from $\{0\} \cup [\degree]$, rather than $[\degree + 1]$, to $\field$; then they satisfy the previous condition for all $0 \leq i, j \leq \degree$. We will use the characterisation that is most convenient, which will be clear from context (e.g., an LDE that involves the evaluation of a polynomial at $0$ is of the latter type).

We will also use $\bm{\chi}(\bm{\fe1})$ to denote the vector $\big(\chi_1(\bm{\fe1}), \ldots, \chi_n(\bm{\fe1})\big)$ of evaluations of Lagrange polynomials; note that this allows us to write $\hat x(\bm{\fe1})$ as the dot product $\bm{\chi}(\bm{\fe1}) \cdot x$ of $n$-dimensional vectors.

Now, given a string $x\in\field^n$, a \emph{fingerprint} is simply an evaluation of the LDE of $x$ at a random point, that is, $\hat x(\fingerprinttuple)$ with $\fingerprinttuple \sim \field^\dimension$. The key property of fingerprints is that they are extremely unlikely to match for two different strings when the underlying field is large enough, as a consequence of the Schwartz-Zippel lemma \cite{S80,R81}.
\begin{lemma}[Schwartz-Zippel]
	\label{lem:sz}
	If $x,y\in\field_\fieldsize^n$ are distinct, then $\P_{\fingerprinttuple \sim \field^\dimension}\big[\hat x(\fingerprinttuple) = \hat y(\fingerprinttuple)\big] \leq \degree \dimension / \fieldsize$.
\end{lemma}

Importantly for streaming algorithms, fingerprints can be computed with $O(\degree \dimension)$ time per entry of the input and $O(\dimension)$ field elements (thus $O(\dimension \log \fieldsize)$ bits) of space \cite{CTY11}.

The polynomial evaluation protocol is an interactive proof that enables a streaming verifier with a single random evaluation $f(\fingerprinttuple)$ of a degree-$\degree$ polynomial $f:\field^\dimension\to\field$ to evaluate $f$ at any other point, assisted by a prover with knowledge of $f$ in its entirety. Note that the prover could help the verifier compute $f$ at a point (non-interactively) by simply sending an interpolating set of the polynomial; but any such set has size $(\degree + 1)^\dimension$. The \pepprotocol (polynomial evaluation) protocol, detailed in \cref{prot:pep}, allows us to reduce the communication from $O(\degree^\dimension\log \fieldsize)$ to $O(\degree \dimension \log \fieldsize)$ by adding interaction.

In order to better compare the original \pepprotocol protocol with the zero-knowledge version that we will construct, we consider a general problem that the protocol is able to solve (as in \cite{CCMTV19}). We use $f$ as shorthand for a mapping $x \mapsto f^x$ (or, equivalently, a set $f \subseteq \set{f^x : x \in \field^n}$) where one evaluation $f^x(\fingerprinttuple)$ can be computed by a space-bounded algorithm that streams $x$. The problem $\pep(f, \fe1)$ is to decide whether $f^x(\evaluationpoint) = \fe1$ when the input stream is $x$ followed by an evaluation point $\evaluationpoint \in \field^\dimension$.

\iflipics
\begin{figure}
\else
\begin{protocol}[label={prot:pep}]{$\pepprotocol(f, \fe1)$}
\fi
    \textbf{Input:} Explicit access to $\fe1 \in \field$ and a set $f \subseteq \set{f^x : x \in \field^n}$ of $\dimension$-variate degree-$\degree$ polynomials over $\field$. Streaming access to $(x, \evaluationpoint) \in \field^n \times \field^\dimension$.

    \begin{enumerate}[label={}]
        \item[$\bm V$:] Sample $\fingerprinttuple \sim \field^\dimension$. Stream $x$ and compute $f^x(\fingerprinttuple)$. Store $\evaluationpoint$.
        
        Compute the line $\line:\field\to\field^\dimension$ such that $\line(0)=\evaluationpoint$ and $\line(\rfe1)=\fingerprinttuple$ with $\rfe1\sim\field$, then send $\line$ to the prover.

        \item[$\bm P$:] Compute and send $f^x_{|\line}$.\footnotemark

        \item[$\bm V$:] Compute $g(\rfe1)$, where $g: \field \to \field$ is the degree-$\degree\dimension$ low-degree extension of the sequence of evaluations sent by $P$ such that $g(0) = \fe1$.\footnotemark Accept if $g(\rfe1) = f^x(\fingerprinttuple)$ and reject otherwise.
    \end{enumerate}
\iflipics
\caption{Protocol $\pepprotocol(f, \fe1)$}
\label{prot:pep}
\end{figure}
\else
\end{protocol}
\fi
\addtocounter{footnote}{-2}
\addtocounter{Hfootnote}{-2}

\envfootnote\footnotetext{Recall that the line $\line$ and $f^x_{|\line}$ are sent in a canonical form: $\line$ as the evaluation $\line(1)$ and $f^x_{|\line}$ as the vector $\big(f^x \circ \line(i) : i \in [\degree \dimension]\big)$. (There is no need to send $\line(0) = \evaluationpoint$ or $f^x_{|\line}(0) = f^x(\evaluationpoint) = \fe1$, as they are known to $V$.)}
\envfootnote\footnotetext{Note that the Lagrange polynomials in this case satisfy $\chi_i(j) = \chf[i = j]$ for all $0 \leq i, j \leq \degree\dimension$.}

Assuming an evaluation of $f^x$ can be computed by streaming $x$ with $O(\dimension \log \fieldsize)$ space, \cref{prot:pep} is a streaming interactive proof for $\pep(f, \fe1)$ with communication complexity $O(\degree \dimension \log \fieldsize)$ and verifier space complexity $O(\dimension \log \fieldsize)$. We note that $\pepprotocol(f,\fe1)$ can easily be modified into an algorithm for a search problem without a candidate value $\fe1$ for $f^x(\evaluationpoint)$, by having $V$ output $g(0)$ instead of accepting. 

It is clear that $V$ accepts in \cref{prot:pep} when $P$ is honest; the protocol's soundness relies on the fact that if the prover were to send an incorrect $g \neq f^x_{|\line}$, it is highly unlikely that it will agree with the verifier's evaluation at the (unknown) location $\fingerprinttuple$.

In conjuction with the streaming nature of LDEs, (the search version of) \cref{prot:pep} yields a simple and efficient streaming interactive proof for \searchindex. This SIP, introduced by \cite{CCMTV19}, has $O(\log n \log\log n)$ space and communication complexities for a stream $(x, j) \in \field^n \times [n]$ where $\fieldsize = \abs{\field} = \polylog(n)$ (and $\evaluationpoint \in \field^\dimension$ is the identification of $j$); it is simply an instantiation of \pepprotocol where $\degree = 2$, $\dimension = \log n$ and the function $f^x = \hat x$ is the $\dimension$-variate (multilinear) LDE of $x$,\footnote{The space complexity can be reduced to $O(\log n)$ with the choice of parameters for $\fieldsize$, $\degree$ and $\dimension$ in \cref{cor:index}.} an evaluation $\hat x(\fingerprinttuple)$ of which can be computed incrementally as values of $x$ are revealed in the stream. Then $\hat x(\fingerprinttuple) = \hat x_{|\line}(\rfe1)$ allows the verifier to check that the prover is being honest (i.e., that the polynomial it sent is $\hat x_{|\line}$), as well as to learn $x_j = \hat x(j) = \hat x_{|\line}(0)$.

Observe that $\pepprotocol$ is \emph{not} zero knowledge: the verifier learns all of $f^x_{|\line}$, which it is not be able to construct by virtue of only learning $\evaluationpoint$ (and thus $\line$) \emph{after} streaming $x$. Note, however, that the \emph{honest} verifier only inspects two evaluations of $f^x_{|\line}$, namely, at $0$ and $\rfe1$. In the following sections we construct a commitment protocol that lets the prover only reveal information about these two points, without sacrificing soundness.

\subsection{A prover-to-verifier commitment protocol}
\label{sec:pv-commitment}

Our commitment protocol, designed to allow an unbounded-space sender to commit to a streaming receiver, directly uses the (average-case) hardness of the \indexproblem problem. By sending a message hidden at a random coordinate, we exploit the fact that any streaming algorithm requires a linear amount of space to be able to recall a random item from a string after it has been seen. We begin by formally defining (the search and decision versions of) \indexproblem \emph{in the one-way communication complexity model}.

\begin{definition}
    \label{def:search-index}
    \searchindex, over alphabet $\alphabet$ and with message length $s$, is the one-way communication problem defined as follows: Alice receives a string $x \in \alphabet^n$ and sends Bob an $s$-bit message $a = A(x)$. Bob receives, besides $a \in \bitset^s$, an index $j \in [n]$, and outputs a symbol $b = B(a, j) \in \alphabet$. The execution succeeds if $b = x_j$.
\end{definition}

\begin{definition}
    \label{def:decision-index}
    $\decisionindex(\fe1)$ (with alphabet $\alphabet$ and message length $s$) is the one-way communication problem defined as follows: Alice receives a string $x \in \alphabet^n$ and sends Bob an $s$-bit message $a = A(x)$. Bob receives, besides Alice's message, an index $j \in [n]$, and outputs a bit $b = B(a, j) \in \bitset$. The execution succeeds if $b = 1$ when $x_j = \fe1$, and $b = 0$ otherwise.
\end{definition}

It is well known that \indexproblem is extremely hard, even \emph{on average} and in the one-way communication model \emph{with shared randomness}. 

\begin{proposition}
\label{prop:index-hardness}
    Any one-way communication protocol $(A, B)$ for \searchindex that sends a message of length $s$ satisfies 
    \begin{equation*}
        \P_{\substack{x \sim \alphabet^\pcommitlength\\ j \sim [\pcommitlength]}}\left[B\big(A(x), j\big) = x_j\right] = \frac1{\abs{\alphabet}}+O\left(\sqrt{\frac{s}{\pcommitlength}}\right).
    \end{equation*}
\end{proposition}

In other words, the chance of correctly recalling a random symbol is at best slightly better than uniform guessing if the string $\pcommitlength$ is much longer than the message length $s$ of the protocol. We note that this bound was known for $\alphabet = \bitset$ \cite{RY20}, but it extends to larger alphabets (we provide a proof of this fact in \cref{sec:deferred-index-hardness} for completeness).

The commitment phase of our scheme exploits this hardness result directly: we take $\alphabet \hookrightarrow \field$ where $\field$ is a large enough finite field (which will allow us to use \pepprotocol to decommit) and have $P$ send the triple $(\pcommitstring, \fe1 - \pcommitstring_k, k)$ for random $\pcommitstring$ and $k$ as a commitment to $\fe1$. (In particular, the commitment key is a random string-coordinate pair $(\pcommitstring, k)$). Loosely speaking, the protocol has the sender communicate a random stream $\pcommitstring$ with the message hidden at a random coordinate $k$, which is revealed after $\pcommitstring$.

The honest verifier keeps a (random) fingerprint of $\pcommitstring$, which it can use to validate the message at $\pcommitstring_k$ (see \cref{prot:commit}), while the \textsf{decommit} stage simply instantiates \pepprotocol appropriately (see \cref{prot:decommit}). We note that the inputs listed in the description of the protocols are those available to the verifier.

\iflipics
\begin{figure}
\else
\begin{protocol}[label={prot:commit}]{$\textsf{commit}(\fe1)$}
\fi
	\textbf{Input:} explicit access to $\pcommitlength, \degree, \dimension, \fieldsize \in \N$ with $\pcommitlength \leq \degree^\dimension$, $\fieldsize > \degree$ and $\field = \field_\fieldsize$. Streaming access to $\pcommitstring\sim\field^\pcommitlength$ followed by a correction $\fe3 \in \field$ and a coordinate $k \sim [\pcommitlength]$.
	\begin{enumerate}
	    \item[$\bm V$:] Sample $\fingerprinttuple \sim \field^\dimension$ and compute $\hat \pcommitstring(\fingerprinttuple) = \sum_{i = 1}^\pcommitlength \chi_i(\fingerprinttuple) \pcommitstring_i$ while streaming $\pcommitstring$.
	    
	    Store $\fingerprinttuple, k, \fe3$ and $\hat \pcommitstring(\fingerprinttuple)$.
	\end{enumerate}
\iflipics
\caption{Protocol $\textsf{commit}(\fe1)$}
\label{prot:commit}
\end{figure}
\else
\end{protocol}
\fi

\iflipics
\begin{figure}
\else
\begin{protocol}[label={prot:decommit}]{$\textsf{decommit}(\fe1, \pcommitstring, k)$}
\fi
    \textbf{Input:} $\fe1 \in \field$, as well as the (parameters and) values stored in the \textsf{commit} stage: $k, \fe3, \fingerprinttuple, \hat \pcommitstring(\fingerprinttuple)$.
    \begin{enumerate}[label={}]      
        \item[$\bm V$:] Compute and send the line $\line:\field\to\field^\dimension$ such that $\line(0) = k$ and $\line(\rfe1)=\fingerprinttuple$ with $\rfe1\sim\field$.

        \item[$\bm P$:] Send $\hat y_{|\line}$.

        \item[$\bm V$:] Compute $g(\rfe1)$ and $g(0)$, where $g: \field \to \field$ is the degree-$\degree\dimension$  extension of the sequence of evaluations sent by $P$.
        
        Accept if $g(\rfe1) = \hat \pcommitstring(\fingerprinttuple)$ and $g(0) + \fe3 = \fe1$, rejecting otherwise.
    \end{enumerate}
\iflipics
\caption{Protocol $\textsf{decommit}(\fe1, \pcommitstring, k)$}
\label{prot:decommit}
\end{figure}
\else
\end{protocol}
\fi

Now, we show that \cref{prot:commit,prot:decommit} form a streaming commitment protocol, i.e., they satisfy the hiding and binding properties of \cref{def:commitment} if $\pcommitlength$ is large enough; these follow from the hardness of \searchindex and the soundness of \pepprotocol, respectively.
\begin{theorem}
    \label{thm:pv-commitment}
    \cref{prot:commit,prot:decommit} form a streaming commitment protocol with space complexity $s = O(\dimension \log \fieldsize)$ when $\pcommitlength = \fieldsize^3$ and $\degree \dimension = \polylog(\fieldsize)$. The protocol is secure against $\poly(s)$-space adversaries and communicates $O(\fieldsize^3 \log \fieldsize)$ bits.
\end{theorem}

\begin{proof}
    First, note that the communication complexity is dominated by the prover sending $\pcommitlength = \fieldsize^3$ field elements in the \textsf{commit} step, for a total of $O(\fieldsize^3 \log \fieldsize)$ bits.

    The binding property is an immediate consequence of the completeness and soundness of \pepprotocol: if $P$ is honest, i.e., sends the correction $\fe3 = \fe1 - \pcommitstring_k$ in the \textsf{commit} stage and the polynomial $\hat y_{|\line}$ in the \textsf{decommit} stage, then $V$ accepts, as $\hat \pcommitstring_{|\line}(\rfe1) = \hat \pcommitstring(\fingerprinttuple)$ and $\hat \pcommitstring_{|\line}(0) + \fe3 = \fe1$. (Recall that the line $\line$ satisfies $\line(0) = k$ and $\line(\rfe1) = \fingerprinttuple$.)
    
    Now, suppose the prover replies with a polynomial $g$ such that $g(0) \neq \pcommitstring_k = \hat \pcommitstring(k) = \hat \pcommitstring_{|\line}(0)$; then the Schwartz-Zippel lemma (\cref{lem:sz}) implies $\hat \pcommitstring(\fingerprinttuple) = \hat \pcommitstring_{|\line}(\rfe1) \neq g(\rfe1)$ except with probability $\degree \dimension / \fieldsize = o(1)$, in which case $V$ rejects.\footnote{We remark that $\rfe1$ need not be sampled from the entire field; the same result holds if $\rfe1 \sim R \subset \field$ when $R$ is large enough. This will be useful in proving that our protocols for \pep and \sumcheck are zero-knowledge.} Note that the verifier only needs to store $\fingerprinttuple \in \field^\dimension$, $k \in [\pcommitlength]$ and a constant number of additional field elements, for a space complexity of $O(\dimension \log \fieldsize + \log \pcommitlength) = O(\dimension \log \fieldsize)$.

    To show the hiding property, assume towards contradiction that there exists a streaming algorithm $D$ with space $\poly(s) = \polylog(\fieldsize)$ that distinguishes commitments between some $\fe1 \in \field$ and $\fe1' \in \field \setminus \set{\fe1}$ with constant bias:\footnote{Note that allowing $\poly(s)$ space for $D$ will imply a space-robust indistinguishability property; bounding it by, say, $\tilde O(s)$ or $O(s^2)$ would prove a weaker but still nontrivial statement.} that is,
    \begin{align*}
        \P_{\subalign{\pcommitstring &\sim \field^\pcommitlength\\k &\sim [\pcommitlength]}}\big[D(\pcommitstring, k, \fe1 - \pcommitstring_k) \text{ accepts}\big] -  \P_{\subalign{\pcommitstring &\sim \field^\pcommitlength\\k &\sim [\pcommitlength]}}\big[D(\pcommitstring, k, \fe1' - \pcommitstring_k) \text{ accepts}\big] \geq \eps
    \end{align*}
    for some $\eps = \Omega(1)$. Now consider the following algorithm $A$ for \searchindex over the alphabet $\field$ with input $(x,j)$: simulate $D$ on the stream $(x,\fe3,j)$ where $\fe3 \sim \field$; output $\fe1 - \fe3$ if $D$ accepts, and otherwise output $\fe1' - \fe3$. Note that $A$ outputs correctly exactly when $\fe3 = \fe1 - \pcommitstring_k$ and $D$ accepts, or $\fe3 = \fe1' - \pcommitstring_k$ and $D$ rejects; moreover, $A$ can simulate $D$ with constant space overhead, so that its space complexity is also $\polylog(\fieldsize)$. We will now show that $A$ solves \searchindex with a bias that is too large, contradicting \cref{prop:index-hardness}.

    \begin{align*}
        \P_{\subalign{x &\sim \field^\pcommitlength\\j &\sim [\pcommitlength]}}\big[A(x,j) = x_j\big] &= \frac1{\fieldsize} \cdot \P_{\subalign{x &\sim \field^\pcommitlength\\j &\sim [\pcommitlength]}}\big[D(x,j,\fe1 - x_j) \text{ accepts}\big] + \frac1{\fieldsize} \cdot \P_{\subalign{x &\sim \field^\pcommitlength\\j &\sim [\pcommitlength]}}\big[D(x,j,\fe1' - x_j) \text{ rejects}\big]\\
        &= \frac1{\fieldsize} \left(1 + \P_{\subalign{x &\sim \field^\pcommitlength\\j &\sim [\pcommitlength]}}\big[D(x,j,\fe1 - x_j) \text{ accepts}\big] - \P_{\subalign{x &\sim \field^\pcommitlength\\j &\sim [\pcommitlength]}}\big[D(x,j,\fe1' - x_j) \text{ accepts}\big] \right)\\
        &\geq \frac{1 + \eps}{\fieldsize}\\
        &= \frac1{\fieldsize} + \Omega\left(\frac1{\fieldsize}\right).
    \end{align*}

    Since $1/\fieldsize = \sqrt{\fieldsize/\pcommitlength} = \omega\left(\sqrt{\poly(s)/\pcommitlength}\right)$, owing to $s = \polylog(\fieldsize)$, the result follows.
\end{proof}

\begin{remark}
    Just as in \pepprotocol, the verifier learns much more than than the message $\hat \pcommitstring_{|\line}(0) = \fe1 \in \field$: it learns all of $\hat \pcommitstring_{|\line}$. Crucially, however, the additional information consists of \emph{random field elements uncorrelated with $\fe1$}. This enables the commitment protocol laid out in this section to be proven zero-knowledge when the simulator has read-only access to a large random string $t$, as in \cref{def:streaming-zk}. (More accurately, such a simulator can perfectly generate the random variable that corresponds to the view resulting from the \textsf{commit} followed by the \textsf{decommit} steps.)
    
    Indeed, a simulator with space $O(\dimension \log \fieldsize)$ and query access to $\pcommitstring \sim \field^\pcommitlength$ may sample $k \sim [\pcommitlength]$ and send $(\pcommitstring, \fe1 - \pcommitstring_k, k)$ in the \textsf{commit} step; then, in \textsf{decommit}, after receiving the line $\line$, it computes and sends $\hat y_{|\line} = \big(\hat \pcommitstring_{|\line}(i) : i \in \set{0} \cup [\degree\dimension]\big)$ by reading the string $\pcommitstring$ an additional $\degree \dimension + 1$ times, computing and sending one LDE evaluation at a time.
\end{remark}

However, this basic commitment protocol is not yet sufficient. As discussed in \cref{sec-to:pv-commitment}, it allows $P$ to commit (and decommit) to a single field element; but the prover should be able to commit to a polynomial and decommit to a single evaluation thereof. In the next section we show how to accomplish this, by modifying our scheme to make it \emph{algebraic}.

\subsection{Making the commitment algebraic}
\label{sec:pv-algebraic-commitment}

In this section, we will show how to modify the commitment protocol laid out in \cref{sec:pv-commitment} so that the prover can commit to $\ell$ messages and decommit to \emph{a single linear combination} of the verifier's choosing. As we shall see, this can in fact be accomplished by adapting only the commitment step.

The idea behind this new protocol is simple, but has an important caveat. If the prover $P$ wishes to commit to the messages $\bm{\fe1} = (\bm{\fe1}_1, \bm{\fe1}_2, \ldots, \bm{\fe1}_\ell)$, the obvious solution is to send $(\pcommitstring_i, \bm{\fe1}_i - \pcommitstring_{ik_i}, k_i)$ for all $i$, a sequence of commitments to each $\bm{\fe1}_i$. However, the indices $k_i$ where each message is hidden are sampled independently, so that even though taking low-degree extensions is a linear operation (i.e., the LDE of $\sum_i \evaluationpoint_i \pcommitstring_i$ is $\sum \evaluationpoint_i \hat \pcommitstring_i$), a linear combination of the $\pcommitstring_i$ does not yield a commitment to a linear combination of the $\bm{\fe1}_i$: evaluating it at $k_i$ yields a sum where only the $i^\text{th}$ summand is guaranteed to be correct.

We can fix this problem by hiding all the messages at the same coordinate $k$. Then, setting $\bm{\fe3} = \big(\bm{\fe1}_i - \pcommitstring_{ik} : i \in [\ell]\big)$ and $\fe3 = \bm{\fe2} \cdot \bm{\fe3} = \sum \bm{\fe2}_i \bm{\fe3}_i$, we have
\begin{equation*}
    \fe3 + \big(\sum \evaluationpoint_i \pcommitstring_i\big)_k = \sum \evaluationpoint_i (\pcommitstring_{ik} + \bm{\fe3}_i) = \bm{\fe1} \cdot \evaluationpoint;
\end{equation*}
so a linear combination of commitments yields a commitment to a linear combination of the messages. Therefore, the prover may send $(\pcommitstring_1, \ldots, \pcommitstring_\ell, \bm{\fe3}, k)$ and the new protocol will satisfy the binding property (a slightly stronger version of which, with respect to a random $\evaluationpoint$, will be necessary; we elaborate upon this later in the section).

More precisely, viewing $\pcommitstring \in \field^{\ell \times \pcommitlength}$ as a matrix whose $i^\text{th}$ row is $\pcommitstring_i$, the prover may send $\pcommitstring$, say, column by column.\footnote{We remark that while sending $\pcommitstring$ column by column naturally corresponds to an \indexproblem instance with a larger alphabet (where symbols are $\ell$-tuples of field elements), since the hardness of \indexproblem holds for the stronger model of one-way communication protocols, the hiding property of the scheme is preserved regardless of the order in which $\pcommitstring$ is sent. This is important in our sumcheck protocol, where a column cannot be sent all at once.} The resulting string, appended with $\bm{\fe3}$ and $k$, is a random \indexproblem instance whose alphabet is $\field^\ell$; and this enables us to show the hiding property for \textsf{algebraic-commit} as we did for \textsf{commit}.

The result is \cref{prot:algebraic-commit}, which enables a prover to commit to multiple messages and decommit (via \cref{prot:decommit}, using $\hat \pcommitstring(\fingerprinttuple, \bm{\fe2})$ as the fingerprint and $\bm{\fe2} \cdot \bm{\fe3}$ as the correction) to an arbitrary linear combination of them.

\begin{theorem}
    \label{thm:pv-algebraic-commitment}
    \cref{prot:algebraic-commit} (\textsf{algebraic-commit}) and \cref{prot:decommit} (\textsf{decommit}) form a streaming commitment protocol with space complexity $s = O\big((\ell + \dimension) \log \fieldsize\big)$ if $\pcommitlength = \fieldsize^{3\ell}$ and $\degree \dimension = \polylog(\fieldsize)$. The scheme is secure against $\poly(s)$-space adversaries and communicates $O(\ell \fieldsize^{3 \ell} \log \fieldsize)$ bits.
    
    Furthermore, if each linear coefficient can be computed in $O(\dimension \log \fieldsize)$ space, then $s = O(\dimension \log \fieldsize)$.
\end{theorem}
Since the proof is a straightforward extension of \cref{thm:pv-commitment}, we defer it to \cref{sec:deferred-pv-algebraic-commit}.

\iflipics
\begin{figure}
\else
\begin{protocol}[label={prot:algebraic-commit}]{$\textsf{algebraic-commit}(\bm{\fe1})$}
\fi
    \textbf{Input:} explicit access to $\pcommitlength, \dimension, \degree, \fieldsize \in \N$ with $\pcommitlength \leq \degree^\dimension$, $\fieldsize > \degree$ and $\field = \field_\fieldsize$; as well as linear coefficients $\evaluationpoint \in \field^\ell$. Streaming access to $\pcommitstring \in \field^{\ell \times \pcommitlength}$ followed by $\bm{\fe3} \in \field^\ell$ and $k \in [\pcommitlength]$.
    
	\begin{enumerate}[label={}]
	    \item[$\bm V$:] Sample $\fingerprinttuple \sim \field^\dimension$ and compute $\hat \pcommitstring(\fingerprinttuple, \evaluationpoint) = \sum_{i = 1}^\ell \evaluationpoint_i \hat \pcommitstring_i(\fingerprinttuple)$, a random linear fingerprint of $\pcommitstring$ with coefficients $\evaluationpoint$, while streaming $\pcommitstring$.
	    
	    Store $\fingerprinttuple, k, \hat \pcommitstring(\fingerprinttuple, \evaluationpoint)$ and the correction $\fe3 = \sum_{i = 1}^\ell \evaluationpoint_i \bm{\fe3}_i$.
	\end{enumerate}
\iflipics
\caption{Protocol $\textsf{algebraic-commit}(\bm{\fe1})$}
\label{prot:algebraic-commit}
\end{figure}
\else
\end{protocol}
\fi

We stress that the binding property of the linear commitment protocol has an important caveat: it is with respect to \emph{the linear combination} $\bm{\fe1} \cdot \bm{\fe2}$, rather than the entire tuple $\bm{\fe1}$. Therefore, if the prover has knowledge of the linear coefficients, it can easily commit to a set of messages $\bm{\fe1}' \neq \bm{\fe1}$ that nonetheless decommits to the same linear combination $\bm{\fe1} \cdot \evaluationpoint$, and $P$ has many choices indeed: the equation $\sum \evaluationpoint_i \bm{\fe1}_i' =  \sum \evaluationpoint_i \bm{\fe1}_i$ is satisfied by all $\evaluationpoint$ in the hyperplane (of size $\fieldsize^{\ell - 1}$) orthogonal to $\bm{\fe1}' - \bm{\fe1}$.

Since our applications require a stronger guarantee -- that $V$ should be able to detect when $P$ commits to $\bm{\fe1}$ and a decommits according to $\bm{\fe1}' \neq \bm{\fe1}$ -- this binding property is insufficient unless $V$ chooses the coefficients $\bm{\fe2}$ \emph{at random}; then the linear combination of $\bm{\fe1}'$ matches that of $\bm{\fe1}$ only with probability $1/\fieldsize$. While in our zero-knowledge protocol for \pep the coefficients are not \emph{uniform}, they are a random evaluation of low-degree polynomials, and the same reasoning holds with a small loss (see \cref{thm:pep-correctness}).

However, an important issue still remains: the exponential dependency of \cref{thm:pv-algebraic-commitment} in the number $\ell$ of field elements that comprise the tuple $P$ commits and decommits to. Concretely, in our applications we have $\ell = \omega(1)$ but can only afford to communicate $\poly(\fieldsize)$ bits. To circumvent this issue, we shall use the following efficient reduction from \indexproblem over bits to the problem of distinguishing a commitment to a fixed element of $\field^\ell$ from a commitment to a random one.

\begin{lemma}
    \label{lem:string-to-bit-index}
    Let $(A, B)$ be a one-way protocol with $s$-bit messages that distinguishes between a length-$\pcommitlength$ algebraic commitment to a fixed $\bm{\fe1} \in \field^\ell$ and a random commitment with advantage $\eps$; that is, such that
    \begin{equation*}
        \abs{\P_{\substack{\pcommitstring \sim \field^{\ell \times \pcommitlength}\\k \sim [\pcommitlength]}}\big[B\big(A(\pcommitstring), (\bm{\fe1}_i \oplus \pcommitstring_{ik} : i \in [\ell]), k\big) \text{ accepts}\big] - \P_{\substack{\pcommitstring \sim \field^{\ell \times \pcommitlength}\\k \sim [\pcommitlength]\\ \bm{\tau} \sim \field^\ell}}\big[B(A(\pcommitstring), \bm{\tau}, k) \text{ accepts}\big]} = \eps.
    \end{equation*}
    Then there exists an average-case one-way communication protocol for (binary) \indexproblem over $\pcommitlength$-bit strings that communicates $O(\ell^2 s \log^2 \fieldsize / \eps^2)$ bits and succeeds with probability $1 - \frac1 e = \frac12 + \Omega(1)$.
\end{lemma}
\begin{proof}
    Define, for ease of notation, $\pcommitstring^{(k)} \coloneqq (\pcommitstring_{ik} : i \in [\ell])$ (i.e., the $k^\text{th}$ column of $\pcommitstring$) and
    \begin{equation*}
        a_{\bm{\tau}} \coloneqq \P\left[B\left(A(\pcommitstring), \bm{\tau} \oplus \pcommitstring^{(k)}, k\right) \text{ accepts}\right] = \E\left[B\left(A(\pcommitstring), \bm{\tau} \oplus \pcommitstring^{(k)}, k\right)\right],
    \end{equation*}
    where we interpret Bob's output as $1$ (respectively\ $0$) when he accepts (respectively\ rejects). Define, also, $\eps_{\bm{\tau}} \coloneqq a_{\bm{\fe1}} - a_{\bm{\tau}}$.
    
    We first argue that, without loss of generality, we can assume $\field = \field_2 = \bitset$. Note that, with $\fieldsize = \abs{\field}$,\footnote{This assumes the acceptance probability of a commitment to $\bm{\fe1}$ is larger than that of a random commitment, which is without loss of generality (otherwise Bob can simply flip his output bit).}
    \begin{equation*}
        \eps = a_{\bm{\fe1}} - \frac1{\fieldsize^\ell} \sum_{\bm{\tau} \in \field^\ell} a_{\bm{\tau}} = \frac1{\fieldsize^\ell} \sum_{\bm{\tau} \in \field^\ell} \eps_{\bm{\tau}}.
    \end{equation*}
    Taking $\ell' \coloneqq \lfloor \ell \log \fieldsize \rfloor$ and $S \subseteq \field^\ell$ as the set of size $2^{\ell'}$ containing $\bm{\fe1}$ and the tuples $\bm{\tau}$ with the largest $\eps_{\bm{\tau}}$, and viewing $\bitset^{\ell'} \subseteq \field^\ell$ via a bijection between $\bitset^{\ell'}$ and $S$, we have 
    \begin{equation*}
        \eps' \coloneqq \frac1{2^{\ell'}} \sum_{\bm{\tau} \in \bitset^{\ell'}} \eps_{\bm{\tau}} \geq \frac{\eps}{3},
    \end{equation*}
    owing to $\abs{S} \geq \fieldsize^\ell/2$ and $\eps_{\bm{\tau}} \geq \eps_{\bm{\tau'}}$ when $\bm{\tau} \in S \setminus \set{\bm{\fe1}}$ and $\bm{\tau'} \in \field^\ell \setminus S$. Therefore, assuming $\field = \field_2$ incurs at most a constant factor in $\eps$ and a $\log \fieldsize$ factor in $\ell$; we shall use $\eps$ and $\ell$ (rather than $\eps'$ and $\ell'$) hereafter for simplicity of notation.
    
    Finally, define, for each $0 \leq i < \ell$,
    \begin{equation*}
        \eps_i \coloneqq \frac1{2^{\ell - i}} \sum_{\substack{\bm{\tau} \in \bitset^{\ell} \\ \forall i' \leq i, ~ \bm{\tau}_{i'} = \bm{\fe1}_{i'}}} \eps_{\bm{\tau}}.
    \end{equation*}
    We divide the analysis into two cases: suppose, first, that $\eps_i \geq \eps_{i - 1} \cdot \left(1 - \frac1{2\ell}\right)$ for all $i \in [\ell - 1]$. Then, by Bernoulli's inequality ($t \leq -1$ implies $(1 + t)^{\ell} \geq 1 + t\ell$), we have
    \begin{equation*}
        \eps_{\ell - 1} = \frac12 \left(a_{\bm{\fe1}} - a_{\bm{\fe1}^{\oplus \ell}}\right) \geq \left(1 - \frac1{2\ell}\right)^{\ell} \cdot \eps \geq \frac{\eps}{2},
    \end{equation*}
    where $\bm{\fe1}^{\oplus i} = \left(\bm{\fe1}_1, \ldots, \bm{\fe1}_{i - 1}, 1 - \bm{\fe1}_i, \bm{\fe1}_{i + 1}, \ldots, \bm{\fe1}_\ell\right)$. Consider the following one-way protocol (with shared randomness) for an \indexproblem instance $(x, j) \in \bitset^\pcommitlength \times [\pcommitlength]$: Alice and Bob jointly sample $2/\eps^2$ independent matrices $\pcommitstring' \sim \bitset^{\ell \times \pcommitlength}$ and permutations $\sigma \sim S_\pcommitlength$; Alice sets $\pcommitstring_i = \pcommitstring'_i \oplus \chf[i = \ell] \cdot \sigma(x)$ (where $\sigma(x)_k \coloneqq x_{\sigma(k)}$), simulates $A(\pcommitstring)$ and sends the resulting messages in a $2s/\eps^2$-bit string to Bob.
    
    With knowledge of $j$, Bob finishes the simulations $B(A(\pcommitstring), \bm{\fe3}, k)$, using coordinate $k = \sigma^{-1}(j)$ and correction $\bm{\fe3} = \bm{\fe1} \oplus \pcommitstring'^{(k)}$; he computes their empirical mean $\mu$, outputs $\bm{\fe1}_\ell$ if $\mu \geq a_{\bm{\fe1}} + \eps/2$, and outputs $1 - \bm{\fe1}_\ell$ otherwise.
    
    Correctness follows from the observation that, if $x_j = \sigma(x)_k = \bm{\fe1}_\ell$, then $\bm{\fe3} = \bm{\fe1} \oplus \pcommitstring^{(k)}$, so $\E[\mu] = a_{\bm{\fe1}}$; since the $(\pcommitstring, k)$ pairs are uniform and independent,
    \begin{equation*}
        \P\left[\mu \leq a_{\bm{\fe1}} - \frac{\eps}{2}\right] \leq \frac1 e
    \end{equation*}
    by the Chernoff-Hoeffding bound (\cref{lem:additive-chernoff}, with $2/\eps^2$ samples and $\delta = \eps/2$). Likewise, when $x_j = 1$ we have $\bm{\fe1} = \bm{\fe1}^{\oplus \ell} \oplus \pcommitstring^{(k)}$; then $\E[\mu] = a_{\bm{\fe1}^{\oplus \ell}} \leq a_{\bm{\fe1}} - \eps$ and an application of the Chernoff-Hoeffding bound (with the same parameters) yields the same guarantee.
    
    We now consider the second case: suppose $\eps_i < \eps_{i - 1} \cdot \left(1 - \frac1{2\ell}\right)$ for some $i \in [\ell - 1]$; we take, without loss of generality, the minimal such $i$. Then
    \begin{align*}
        \frac1{2^{\ell - i}} \sum_{\substack{\bm{\tau} \in \bitset^{\ell} \\ \forall i' \leq i, ~ \bm{\tau}_{i'} = \bm{\fe1}_{i'}}} \eps_{\bm{\tau}^{\oplus i}} &= \frac1{2^{\ell - i}} \sum_{\substack{\bm{\tau} \in \bitset^{\ell} \\ \forall i' < i, ~ \bm{\tau}_{i'} = \bm{\fe1}_{i'}\\\bm{\tau}_i = 1 - \bm{\fe1}_i}} \eps_{\bm{\tau}} \\
        &= 2 \eps_{i - 1} - \eps_i\\
        &> \eps_{i - 1} \left(1 + \frac1{2\ell}\right),
    \end{align*}
    and thus
    \begin{align*}
        \frac1{2^{\ell - i}} \left(\sum_{\substack{\bm{\tau} \in \bitset^{\ell} \\ \forall i' \leq i, ~ \bm{\tau}_{i'} = \bm{\fe1}_{i'}}} \left(\eps_{\bm{\tau}^{\oplus i}} - \eps_{\bm{\tau}}\right)\right) &= \frac1{2^{\ell - i}}\left(\sum_{\substack{\bm{\tau} \in \bitset^{\ell} \\ \forall i' \leq i, ~ \bm{\tau}_{i'} = \bm{\fe1}_{i'}}} \left(a_{\bm{\tau}^{\oplus i}} - a_{\bm{\tau}}\right)\right)\\
        &> \frac{\eps_{i - 1}}{\ell}\\
        &\geq \frac{\eps}{2 \ell}.
    \end{align*}
    
    We will use a similar strategy to the previous case, although the expression we must estimate involves many more terms (indeed, $2^{\ell - i + 1}$ of them). Consider the following one-way protocol for an \indexproblem instance $(x, j) \in \bitset^\pcommitlength \times [\pcommitlength]$: Alice and Bob jointly sample $64 \ell^2/\eps^2$ independent matrices $\pcommitstring' \sim \bitset^{\ell \times \pcommitlength}$ and permutations $\sigma \sim S_\pcommitlength$; Alice sets $\pcommitstring_{i'} = \pcommitstring'_{i'} \oplus \chf[i' = i] \cdot \sigma(x)$, computes and sends all messages $A(\pcommitstring)$ in a $64 \ell^2s/\eps^2$-bit string to Bob.\footnote{Note that the only difference in Alice's strategy, as compared to the previous case, is the row where she inserts $\sigma(x)$ and the number of simulations of $A$.} (Recall that assuming $\field = \bitset$ incurs constant and logarithmic factors in $\eps$ and $\ell$, respectively, so that Alice's message is $O(\ell^2 s \log^2 \fieldsize / \eps^2)$ bits long.)
    
    For each $A(\pcommitstring)$ sent by Alice, Bob simulates $B\left(A(\pcommitstring), \bm{\tau} \oplus \pcommitstring'^{(k)}, k\right)$ with $k = \sigma^{-1}(j)$ \emph{for all} $\bm{\tau}$ satisfying  $\bm{\tau}_{i'} = \bm{\fe1}_{i'}$ when $i' \leq i$. He computes the empirical mean $\mu$ of
    \begin{equation*}
        \frac1{2^{\ell - i}}\left(\sum_{\substack{\bm{\tau} \in \bitset^{\ell} \\ \forall i' \leq i, ~ \bm{\tau}_{i'} = \bm{\fe1}_{i'}}} \left(B\left(A(\pcommitstring), \bm{\tau}^{\oplus i} \oplus \pcommitstring'^{(k)}, k\right) - B\left(A(\pcommitstring), \bm{\tau} \oplus \pcommitstring'^{(k)}, k\right)\right)\right),
    \end{equation*}
    outputs $0$ if the result is non-negative, and outputs $1$ otherwise.
    
    To prove correctness, first note that
    \begin{equation*}
        \bm{\tau} \oplus \pcommitstring'^{(k)} = \left\{\begin{array}{ll}\bm{\tau} \oplus \pcommitstring^{(k)}, & \text{ when } x_j = 0\\\bm{\tau}^{\oplus i} \oplus \pcommitstring^{(k)}, & \text{ when } x_j = 1,\end{array}\right.
    \end{equation*}
    so that, when $x_j = 0$,
    \begin{equation*}
        \E[\mu] = \frac1{2^{\ell - i}}\left(\sum_{\substack{\bm{\tau} \in \bitset^{\ell} \\ \forall i' \leq i, ~ \bm{\tau}_{i'} = \bm{\fe1}_{i'}}} \left(a_{\bm{\tau}^{\oplus i}} - a_{\bm{\tau}}\right)\right) > \frac{\eps}{2\ell},
    \end{equation*}
    and when $x_j = 1$ we have $\E[\mu] < -\eps/2\ell$ (since the order of each pair of terms in the sum is flipped).
    
    We conclude with an application of Hoeffding's inequality (\cref{lem:hoeffding}, with $a = -1$, $b = 1$, $\delta = 1/2$ and $64 \ell^2/\eps^2$ samples): in the $x_j = 0$ case,
    \begin{equation*}
        \P\left[\mu \leq \frac{\eps}{4\ell}\right] \leq \frac1 e;
    \end{equation*}
    and, likewise, in the $x_j = 1$ case we have $\P\left[\mu \geq -\frac{\eps}{4\ell}\right] \leq \frac1 e$.
\end{proof}

\subsection{A verifier-to-prover temporal commitment}
\label{sec:vp-commitment}

The goal of this section is to construct the second main component towards our streaming zero-knowledge protocols. While it is not formally a commitment protocol (as per \cref{def:commitment}), it is useful to conceptualise it as $V$ committing to its internal randomness \emph{before} the input is streamed (hence \emph{temporal}). 

Roughly speaking, we would like to ensure that a malicious verifier cannot choose the point $\fingerprinttuple$ at which it (allegedly) computes its fingerprint \emph{after it sees the input $(x, \evaluationpoint)$}, as that would allow it to learn more than $f^x(\evaluationpoint)$. (For example, in the \indexproblem case it could claim that $\fingerprinttuple = j + 1$ and learn $\hat x(j+1) = x_{j+1}$.) We will prove, in 3 steps, a lemma formalising the intuition that a space-$s$ algorithm cannot remember the positions of significantly more than $s$ elements, which will later enable the construction of a simulator. As in the case of algebraic commitments, we will in fact prove a stronger statement: that this holds not only in the case of streaming algorithms, but in the stronger model of one-way communication protocols.

We first define two variants of \searchindex in the one-way communication complexity model, which we call \reconstruct and \pair (see \cref{def:reconstruct,def:pair}). In \reconstruct, Bob's task is to output the symbols at \emph{every} coordinate of the input $\vcommitstring$ (rather than receiving a single coordinate $j$ and outputting only $\vcommitstring_j$, as in \indexproblem); in other words, Bob should reconstruct the input as best he can. In \pair, as in \searchindex, Bob's task is again to output the symbol at a single coordinate; but rather than receiving the index as part of the input, Bob is free to choose a coordinate-symbol pair $(i, \fe1)$ and succeeds if $\fe1 = \vcommitstring_i$. (Note that in both \reconstruct and \pair, Bob does not receive any additional input besides Alice's message.)

Our first two steps are as follows. We first study \reconstruct and show, in \cref{lem:reconstruct-bound}, that if Alice's message has $s$ bits, Bob cannot reconstruct significantly more than $s$ coordinates of the input. Then, in \cref{lem:pair-bound}, we show how this bound for \reconstruct implies a related bound for \pair; more precisely, we prove that there exists a size-$s$ set $C$ of coordinates such that the probability Bob outputs a correct coordinate-symbol pair $(i, \vcommitstring_i)$ where $i \notin C$ is arbitrarily small.

While \cref{lem:pair-bound} immediately implies an analogous statement for streaming algorithms, it is not yet enough for our purposes. The reason is that our verifier will read additional information, i.e., a fixed -- but unknown -- \pep instance $(x, \evaluationpoint)$ between reading a \pair input and writing its output. While it is intuitively clear that this should not help the verifier in any way (as the \pep and \pair instances are uncorrelated), we still require a slight extension of \cref{lem:pair-bound}.

To this end we define, for each fixed string $x \in \alphabet^n$, a variant of \pair that we call $\pair(x)$ (\cref{def:pair-x}). The only difference between this one-way communication problem and \pair is that Bob receives the string $x$ in addition to Alice's message $a$. In \cref{thm:correct-set}, we show that the existence of a set capturing most of the correct outputs of \pair implies such a set $C$ also exists for $\pair(x)$; crucially, $C$ is determined by $a$ and \emph{does not depend on $x$}. This last result then immediately implies an analogous one for streaming algorithms.

Let us begin with the definitions:

\begin{definition}
    \label{def:reconstruct}
    \reconstruct is the following one-way communication problem: Alice receives a string $\vcommitstring \sim \alphabet^\vcommitlength$ and sends Bob an $s$-bit message $a$; after receiving $a$, Bob outputs a string $b \in \alphabet^\vcommitlength$. The \emph{score} of an execution is the number of matching coordinates between $\vcommitstring$ and $b$, i.e., $\abs{\set{i \in [\vcommitlength] : b_i = \vcommitstring_i}}$.
\end{definition}

\begin{definition}
    \label{def:pair}
    Let \pair denote the following one-way communication problem: Alice receives a string $\vcommitstring \sim \alphabet^\vcommitlength$ and sends Bob an $s$-bit message $a$; after receiving $a$, Bob outputs a pair $(\fe1, i) \in \alphabet \times [\vcommitlength]$. The execution succeeds if $\fe1 = \vcommitstring_i$.
\end{definition}

Note that both are definitionally average-case problems, as $\vcommitstring$ is sampled uniformly. We now proceed to the first step towards the goal of this section: a proof that, in our parameter settings of interest for $\abs{\alphabet}$ and $s$ (as functions of $\vcommitlength$), the expected score of any protocol for \reconstruct is tightly constrained by the message length $s$.

\begin{lemma}
    \label{lem:reconstruct-bound}
    Any one-way protocol for \reconstruct with alphabet size $\abs{\alphabet} = O(\vcommitlength / \log\log \vcommitlength)$, $\abs{\alphabet} \geq 32\vcommitlength / \log\log \vcommitlength$ and message length $s$, where $\log \vcommitlength \leq s = \polylog(\vcommitlength)$, achieves an expected score of at most $s + o(s)$.
\end{lemma}
\begin{proof}
    By the minimax theorem, we may assume Alice's and Bob's strategies are both deterministic, so that there exists a set of messages $A \subseteq \bitset^s$ that partitions the set $\alphabet^\vcommitlength$ of input strings by $\set{P_a : a \in A}$, where Bob outputs $b = b(a) \in \alphabet^\vcommitlength$ whenever $\vcommitstring \in P_a$.
    
    Observe that Bob's optimal strategy is to set $b_i$ as the most frequent symbol at the $i^\text{th}$  coordinate among the strings of $P_a$; we can thus index the partition by $b \in B \coloneqq \set{b(a) : a \in A}$, setting $P_b = P_{b(a)} = P_a$. (Note that while $\set{P_b : b \in B}$ may be a smaller partition than $\set{P_a : a \in A}$, the expected scores of the protocols induced by both partitions are the same.)
    
    Define the random variable $M_b \coloneqq \set{i \in [\vcommitlength] : \vcommitstring_i = b_i}$. For simplicity of notation, denote also $\alphabetsize \coloneqq \abs{\alphabet}$. Note that the expected score of this one-way protocol is
    \begin{align*}
        \E_{\vcommitstring \sim \alphabet^\vcommitlength}\left[\sum_{b \in B} \chf[\vcommitstring \in P_b] \cdot \abs{M_b}\right] &= \sum_{b \in B} \P[\vcommitstring \in P_b] \cdot \E_{\vcommitstring \sim P_b}\left[\abs{M_b}\right]\\
        &= \sum_{\substack{b \in B\\\abs{P_b} \geq \frac{s}{\vcommitlength} \cdot \frac{\alphabetsize^\vcommitlength}{2^s}}} \P[\vcommitstring \in P_b] \cdot \E_{\vcommitstring \sim P_b}\left[\abs{M_b}\right] + \sum_{\substack{b \in B\\\abs{P_b} < \frac{s}{\vcommitlength} \cdot \frac{\alphabetsize^\vcommitlength}{2^s}}} \P[\vcommitstring \in P_b] \cdot \E_{\vcommitstring \sim P_b}\left[\abs{M_b}\right].
    \end{align*}
    We bound the first term by the largest expectation, and the second by observing that the union of sets $P_b$ with $\abs{P_b} \leq \frac{s}{\vcommitlength} \cdot \frac{\alphabetsize^\vcommitlength}{2^s}$ contain at most an $s/\vcommitlength$ fraction of all length-$\vcommitlength$ strings:
    \begin{align*}
        \E_{\vcommitstring \sim \alphabet^\vcommitlength}\left[\sum_{b \in B} \chf[\vcommitstring \in P_b] \cdot \abs{M_b}\right] &\leq \max_{\substack{b \in B\\\abs{P_b} \geq \frac{s \alphabetsize^\vcommitlength}{\vcommitlength \cdot 2^s}}} \E_{\vcommitstring \sim P_b}\left[\abs{M_b}\right] + \sum_{\substack{b \in B\\\abs{P_b} < \frac{s \alphabetsize^\vcommitlength}{\vcommitlength \cdot 2^s}}} \P[\vcommitstring \in P_b] \cdot \vcommitlength\\
        &\leq \max_{\substack{b \in B\\\abs{P_b} \geq \frac{s \alphabetsize^\vcommitlength}{\vcommitlength \cdot 2^s}}} \E_{\vcommitstring \sim P_b}\left[\abs{M_b}\right] + s.
    \end{align*}

    Let $\delta \in (0,1)$ be such that the volume of Hamming balls of radius $\delta$ is $\mathcal{V} \coloneqq \frac{s \alphabetsize^\vcommitlength}{\vcommitlength \cdot 2^s} \leq \frac{s \alphabetsize^\vcommitlength}{\vcommitlength^2}$. (Recall that $s \geq \log \vcommitlength$.) For any $b \in B$, the set $P_b$ that maximises
    \begin{equation*}
        E_{\vcommitstring \sim P_b}\left[\abs{M_b}\right] = \abs{P_b}^{-1} \sum_{\vcommitstring \in P_b} \abs{\set{i \in [\vcommitlength] : \vcommitstring_i = b_i}}
    \end{equation*}
    is $P_b = \ball(b, \delta')$, the ball centered at $b$ (whose radius $\delta'$ is determined by the equality $\abs{\ball(b, \delta')} = \abs{P_b}$). Since $\abs{P_b} \geq \mathcal{V}$ implies $\delta' \geq \delta$, we have
    \begin{equation*}
        \frac1{\abs{P_b}} \cdot \sum_{\vcommitstring \in \ball(b, \delta')} \abs{\set{i \in [\vcommitlength] : \vcommitstring_i = b_i}} \leq \frac1{\mathcal{V}} \cdot \sum_{\vcommitstring \in \ball(b, \delta)} \abs{\set{i \in [\vcommitlength] : \vcommitstring_i = b_i}},
    \end{equation*}
    so it suffices to bound the right-hand side. (The inequality follows from the observation that the left-hand side is a weighted average between the right-hand side and the expectation over $\vcommitstring \sim \ball(b, \delta') \setminus \ball(b, \delta)$, which is smaller.)
    
    Define $\eps \coloneqq 1 - \delta$. We aim to upper bound $E_{\vcommitstring \sim P_b}\left[\abs{M_b}\right]$, and set as an intermediate goal to prove upper and lower bounds for $\eps$. To this end, we will use the following standard approximations (see, e.g., \cite{GRS12}) for $H = H_2$ when $\sigma$ (or $1 - \sigma$) is small:
    \begin{equation}
        \label{eq:entropy-approximation}
        H(\sigma) = H(1 - \sigma) \in \left[\sigma \log \frac1{\sigma}, \sigma \left(\log \frac1{\sigma} + \frac{2}{\ln 2}\right)\right]
    \end{equation}

    We begin with the lower bound on $\eps$, which uses the lower bound of \cref{eq:entropy-approximation}  and follows by showing that the volume of a ball with radius $1 - \frac{\log \alphabetsize}{\vcommitlength \log\log \alphabetsize}$ is larger than $\mathcal{V}$; then $\delta < 1 - \frac{\log \alphabetsize}{\vcommitlength \log\log \alphabetsize}$, or, equivalently, $\eps = 1 - \delta > \frac{\log \alphabetsize}{\vcommitlength \log\log \alphabetsize}$.
    
    We have
    \begin{flalign*}
        && H_\alphabetsize&\left(1 - \frac{\log \alphabetsize}{\vcommitlength \log\log \alphabetsize}\right) &&\\
        &&&= \frac{\left(1 -\frac{\log \alphabetsize}{\vcommitlength \log\log \alphabetsize}\right)\log(\alphabetsize - 1) + H\left(1 - \frac{\log \alphabetsize}{\vcommitlength \log\log \alphabetsize}\right)}{\log \alphabetsize} && (\text{by \cref{eq:q-ary-entropy}})\\
        &&&= \frac{\left(1 - \frac{\log \alphabetsize}{\vcommitlength \log\log \alphabetsize}\right)\log(\alphabetsize - 1) + H\left(\frac{\log \alphabetsize}{\vcommitlength \log\log \alphabetsize}\right)}{\log \alphabetsize} && (\text{by \cref{eq:binary-entropy}})\\
        &&&\geq \frac{\left(1 - \frac{\log \alphabetsize}{\vcommitlength \log\log \alphabetsize}\right) \left(\log \alphabetsize + \log\left(1 - \frac1{\alphabetsize}\right)\right)}{\log \alphabetsize} + \frac{\log\left(\frac{\vcommitlength \log\log \alphabetsize}{\log \alphabetsize}\right)}{\vcommitlength \log\log \alphabetsize} && (\text{by \cref{eq:entropy-approximation}})\\
        &&&= 1 + \left(\frac1{\log \alphabetsize} - \frac1{\vcommitlength \log\log \alphabetsize}\right)\log\left(1 - \frac1{\alphabetsize}\right) + \frac{\log \frac{\vcommitlength}{\alphabetsize} + \log\log\log \alphabetsize}{\vcommitlength \log\log \alphabetsize} - \frac1{\vcommitlength} && \\
        &&&\geq 1 - \frac1{\alphabetsize \ln 2}\left(1 + \frac1\alphabetsize\right)\left(\frac1{\log \alphabetsize} - \frac1{\vcommitlength \log\log \alphabetsize}\right) + \frac{\log \frac{\vcommitlength}{\alphabetsize} + \log\log\log \alphabetsize}{\vcommitlength \log\log \alphabetsize} - \frac1{\vcommitlength} && (\text{by \cref{eq:log-approximation}})\\
        &&&\geq 1 - \frac1{\alphabetsize \ln 2}\left(1 + \frac1\alphabetsize\right)\left(\frac1{\log \alphabetsize} - \frac1{\vcommitlength \log\log \alphabetsize}\right) - \frac1{\vcommitlength} && \\
        &&&\geq 1 - \frac{3}{2\vcommitlength}, &&
    \end{flalign*}
    where the second-to-last inequality uses $\vcommitlength \geq \alphabetsize$;
    and the last inequality uses $\alphabetsize = \Theta\left(\frac{\vcommitlength}{\log\log \vcommitlength}\right)$ to bound the first negative term to order $\Theta\left(\frac{\log\log \vcommitlength}{\vcommitlength \log \vcommitlength}\right)$, so the $1/\vcommitlength$ term dominates. Therefore,
    \begin{equation*}
        \alphabetsize^{H_\alphabetsize\left(1 - \frac{\log \alphabetsize}{\vcommitlength \log\log \alphabetsize}\right) \vcommitlength} \geq \alphabetsize^\vcommitlength/\alphabetsize^{3/2},
    \end{equation*}
    and thus, by \cref{eq:hamming-volume-approximation}, the volume of a ball (centered at any point $b$) of radius $1 - \frac{\log \alphabetsize}{\vcommitlength \log\log \alphabetsize} = 1 - \frac{\polylog(\vcommitlength)}{\vcommitlength}$ satisfies 
    \begin{align*}
        \abs{\ball\left(b, 1 - \frac{\log \alphabetsize}{\vcommitlength \log\log \alphabetsize}\right)} &\geq \frac{\alphabetsize^{H_\alphabetsize\left(1 - \frac{\log \alphabetsize}{\vcommitlength \log\log \alphabetsize}\right) \vcommitlength}}{\sqrt{\log \vcommitlength}}\\
        &\geq \frac{\alphabetsize^\vcommitlength}{2^{\frac{3}{2}\log \alphabetsize + \frac1 2 \log\log \vcommitlength}}\\
        &\geq \frac{\alphabetsize^\vcommitlength}{2^{\frac{7}{4}\log \alphabetsize}}.
    \end{align*}
    Then
    \begin{align*}
        \abs{\ball(b, \delta)} = \mathcal{V} &= \frac{s \alphabetsize^\vcommitlength}{\vcommitlength \cdot 2^s}\\
        &\leq \frac{\alphabetsize^\vcommitlength \polylog(\vcommitlength)}{\vcommitlength^2}\\
        &\leq \frac{\alphabetsize^\vcommitlength}{2^{\frac{15}{8}\log \vcommitlength}}\\
        &\leq \abs{\ball\left(b, 1 - \frac{\log \alphabetsize}{\vcommitlength \log\log \alphabetsize}\right)},
    \end{align*}
    and we conclude that $\eps = 1 - \delta > \frac{\log \alphabetsize}{\vcommitlength \log\log \alphabetsize}$.

    We now proceed to the upper bound on $\eps$, which will use the upper bound of \cref{eq:entropy-approximation}. Since  $\alphabetsize^{H_\alphabetsize(\delta) \vcommitlength} \geq \mathcal{V} = \frac{s \alphabetsize^\vcommitlength}{\vcommitlength \cdot 2^s}$ (\cref{eq:hamming-volume-approximation}), taking the logarithm of both sides and using \cref{eq:q-ary-entropy} yields
    \begin{equation}
        \label{eq:entropy-volume-inequality}
        \frac{(1 - \eps)\log(\alphabetsize - 1) + H(1 - \eps)}{\log \alphabetsize} = H_\alphabetsize(1 - \eps) = H_\alphabetsize(\delta) \geq 1 - \frac{s + \log \frac{\vcommitlength}{s}}{\vcommitlength \log \alphabetsize}.
    \end{equation}
    
    Note that the right-hand side is $1 - o(1)$ because $s = o(\vcommitlength)$; then, $\delta$ is within $o(1)$ distance of the maximiser $1 - 1/\alphabetsize = 1 - o(1)$ of $H_\alphabetsize$, so that $\delta = 1 - o(1)$ and $\eps = o(1)$.
    
    This allows us to bound $H(\eps) = H(1 - \eps)$ from above via \cref{eq:entropy-approximation}, which, combined with \cref{eq:entropy-volume-inequality} (multiplied by $\log \alphabetsize$), implies
    \begin{equation*}
        (1 - \eps) \log(\alphabetsize - 1) + \eps \log \frac1{\eps} + \frac{2 \eps}{\ln 2} \geq \log \alphabetsize - \frac{s + \log \vcommitlength - \log s}{\vcommitlength}.
    \end{equation*}
    Rearranging yields
    \begin{equation*}
        \eps \left(\log \eps + \log \alphabetsize + \log \left(1 - \frac1{\alphabetsize}\right) - \frac{2}{\ln 2}\right) \leq  \frac{s + \log \vcommitlength - \log s}{\vcommitlength} + \log \left(1 - \frac1{\alphabetsize}\right).
    \end{equation*}

    The bounds $-\log(1-1/\alphabetsize) = O(1/\alphabetsize) = O(\log\log \vcommitlength / \vcommitlength)$ (\cref{eq:log-approximation}) and $s \geq \log \vcommitlength$ show that the right-hand side is $O(s/\vcommitlength)$; and \cref{eq:log-approximation} along with $\log \alphabetsize = \log \vcommitlength - \log\log\log \vcommitlength + \Theta(1) = \big(1 - o(1)\big)\log \vcommitlength$ implies the left-hand side is $\Omega\big(\eps (\log \eps + \log \vcommitlength)\big)$. Therefore, the inequality above simplifies to
    \begin{equation*}
        \eps (\log \eps + \log \vcommitlength) = O\left(\frac{s}{\vcommitlength}\right).
    \end{equation*}

    Now, if we had $\eps = \Omega(s/\vcommitlength)$, then
    \begin{align*}
        \eps (\log \eps + \log \vcommitlength) &= \eps \big(\log s - \log \vcommitlength + \log \vcommitlength + \Omega(1)\big)\\
        &= \Omega(\eps \log s) = \omega(s/\vcommitlength),
    \end{align*}
    a contradiction. We thus conclude that $\eps = o(s/\vcommitlength)$ (and, in particular, that $\eps$ is both lower and upper bounded by $\polylog(\vcommitlength)/\vcommitlength$).

    Returning to the goal of bounding the expected score, we now show that most of the volume of a Hamming ball of radius $\delta$ is close to its boundary. More precisely, consider the volume $\mathcal{V}'$ of a ball of radius $\delta' = 1 - 2\eps$. As $\eps = \vcommitlength^{-1} \polylog(\vcommitlength)$, \cref{eq:hamming-volume-approximation} applies, giving $\mathcal{V}' \leq \alphabetsize^{H_\alphabetsize(1 - 2\eps)}$ and
    \begin{equation*}
        \mathcal{V} = \Omega\left(\frac{\alphabetsize^{H_\alphabetsize(1 - \eps)}}{\sqrt{\eps \vcommitlength}}\right) = \Omega\left(\frac{\alphabetsize^{H_\alphabetsize(1 - \eps)}}{\sqrt{s}}\right).
    \end{equation*}
    so that
    \begin{equation*}
        \frac{\mathcal{V}'}{\mathcal{V}} = O\left(\sqrt{s} \cdot \alphabetsize^{-\left(H_\alphabetsize(1 - \eps) - H_\alphabetsize(1 - 2\eps)\right)\vcommitlength}\right).
    \end{equation*}
    
    We can bound the coefficient in the exponent as follows:
    \begin{flalign*}
        && H_\alphabetsize(1 - \eps) - H_\alphabetsize(1 - 2\eps) &= \frac{\eps \log (\alphabetsize - 1) + H(\eps) - H(2\eps)}{\log \alphabetsize} &&\\
        &&&\geq \frac{\eps}{\log \alphabetsize} \left(\log (\alphabetsize - 1) + \log \frac1{\eps} - 2 \log \frac1{2\eps} - \frac{4}{\ln 2}\right) && (\text{by \cref{eq:entropy-approximation}})\\
        &&&= \frac{\eps}{\log \alphabetsize} \left(\log (\eps \alphabetsize) + \log \left(1 - \frac1{\alphabetsize}\right) + 2 - \frac{4}{\ln 2}\right) &&\\
        &&&\geq \frac{\eps \log\log \alphabetsize}{\log \alphabetsize},  &&
    \end{flalign*}
    where the last inequality follows from $\eps \alphabetsize > \frac{\alphabetsize \log \alphabetsize}{\vcommitlength \log\log \alphabetsize} = \Theta\left(\frac{\log \alphabetsize}{\log^2 \log \alphabetsize}\right)$ when the constant in $\Theta(\cdot)$ is large enough ($\alphabetsize \geq 32 \vcommitlength / \log\log \vcommitlength$ suffices, as $\log(1 - 1/\alphabetsize) + 2 - 4/\ln 2 > -5$).
    Therefore,
    \begin{align*}
        \sqrt{s} \cdot \alphabetsize^{-\left(H_\alphabetsize(1 - \eps) - H_\alphabetsize(1 - 2\eps)\right)\vcommitlength} &\leq \sqrt{s} \cdot \alphabetsize^{- \frac{\eps \vcommitlength \log\log \alphabetsize}{\log \alphabetsize}}\\
        &= \sqrt{s} \cdot 2^{- \eps \vcommitlength \log\log \alphabetsize}\\
        &< \sqrt{s} \cdot 2^{- \log \alphabetsize} \\
        &= \frac{\sqrt{s}}{\alphabetsize}\\
        &= \Theta\left(\frac{\sqrt{s} \log\log \vcommitlength}{\vcommitlength}\right)\\
        &= o(s/\vcommitlength),
    \end{align*}
    where the last line is due to $\sqrt{s} \geq \sqrt{\log \vcommitlength} = \omega(\log\log \vcommitlength)$ and the strict inequality to $\eps > \frac{\log \alphabetsize}{\vcommitlength \log\log \alphabetsize}$. Therefore, $\mathcal{V}'/\mathcal{V} = o(s/\vcommitlength)$, showing that the volume of a ball of radius $1 - \eps$ is indeed concentrated in points of distance at least $1 - 2\eps$.
    
    Finally, we conclude that
    \begin{align*}
        \E_{\vcommitstring \sim \alphabet^\vcommitlength}\left[\sum_{b \in B} \chf[\vcommitstring \in P_b] \cdot \abs{M_b}\right] &\leq \max_{\substack{b \in B\\\abs{P_b} \geq \frac{s \alphabetsize^\vcommitlength}{\vcommitlength \cdot 2^s}}} \E_{\vcommitstring \sim P_b}\left[\abs{M_b}\right] + s\\
        &\leq s + \frac1{\mathcal{V}} \cdot \sum_{\vcommitstring \in \ball(b, \delta)} \abs{\set{i \in [\vcommitlength] : \vcommitstring_i = b_i}}\\
        &\leq s +  \frac{\mathcal{V}'}{\mathcal{V}} \cdot \vcommitlength + \left(1 - \frac{\mathcal{V}'}{\mathcal{V}} \right) \cdot 2\eps \vcommitlength\\
        &= s + o(s),
    \end{align*}
    as desired.
\end{proof}

At this stage, we have an upper bound on the expected score of any one-way communication protocol for \reconstruct. The next step is to show that it implies a similar bound for the communication problem \pair; indeed, it seems intuitively clear that \reconstruct is no harder than \pair, as it allows Bob to output an independent guess for each coordinate. We formalise this intuition in the following lemma.

\begin{lemma}
    \label{lem:pair-bound}
    Any one-way protocol for \pair with alphabet size $\frac{32 \vcommitlength}{\log\log \vcommitlength} 
    \leq \abs{\alphabet} = O\left(\frac{\vcommitlength}{\log\log \vcommitlength}\right)$ and message length $s$, where $\log \vcommitlength \leq s = \polylog(\vcommitlength)$, satisfies the following: there exists an event $E$ (depending only on $\vcommitstring$) with $\P[E] = 1 - o(1)$ and a set $C$ of size $s$ (depending only on Alice's message) such that
    \begin{equation*}
        \P\big[\text{Bob outputs } (\vcommitstring_i,i) \text{ with } i \notin C \big| E \big] = o(1).
    \end{equation*}
\end{lemma}
\begin{proof}
    We will first show how to construct a protocol for \reconstruct given one for \pair, and then use \cref{lem:reconstruct-bound} to conclude; as in that lemma, we define $\set{P_a}$ as the partition induced by Alice's messages $a = a(\vcommitstring) \in A$ (we can assume Alice to be deterministic, as before, by the minimax theorem; then $a$ is a random variable determined by $\vcommitstring$).
    
    Recall that in a protocol for \pair, Bob's output is a random variable $b(a) \in \alphabet \times [\vcommitlength]$;\footnote{Note that, in contrast with Alice, we cannot assume Bob is deterministic. We wish to bound the number of points in the support of $b$ that aggregate all but a subconstant amount of probability weight in correct solutions to the problem. This is not a function of the \emph{value} of $b$, but of its \emph{distribution}, so the minimax principle does not apply.} our goal is to construct, from this random variable, an entire string $y \in \alphabet^\vcommitlength$ and apply the expected score bound to it. For ease of notation, when the message $a$ is fixed we write $b = (b_1, b_2) = b(a)$; note that $b$ is independent of the conditional distribution $\vcommitstring \sim P_a$ of the input, since upon fixing $a$ it is solely a function of Bob's internal randomness. We will denote its distribution by $\mu = \mu(a)$, and the conditional distribution of $b_2$ when $b_1 = i$ by $\mu_i$.

    The (\pair) protocol's success probability, conditional on receiving $a$, is given by
    \begin{align*}
        \sum_{i = 1}^\vcommitlength \P_{\subalign{\vcommitstring &\sim P_a\\b &\sim \mu}}[b = (\vcommitstring_i, i)] &= \sum_{i = 1}^\vcommitlength \P_{b \sim \mu}[b_2 = i] \cdot \P_{\subalign{\vcommitstring &\sim P_a\\b &\sim \mu}}[b_1 = \vcommitstring_i ~|~ b_2 = i]\\
        &= \sum_{i = 1}^\vcommitlength \P_{b \sim \mu}[b_2 = i] \cdot \P_{\subalign{\vcommitstring &\sim P_a\\b_1 &\sim \mu_i}}[b_1 = \vcommitstring_i].
    \end{align*}

    Define $y = y(a) \in \alphabet^\vcommitlength$ as the string whose $i^\text{th}$ coordinate is the most frequent symbol at the $i^\text{th}$ coordinate in $P_a$ (as before, $y$ is the best attempt at reconstructing the input $\vcommitstring$ given to Alice). Now, consider the \reconstruct protocol that outputs the string whose $i^\text{th}$ coordinate is the random variable $b_1 \sim \mu_i$. Since, for each $i \in [\vcommitlength]$, the symbol $\alpha \in \Gamma$ maximising $\P_{\vcommitstring \sim P_a}[\alpha = \vcommitstring_i]$ is $y_i$, the expected score of the resulting protocol (conditioned on $a$) is
    \begin{align*}
        \sum_{i = 1}^\vcommitlength \P_{\vcommitstring \sim P_a}[b_1 = \vcommitstring_i \mid b_2 = i] &= \sum_{i = 1}^\vcommitlength \P_{\subalign{\vcommitstring &\sim P_a\\b_1 &\sim \mu_i}}[b_1 = \vcommitstring_i]\\
        &=\sum_{i = 1}^\vcommitlength \sum_{\alpha \in \alphabet} \P_{b_1 \sim \mu_i}[b_1 = \alpha] \cdot \P_{\vcommitstring \sim P_a}[\alpha = \vcommitstring_i]\\
        &\leq \sum_{i = 1}^\vcommitlength \P_{\vcommitstring \sim P_a}[y_i = \vcommitstring_i]\\
        &= \E_{\vcommitstring \sim P_a}\left[\abs{M_a}\right],
    \end{align*}
    where, as before, $M_a = \set{i \in [\vcommitlength] : y_i = \vcommitstring_i}$.
    
    Recall that in \cref{lem:reconstruct-bound} we showed that, as long as $\abs{P_a} \geq \frac{s \abs{\alphabet}^\vcommitlength}{\vcommitlength 2^s}$, the above expectation is $o(s)$.
    To conclude, we will use the following claim, whose proof is deferred to \cref{sec:deferred-probability-ip}:
    \begin{claim}
    \label{clm:probability-ip}
        Let $p,q \in [0,1]^\vcommitlength$ be probability vectors and $t \leq \vcommitlength$ be an integer. There exists a set $C \subseteq [\vcommitlength]$ of size $t$ such that $\sum_{i \in [\vcommitlength] \setminus C} p_i q_i \leq 1/t$.
    \end{claim}

    Note that while $r \in [0,1]^\vcommitlength$ defined by $r_i = \P[b_1 = \vcommitstring_i ~|~ b_2 = i]$ is not a probability vector, we may normalise it to obtain one: applying \cref{clm:probability-ip} to $p = \big(\P[b_2 = i] : i \in [\vcommitlength]\big)$, $q = r/\norm{r}_1$ and $t = s$, we obtain a set $C_a \subset [\vcommitlength]$ of size $s$ such that
    \begin{align*}
        \P_{\subalign{\vcommitstring &\sim P_a\\b &\sim \mu(a)}}\big[b = (\vcommitstring_i,i) \text{ with } i \notin C_a] &=  \sum_{i \notin C_a}^\vcommitlength p_i r_i\\
        &= \norm{r}_1 \sum_{i \notin C_a}^\vcommitlength p_i q_i\\
        &\leq \frac{\norm{r}_1}{s}\\
        &= \frac{\sum_{i = 1}^\vcommitlength \P[b_1 = \vcommitstring_i ~|~ b_2 = i]}{s}\\
        &= o(1)
    \end{align*}
    whenever $\abs{P_a} \geq \frac{s \abs{\alphabet}^\vcommitlength}{\vcommitlength 2^s}$. Finally, take $C_a$ as given by the claim. Recall that the sets $P_a$ of size less than $\frac{s \abs{\alphabet}^\vcommitlength}{\vcommitlength 2^s}$ cover at most a $s/\vcommitlength = o(1)$ fraction of length-$\vcommitlength$ strings, so that the probability $\vcommitstring \sim \alphabet^\vcommitlength$ falls into the union of such sets is $o(1)$. In the complement of this event, we have
    \begin{align*}
        \P&\left[b(a) = (\vcommitstring_i,i) \text{ with } i \notin C_a ~\left|~ \abs{P_a} \geq \frac{s \abs{\alphabet}^\vcommitlength}{\vcommitlength 2^s}\right.\right]\\
        &= \frac1{\P_{\vcommitstring \sim \alphabet^\vcommitlength}\left[\abs{P_a} \geq \frac{s \abs{\alphabet}^\vcommitlength}{\vcommitlength 2^s}\right]} \sum_{\substack{a \in A\\\abs{P_a} \geq \frac{s \abs{\alphabet}^\vcommitlength}{\vcommitlength 2^s}}}\P_{\vcommitstring \sim \alphabet^\vcommitlength}[\vcommitstring \in P_a] \cdot \P_{\subalign{\vcommitstring &\sim  P_a\\b &\sim \mu(a)}}[b = (\vcommitstring_i,i) \text{ with } i \notin C_a]\\
        &= \frac1{1 - o(1)} \cdot o(1) = o(1),
    \end{align*}
    which concludes the proof.
\end{proof}

With the second step of our proof finished, we already have a nontrivial result by the known implication from hardness for one-way communication complexity: any streaming algorithm that streams a uniformly random string $\vcommitstring \in \alphabet^\vcommitlength$ \emph{and immediately outputs} a pair $(\alpha, i)$ has a small set $C \subset [\vcommitlength]$ capturing most of the probability that it outputs correctly. However, the verifier in our zero-knowledge streaming protocol will stream an \indexproblem instance between streaming $\vcommitstring$ and outputting a pair. To capture this behaviour, we define a (slight) variant of \pair and prove that the result of \cref{lem:pair-bound} carries over to it.

\begin{definition}
    \label{def:pair-x}
    For each string $x \in \alphabet^n$, let $\pair(x)$ denote the following one-way communication problem: Alice receives a string $\vcommitstring \sim \alphabet^\vcommitlength$ and sends Bob an $s$-bit message $a$; Bob reads $x$ and $a$ and outputs a pair $(\alpha, i) \in \alphabet \times [\vcommitlength]$. The protocol succeeds if $\alpha = \vcommitstring_i$.
\end{definition}

We have now reached the end goal of this section:
\begin{lemma}
    \label{lem:pair-x-bound}
     Fix a (single) one-way communication protocol for $\pair(x)$ for all $x \in \alphabet^n$ with alphabet size $32 \vcommitlength / \log\log \vcommitlength \leq \abs{\alphabet} = O(\vcommitlength / \log\log \vcommitlength)$ and message length $\log \vcommitlength \leq s = \polylog(\vcommitlength)$. Then, for any $x \in \alphabet^n$, there exists an event $E$ (that depends only on $\vcommitstring$) with $\P[E] = 1 - o(1)$ and a set $C$ of size $s$ (that depends only on Alice's message) satisfying
    \begin{equation*}
        \P\big[b(a,x) = (\vcommitstring_i,i) \text{ with } i \notin C_a \big| E \big] = o(1).
    \end{equation*}
\end{lemma}
\begin{proof}
    We will make a small adaptation in one of the steps of \cref{lem:pair-bound} to show there is a size-$s$ set $C$ \emph{independent of $x$} that captures most of the probability of Bob's correct outputs.

    Following the notation of \cref{lem:pair-bound}, $\set{P_a}$ is the partition induced by Alice's messages and Bob's output is a random variable $b(a(\vcommitstring), x) = b(a,x) \in \alphabet \times [\vcommitlength]$. We also denote the distribution of $b = b(a,x)$ by $\mu(a,x)$ and the conditional distribution of $b_1$ when $b_2 = i$ by $\mu_i(a,x)$.

    For every $x$ and $a$, the protocol's success probability conditioned on $\vcommitstring \in P_a$ is
    \begin{align*}
        \sum_{i = 1}^\vcommitlength \P_{\subalign{\vcommitstring &\sim P_a\\b &\sim \mu(a,x)}}[b = (\vcommitstring_i, i)] &= \sum_{i = 1}^\vcommitlength \P_{b \sim \mu(a,x)}[b_2 = i] \cdot \P_{\subalign{\vcommitstring &\sim P_a\\b &\sim \mu(a,x)}}[b_1 = \vcommitstring_i ~|~ b_2 = i]\\
        &= \sum_{i = 1}^\vcommitlength \P_{b \sim \mu(a,x)}[b_2 = i] \cdot \P_{\subalign{\vcommitstring &\sim P_a\\b_1 &\sim \mu_i(a,x)}}[b_1 = \vcommitstring_i].
    \end{align*}

    With $y = y(a) \in \alphabet^\vcommitlength$ as the string whose $i^\text{th}$ coordinate is the most frequent symbol at the $i^\text{th}$ coordinate in $P_a$, we know that $\P_{\vcommitstring \sim P_a}[\symbol1 = \vcommitstring_i]$ is maximal when $\symbol1 = y_i$. This holds also if $\symbol1$ is a random variable (independent from $\vcommitstring$), so that, in particular, with $r \in [0,1]^\vcommitlength$ defined by
    \begin{equation*}
        r_i \coloneqq \max_{x \in \alphabet^n} \set{\P_{\subalign{\vcommitstring &\sim P_a\\b_1 &\sim \mu_i(a,x)}}[b_1 = \vcommitstring_i]} \leq \P_{\vcommitstring \sim P_a}[y_i = \vcommitstring_i]
    \end{equation*}
    we have $\norm{r}_1 = o(s)$ when $\abs{P_a}$ is sufficiently large. Defining $p \in [0,1]^\vcommitlength$ by $p_i = \P_{b \sim \mu(a,x)}[b_2 = i]$, $p' \in [0,1]^\vcommitlength$ by $q_i = r_i/\norm{r}_1$ and using \cref{clm:probability-ip}, we obtain a set $C_a \subset [\vcommitlength]$ of size $s$ such that \emph{for every $x \in \alphabet^n$},
    \begin{align*}
        \P_{\subalign{\vcommitstring &\sim P_a\\b &\sim \mu(a,x)}}\big[b = (\vcommitstring_i,i) \text{ and } i \notin C_a] &= \sum_{i = 1}^\vcommitlength \P_{b \sim \mu(a,x)}[b_2 = i] \cdot \P_{\subalign{\vcommitstring &\sim P_a\\b_1 &\sim \mu_i(a,x)}}[b_1 = \vcommitstring_i]\\
        &\leq \norm{r}_1 \sum_{i \notin C_a}^\vcommitlength p_i q_i = o(1),
    \end{align*}
    and we conclude with same calculation of \cref{lem:pair-bound}.
\end{proof}

As an immediate corollary (by taking $C$ to be a set of symbol-coordinate pairs, rather than only coordinates; and setting, say, $C = \varnothing$ in the complement of the event $E$), we have:
\begin{theorem}
    \label{thm:correct-set}
    Let $\alphabet$ be an alphabet of size $32 \vcommitlength / \log\log \vcommitlength \leq \abs{\alphabet} = \Theta(\vcommitlength/\log\log \vcommitlength)$ and fix $x \in \alphabet^n$. Let $\widetilde V$ be a streaming space-$s$ algorithm with $\log \vcommitlength \leq s = \polylog(\vcommitlength)$ that streams $\vcommitstring \sim \alphabet^\vcommitlength$ followed by $x$, and outputs a pair $(\alpha, i) \in \alphabet \times [\vcommitlength]$.

   There exists a set $C \subset \alphabet \times [\vcommitlength]$ of size $s$, determined by the snapshot of $\widetilde V$ at the end of the stream $\vcommitstring$, such that
    \begin{equation*}
        \P\left[\widetilde V(\vcommitstring, x) \text{ outputs } (\vcommitstring_i, i) \notin C\right] = o(1).
    \end{equation*}
\end{theorem}

The theorem above attains what we set out for in this section: since $\widetilde V$ cannot remember many pairs $(\vcommitstring_i, i)$, we may prepend to any protocol a step where $P$ sends $\vcommitstring$ to the verifier. Then, whenever $\widetilde V$ sends an allegedly random $\symbol1 \in \alphabet$ to the prover, we ask that it \emph{also} send the coordinate $i$ such that $\symbol1 = \vcommitstring_i$ as evidence that $\symbol1$ was indeed sampled in the past, i.e., before it finished streaming $\vcommitstring$. In other words, this step provides a \emph{temporal commitment} by means of which $\widetilde V$ can show that its internal randomness is uncorrelated with the input. 

\section{A zero-knowledge SIP for polynomial evaluation}
\label{sec:pep}

Our goal in this section will be to combine the components constructed in \cref{sec:pv-algebraic-commitment,sec:vp-commitment} -- \emph{algebraic} and \emph{temporal} commitment protocols -- into a zero-knowledge protocol for \pep. It is useful to keep in mind that \pep is a generalisation of \indexproblem, and thus a protocol for the former yields one for the latter; in other words, for concreteness one may replace \pep by \indexproblem throughout this section. A formal definition of \pep follows.

\begin{definition}
    \label{def:pep}
    Let $\fe1 \in \field$ and $f = \set{f^x : x \in \alphabet^n}$ be a mapping such that $f^x: \field^\dimension \to \field$ is a degree-$\degree$ polynomial. $\pep(f, \fe1)$ is the language $\set{(x, \evaluationpoint) \in \alphabet^n \times \field^\dimension : f^x(\evaluationpoint) = \fe1}$.
\end{definition}
We remark that the parameters of the problem generally increase as a function of $n$; in particular, the field size is always assumed to satisfy $\fieldsize = \abs{\field} = \omega(1)$.

\subsection{The protocol}

For any mapping $f$ and field element $\fe1$, \cref{prot:zk-pep} lays out $\textsf{zk-pep}(f, \fe1)$, our zero-knowledge SIP for $\pep(f, \fe1)$. \cref{thm:pep-correctness,thm:pep-zk} prove, respectively, the correctness (i.e., completeness and soundness) and the zero-knowledge properties of \textsf{zk-pep}.

The protocol uses commitment (sub)protocols to allow each party to only reveal key information after the other party gives evidence that it is being honest; this is achieved by interspersing the commit-decommit steps of one party with those of the other. More precisely, in the setup (\cref{step:pep-temporal-commit}) the verifier performs its (temporal) commitment; after the input is streamed (\cref{step:pep-input}), the prover makes its (algebraic) commitment in \cref{step:pep-algebraic-commit}. Then follow decommitments in the same order: verifier and prover decommit at \cref{step:pep-temporal-decommit,step:pep-algebraic-decommit}, respectively.

For ease of notation, we use $\field^\times$ to denote $\field \setminus \{0\}$, the multiplicative group of the field $\field$. Recall, moreover, that for a matrix $\pcommitstring$, we use $\hat \pcommitstring(\fingerprinttuple, \bm{\theta}) \in \field$ to denote an evaluation of the low-degree extension of the string $\bm{\theta} \cdot \pcommitstring$ over $\field$ (see \cref{sec:preliminaries}), and that $\bm{\chi}(\rfe1)$ denotes a vector of Lagrange polynomials (see \cref{sec:lde-pep}); in the following protocol, the vector contains all but the first point of the interpolating set $\{0\} \cup [\degree\dimension]$ for a univariate degree-$\degree\dimension$ polynomial over $\field$, i.e., $\bm{\chi}(\rfe1) = \big(\chi_i(\rfe1) : i \in [\degree\dimension] \big) \in \field^{\degree\dimension}$.

\iflipics
\begin{figure}
\else
\begin{protocol}[label={prot:zk-pep}]{$\textsf{zk-pep}(f, \fe1)$}
\fi
	\textbf{Input:} Explicit access to $\field$, element $\fe1 \in \field$, degree $\degree$, dimension $\dimension$ and a mapping $x \mapsto f^x$; streaming access to $x\in\alphabet^n$ followed by $\evaluationpoint \in \field^\dimension$.\\
	
	\textbf{Parameters:}
	\begin{itemize}[label={}, nosep, itemindent=-1.5em]
	    \item Field size $\fieldsize = \abs{\field}$ satisfying $\degree \dimension = o(\fieldsize)$;
	    \item Commitment lengths $\vcommitlength = \fieldsize^\dimension (\log \dimension + \log\log \fieldsize)/32$ and $\pcommitlength = \dimension(\degree\dimension\fieldsize)^3$;
	\end{itemize}
    \iflipics\else
	\tcbline
    \vspace*{-.5\baselineskip}
    \fi
    
    \begin{steps}[wide]
        \setcounter{stepsi}{-1}
        \item\label{step:pep-temporal-commit} Temporal commitment
	    \begin{enumerate}[label={}, leftmargin=3em]
	        \item[$\bm P$:] Send a string $\vcommitstring \sim \big(\field^\dimension\big)^\vcommitlength$.
	        \item[$\bm V$:] Sample $\fingerprinttuple \sim \field^\dimension$ and stream $\vcommitstring$. For each $i$, check if $z_i = \fingerprinttuple$ and store $\ell \coloneqq i$ if so.
	        
	        Reject if $\fingerprinttuple \neq \vcommitstring_i$ for all $i \in [\vcommitlength]$.
	    \end{enumerate}
	    \iflipics\else%
        \vspace*{-\baselineskip}
        \tcbline
        \vspace*{-.5\baselineskip}
        \fi

	    \item\label{step:pep-input} Input streaming
	    \begin{enumerate}[label={}, leftmargin=3em]
	        \item[$\bm V$:] Stream $x$ and compute $f^x(\fingerprinttuple) \in \field$. Store $\evaluationpoint \in \field^\dimension$.
	        
	        If $\fingerprinttuple = \evaluationpoint$, check that $f^x(\fingerprinttuple) = \fe1$, accepting if so and rejecting otherwise.
        \end{enumerate}
        \iflipics\else%
        \vspace*{-\baselineskip}
        \tcbline
        \vspace*{-.5\baselineskip}
        \fi

	    \item\label{step:pep-algebraic-commit} Algebraic commitment
	    \begin{enumerate}[label={}, leftmargin=3em]
	        \item[$\bm V$:] Sample $\rfe1 \sim \field^\times \setminus [\degree\dimension]$ and send the line $\line: \field \to \field^\dimension$ with $\line(0) = \evaluationpoint$ and $\line(\rfe1) = \fingerprinttuple$.
	        
	        \item[$\bm P$:] Send an algebraic commitment $(\pcommitstring, \bm{\fe3}, k)$ to $f^x_{|\line}$, i.e., $(\pcommitstring, k) \sim \field^{\degree\dimension \times \pcommitlength} \times [\pcommitlength]$ and $\bm{\fe3} \in \field^{\degree\dimension}$ with $\bm{\fe3}_i = f^x_{|\line}(i) - \pcommitstring_{ik}$ for all $i \in [\degree\dimension]$.
	        
	       \item[$\bm V$:] Sample $\Fingerprinttuple \sim \field^\dimension$ and, while streaming $\pcommitstring$, compute $\hat \pcommitstring\big(\Fingerprinttuple, \bm{\chi}(\rfe1)\big)$.
	       
	       Compute the correction $\fe3 = \bm{\chi}(\rfe1) \cdot \bm{\fe3}$ and save (the identification of) $k \in \field^\dimension$.
	       \end{enumerate}
	    \iflipics\else%
        \vspace*{-\baselineskip}   
        \tcbline
        \vspace*{-.5\baselineskip}
        \fi
        
	    \item\label{step:pep-temporal-decommit} Temporal decommitment
	    \begin{enumerate}[label={}, leftmargin=3em]
	        \item[$\bm V$:] Send $\fingerprinttuple$ and $\ell$.
            \item[$\bm P$:] Check that $\vcommitstring_\ell = \fingerprinttuple \in \line$ and $\rfe1 \coloneqq \line^{-1}(\fingerprinttuple) \notin \{0\} \cup [\degree\dimension]$, aborting otherwise.
        \end{enumerate}
        \iflipics\else%
        \vspace*{-\baselineskip}
        \tcbline
        \vspace*{-.5\baselineskip} 
        \fi

        \item\label{step:pep-algebraic-decommit} Algebraic decommitment
        \begin{enumerate}[label={}, leftmargin=3em]
	        \item[$\bm V$:] Run $\textsf{decommit}\big(f^x(\fingerprinttuple) - \chi_0(\rfe1) \fe1, \bm{\chi}(\rfe1) \cdot \pcommitstring, k\big)$, with correction $\fe3$ and fingerprint $\hat \pcommitstring(\Fingerprinttuple, \bm{\chi}(\rfe1))$. Accept if \textsf{decommit} accepts and reject otherwise.
        \end{enumerate}
    \end{steps}
\iflipics
\caption{Protocol $\textsf{zk-pep}(f, \fe1)$}
\label{prot:zk-pep}
\end{figure}
\else
\end{protocol}
\fi

\subsection{Analysis of the protocol}

We now show that \textsf{zk-pep} is a valid (i.e., complete and sound) streaming interactive proof, as well as compute its space and communication complexities.

\begin{theorem}
    \label{thm:pep-correctness}
    Let $f$ be such that an evaluation of the $\field_\fieldsize$-polynomial $f^x$ can be computed by streaming $x$ in $O(\dimension \log \fieldsize)$ space. Then, for any $\fe1 \in \field_\fieldsize$, \cref{prot:zk-pep} is an SIP for $\pep(f, \fe1)$ with $s = O(\dimension \log \fieldsize)$ space complexity. Its communication complexity is $O(\fieldsize^\dimension \dimension \log^2 \fieldsize)$ in the setup and $O(\degree^4 \dimension^5 \fieldsize^3 \log \fieldsize)$ in the interactive phase.
\end{theorem}

\begin{proof}
    We will prove completeness then soundness, and compute the complexities last.
    
    \paragraph*{Completeness\iflipics\else.\fi} The verifier only aborts in \cref{step:pep-temporal-commit} (the setup) if $\fingerprinttuple$ is not among the $\vcommitlength > \fieldsize^\dimension \log\log \fieldsize$ random tuples sent by the prover, an event with probability $(1 - 1/\fieldsize^\dimension)^\vcommitlength \leq e^{-\vcommitlength/\fieldsize^\dimension} = o(1)$. Otherwise, since the prover behaves honestly, in \cref{step:pep-algebraic-commit} (the algebraic commitment) we have
    \begin{equation*}
        \pcommitstring_{ik} = f^x_{| \line}(i) - \bm{\fe3}_i
    \end{equation*}
    for all $i \in [\degree \dimension]$.
    
    Let $w = \bm{\chi}(\rfe1) \cdot \pcommitstring = \sum_{i = 1}^{\degree\dimension} \chi_i(\rfe1) \pcommitstring_i \in \field^\pcommitlength$ and $\hat w: \field^\dimension \to \field$ be its $\dimension$-variate LDE. Recall that, in $\textsf{decommit}\big(f^x(\fingerprinttuple) - \chi_0(\rfe1) \fe1, w, k\big)$ (\cref{prot:decommit}), with correction $\fe3$ and fingerprint $\hat \pcommitstring(\Fingerprinttuple, \bm{\chi}(\rfe1))$, the verifier sends a line $\line': \field \to \field^\dimension$ with $\line'(0) = k$, $\line'(\rfe2) = \Fingerprinttuple$, receives $\hat w_{|\line'}$ and makes two checks: that $\hat w_{|\line'}(\rfe2)$ matches the fingerprint and that $\hat w(0) + \fe3 = f^x(\fingerprinttuple) - \chi_0(\rfe1) \fe1$. Since
    \begin{equation*}
        \hat w_{|\line'}(\rfe2) = \hat w(\Fingerprinttuple) = \sum_{i = 1}^{\degree\dimension} \chi_i(\rfe1) \hat \pcommitstring_i(\Fingerprinttuple) = \hat \pcommitstring\big(\Fingerprinttuple,  \bm{\chi}(\rfe1)\big)
    \end{equation*}
    and
    \begin{align*}
        \hat w(0) + \fe3 = w_k + \fe3 &= \sum_{i = 1}^{\degree \dimension} \chi_i(\rfe1) \big(\pcommitstring_{ik} + \bm{\fe3}_i\big)\\
        &= \sum_{i = 1}^{\degree \dimension} \chi_i(\rfe1) f^x_{|\line}(i)\\
        &= f^x(\fingerprinttuple) - \chi_0(\rfe1) f^x(\evaluationpoint)\\
        &= f^x(\fingerprinttuple) - \chi_0(\rfe1) \fe1,
    \end{align*}
    the verifier accepts when $P$ is honest except with probability $o(1)$.

    \paragraph*{Soundness\iflipics\else.\fi} First, note that if $\fingerprinttuple \notin \set{\vcommitstring_i : i \in [\vcommitlength]}$, the verifier rejects already in \cref{step:pep-temporal-commit}. We can thus assume the tuple $\fingerprinttuple$ equals some coordinate in $\vcommitstring$, and, since the string and tuple are independent random variables, the distribution of $\fingerprinttuple$ is still uniform conditioned on this event. (We may also assume that $\fingerprinttuple \neq \evaluationpoint$, since otherwise $V$ also rejects regardless of the prover's behaviour.)

    The only other point where $V$ may reject is \cref{step:pep-algebraic-decommit} (the algebraic decommitment). Once again, recall that $V$ sends the prover a line $\line'$ with $\line'(0) = k$, $\line'(\rfe2) = \Fingerprinttuple$ where $\rfe2 \sim \field$ and $P$ replies with a degree-$\degree\dimension$ polynomial $g: \field \to \field$ that is allegedly $\hat w_{|\line'}$. The verifier then checks that $g(\rfe2) = \hat \pcommitstring\big(\Fingerprinttuple, \bm{\chi}(\rfe1)\big) = \hat w_{|\line'}(\rfe2)$ and $g(0) + \fe3 = f^x(\fingerprinttuple) - \chi_0(\rfe1) \fe1$, rejecting if either equality fails to hold.
    
    We now analyse three cases: first, suppose that $g = \hat w_{|\line'}$. Then the first check passes but
   \begin{align*}
        g(0) + \fe3 &= w_k + \fe3\\
        &= f^x(\fingerprinttuple) - \chi_0(\rfe1) f^x(\evaluationpoint)\\
        &\neq f^x(\fingerprinttuple) - \chi_0(\rfe1) \fe1,
   \end{align*}
   so the verifier rejects (with probability 1).

    Suppose, now, that $g(0) \neq \hat w_{|\line'}(0)$. Then \cref{lem:sz} (Schwartz-Zippel) implies $g(\rfe2) \neq \hat w_{|\line'}(\rfe2)$, so the verifier rejects, except with probability $\degree\dimension/\fieldsize = o(1)$.
    
    Finally, suppose that $g \neq \hat w_{|\line'}$ but $g(0) = \hat w(0) = \sum_{i = 1}^{\degree\dimension}\chi_i(\rfe1) \pcommitstring_{ik}$. Then either the first check fails, i.e., $g(\rfe2) \neq \hat \pcommitstring\big(\Fingerprinttuple, \bm{\chi}(\rfe1)\big)$, and $V$ rejects; or $g(\rfe2) = \hat \pcommitstring\big(\Fingerprinttuple, \bm{\chi}(\rfe1)\big)$, and the second check passes if
    \begin{align*}
        g(0) + \fe3 &= \sum_{i = 1}^{\degree\dimension}\chi_i(\rfe1)\big( \pcommitstring_{ik} + \bm{\fe3}_i\big)
    \end{align*}
    is equal to
    \begin{equation*}
        f^x(\fingerprinttuple) - \chi_0(\rfe1) \fe1 = \chi_0(\rfe1) \big(f^x(\evaluationpoint) - \fe1\big) + \sum_{i = 1}^{\degree\dimension}\chi_i(\rfe1) f^x_{|\line}(i).
    \end{equation*}
    Rearranging, the second check corresponds to the following equation:
    \begin{equation*}
        \chi_0(\rfe1) \big(f^x(\evaluationpoint) - \fe1\big) + \sum_{i = 1}^{\degree\dimension} \chi_i(\rfe1)\left(f^x_{|\line}(i) - \bm{\fe3}_i - \pcommitstring_{ik}\right) = 0.
    \end{equation*}
    
    Now, consider the left-hand side of the equation as a polynomial in $\rfe1$: plugging in $0$ for the variable $\rfe1$ evaluates to $f^x(\evaluationpoint) - \fe1 \neq 0$, so that it is a nonzero polynomial; and, crucially, $\rfe1$ was sampled uniformly (from $\field^\times \setminus [\degree\dimension]$) and independently of the communication (in particular, of $\pcommitstring$ and $\bm{\fe3}$) by $V$. By \cref{lem:sz} lemma once again, the equation is satisfied with probability at most $\degree\dimension/(\fieldsize - \degree\dimension - 1) = o(1)$ and soundness follows.

    \paragraph*{Communication complexity\iflipics\else.\fi} Most of the communication occurs in \cref{step:pep-temporal-commit,step:pep-algebraic-commit} (the commitments), which communicate
    \begin{align*}
        O\big(\fieldsize^\dimension (\log \dimension + \log\log \fieldsize) \dimension \log \fieldsize\big) &= O\left(\fieldsize^{\dimension} \dimension \log^2 \fieldsize\right) \qquad \text{and}\\
        O\left(\pcommitlength \degree\dimension \log \fieldsize\right) &= O\left(\degree^4\dimension^5\fieldsize^3 \log \fieldsize\right)
    \end{align*}
    bits, respectively. (The communication in other steps is significantly smaller: \cref{step:pep-input} has none, while \cref{step:pep-temporal-decommit,step:pep-algebraic-decommit} communicate $\dimension \log \fieldsize + \log \vcommitlength = O(\dimension \log \fieldsize)$ and $O(\degree\dimension \log \fieldsize)$ bits, respectively.)

    \paragraph*{Space complexity\iflipics\else.\fi} Apart from a constant number of elements of $\field$ (requiring $O(\log \fieldsize)$ bits), the verifier stores $\ell \in [\vcommitlength]$, $k \in [\pcommitlength]$ and $\fingerprinttuple, \Fingerprinttuple \in \field^\dimension$. Since $\vcommitlength \geq \pcommitlength$, the space complexity is dominated by $\ell$ and $\fingerprinttuple, \Fingerprinttuple$. Since storing $\ell$ requires $\log \vcommitlength = O(\dimension \log \fieldsize)$ bits (as does computing $f^x(\fingerprinttuple)$) while $\fingerprinttuple$ and $\Fingerprinttuple$ require $\dimension \log \fieldsize$ bits each, the space complexity follows.
\end{proof}

\subsection{Zero-knowledge}

Having shown that \textsf{zk-pep} is a valid streaming interactive proof, we now show it is also zero-knowledge.

\begin{theorem}
    \label{thm:pep-zk}
    \cref{prot:zk-pep} is zero-knowledge, secure against distinguishers with space $\degree\dimension^2 \polylog(\fieldsize)$. The simulator runs in $O\big((\degree + \dimension \log \fieldsize) \dimension \log \fieldsize\big) = O\big(\degree \dimension^2 \log^2 \fieldsize\big)$ space.
\end{theorem}
\begin{proof}
    Recall that an SIP with a space-$s$ verifier is zero-knowledge against $\degree\dimension^2 \polylog(\fieldsize)$-space distinguishers if there exists a streaming simulator $S$ that satisfies the following. For any space-$s$ (honest or malicious) verifier $\widetilde V$ and input $(x, \evaluationpoint)$ where $f^x(\evaluationpoint) = \fe1$, given whitebox access to $\widetilde V$ the simulator $S$ produces a view that is indistinguishable to a $\degree\dimension^2 \polylog(\fieldsize)$-space (streaming) algorithm from the view generated by an interaction of $\widetilde V$ with the honest prover. Note that $\widetilde V$ can be simulated in space $O(s)$, so the space complexity of the statement suffices to simulate the verifier of \cref{prot:zk-pep} since  $s = O(\dimension \log \fieldsize)$.

    The simulator interprets its read-only random bit string as $(\vcommitstring, \pcommitstring)$ with $\vcommitstring \sim \field^\vcommitlength$ and $\pcommitstring \sim \field^{\degree \dimension \times \pcommitlength}$ (so that $\vcommitlength \dimension \log \fieldsize + \pcommitlength \degree \dimension \log \fieldsize \leq \fieldsize^{\dimension + 8}$ bits suffice and an algorithm with $(\dimension + 8) \log \fieldsize$ space can address into this string). This pair will be used to simulate prover messages, whereas the simulation of $\widetilde V$ will use a source of randomness that cannot be reread (but has unbounded length). In the description that follows, as well as the more succinct one in \cref{alg:pep-simulator}, recall that $\widetilde V$ is assumed to only output a decision at the end of the protocol (so that, if it decides to reject in the middle, it continues the protocol with dummy messages); and likewise if $S$ (or $P$) aborts.

    In the setup, \cref{step:pep-simulator-temporal-commit} (the temporal commitment), $S$ simulates $\widetilde V(\vcommitstring)$. Then, using the snapshot of the verifier's memory and its whitebox access to $\widetilde V$, the simulator finds the set $C$ of $s$ elements of $\field^\dimension$ that $\widetilde V$ may successfully decommit to with the largest probabilities. More precisely, $S$ calls the whitebox oracle $\whiteboxoracle$ (see \cref{def:whitebox-max}) on the algorithm that corresponds to the verifier immediately before streaming $x$, with initial memory state equal to the current snapshot $\snapshot \in \bitset^s$, and whose output is a pair $(\fingerprinttuple, \ell)$ at \cref{step:simulator-pep-temporal-decommit} (ignoring $\line$, the intermediate output at \cref{step:simulator-pep-algebraic-commit}).
    
    $S$ initialises a(n empty) sorted list of message-probability pairs in $\field^\dimension \times [\vcommitlength] \times [0,1]$, and, for all $\ell \in [\vcommitlength]$, uses its oracle access to both $\vcommitstring$ and $\whiteboxoracle$ to find $\mu_\ell \coloneqq \whiteboxoracle\big(\snapshot,(\vcommitstring_\ell, \ell)\big)$. If the size of the list is smaller than $s$, or $\mu_\ell$ is larger than the smallest probability in it, $S$ adds $(z_\ell, \ell, \mu_\ell)$ to it (and removes the tuple with the smallest $\mu_{\ell'}$ if the size of the resulting would have exceeded $s$). 
    
    This yields the set $C \subset \field^\dimension \times [\vcommitlength]$ with the $s$ most likely correct decommitments of $\widetilde V$. Since the string $\vcommitstring$ is over the alphabet $\field^\dimension$, whose size satisfies
    \begin{equation*}
        \frac{\vcommitlength}{\log\log \vcommitlength} = \frac{\fieldsize^\dimension (\log \dimension + \log\log \fieldsize)}{32 \log \big(\dimension \log \fieldsize + \log(\log \dimension + \log\log \fieldsize) - 5\big)} \leq \frac{\fieldsize^\dimension}{32},
    \end{equation*}
    $\fieldsize^\dimension = \Theta(\vcommitlength / \log\log \vcommitlength)$ as well as $s \geq \log \pcommitlength = \Theta(\log \fieldsize)$ and $s = \polylog(\pcommitlength)$, \cref{thm:correct-set} applies for this parameter setting. This ensures that, except with probability $o(1)$, the verifier $\widetilde V$ will output either $(z_\ell, \ell) \in C$ or an incorrect $(\fingerprinttuple, \ell)$ with $\vcommitstring_\ell \neq \fingerprinttuple$ in its decommitment at \cref{step:simulator-pep-temporal-decommit}. 
    
    Then $S$ proceeds to \cref{step:pep-simulator-input}, where it simulates $\widetilde V(x)$ and, with $F \coloneqq \set{\vcommitstring_i : (\vcommitstring_i, i) \in C}$, computes $f^x(\fingerprinttuple)$ \emph{for every} $\fingerprinttuple \in F$.
    At the start of \cref{step:simulator-pep-algebraic-commit} (the algebraic commitment), $\widetilde V$ sends a line $\line$. The simulator inspects the intersection of $L$ (viewed as a set) with the set of fingerprints $F$ and computes a random degree-$\degree\dimension$ polynomial $g$ subject to the constraints $g(\fe2) = f^x_{|\line}(\fe2) = f^x\big(\line(\fe2)\big)$ for all $\fe2 \in \line^{-1}(F)$.\footnote{Knowledge of $f^x(\fingerprinttuple)$ for all $\fingerprinttuple \in F$ enables the simulator to sample from this distribution: $F$ fixes $\abs{\line \cap F}$ evaluations, and the simulator sets the $\degree\dimension - \abs{\line \cap F}$ remaining ones uniformly.} Note that the description of $g$ is comprised of $O(\degree\dimension)$ field elements.
    
    $S$ samples $k \sim [\pcommitlength]$ then simulates $\widetilde V$ streaming $\pcommitstring$ followed by $\bm{\fe3}_i = g(i) - \pcommitstring_{ik}$ for all $i \in [\degree \dimension]$ and $k$; note that $S$ is able to compute all $\bm{\fe3}_i$ from the description of $g$ combined with its oracle access to $\pcommitstring$.

    There is no prover-to-verifier communication in \cref{step:simulator-pep-temporal-decommit} (the temporal decommitment), so $S$ simulates $\widetilde V$ until the verifier sends a tuple $\fingerprinttuple \in \field^\dimension$ and an index $\ell \in [\vcommitlength]$. The simulator then checks that $\vcommitstring_\ell = \fingerprinttuple \in \line$ and $\rfe1 \coloneqq \line^{-1}(\fingerprinttuple) \in \field^\times \setminus [\degree\dimension]$; if not, then $S$ aborts (as $P$ would).
    
    Finally, in \cref{step:pep-simulator-algebraic-decommit} (the algebraic decommitment), $S$ simulates $\widetilde V$ until it sends a line $\line': \field \to \field^\dimension$. The only remaining part of the verifier's view left to generate are the evaluations of  of the polynomial $\sum_{i \in [\degree\dimension]} \chi_i(\rfe1) \hat \pcommitstring_i \circ \line'$ for all points in $[\degree\dimension]$. These are computed by $S$ in a streaming fashion using its oracle access to $\pcommitstring$.

    The space complexity of $S$ is dominated by the description of the polynomial $g$, which requires $O(\degree\dimension \log \fieldsize)$ bits, and by the set $C$ of $s = O(\dimension \log \fieldsize)$ elements of $\field^\dimension \times [\vcommitlength]$. Since each element requires $\dimension \log \fieldsize + \log \vcommitlength = O(\dimension \log \fieldsize)$ bits, the total space complexity is
    \begin{equation*}
        O(\degree\dimension \log \fieldsize + s\dimension \log \fieldsize) = O\left((\degree + \dimension \log \fieldsize) \dimension \log \fieldsize\right) = O(\degree \dimension^2 \log^2 \fieldsize),
    \end{equation*}
    as claimed. (Apart from $C$, the simulator stores $f^x(\fingerprinttuple) \in \field$ for every $\fingerprinttuple \in C$, which requires $s \log \fieldsize$ bits; and the lines $\line$, $\line'$ as well as $k$, which require $O(\log \fieldsize)$ bits each.) 
    
    \iflipics
    \begin{figure}
    \else
    \begin{algorithm}[label={alg:pep-simulator}]{Simulator for \cref{prot:zk-pep}}
    \fi
	    \textbf{Input:} Whitebox access to $\widetilde V$; oracle access to a length-$\fieldsize^{\dimension + 8}$ random bit string interpreted as $(\vcommitstring, \pcommitstring) \in \big(\field^\dimension\big)^\vcommitlength \times \field^{\degree\dimension \times \pcommitlength}$; streaming access to $(x, \evaluationpoint) \in \alphabet^n \times \field^\dimension$.\\
	    
	    \textbf{Output:} View $\big(\vcommitstring, x, \evaluationpoint, \pcommitstring, \bm{\fe3}, k, (h(i) : i \in [\degree\dimension])\big)$ with $k \in [\pcommitlength]$, $\bm{\fe3} \in \field^{\degree\dimension}$ and $h: \field \to \field$.

        \iflipics\else%
        \tcbline
        \vspace*{-.5\baselineskip}
        \fi%
        \begin{steps}[wide]
            \setcounter{stepsi}{-1}
	        \item\label{step:pep-simulator-temporal-commit} Temporal commitment
	        \begin{enumerate}[label={}, leftmargin=3em]
	            \item[$\bm S$:] Send $\vcommitstring$.
	            
	            \item[$\bm{\widetilde V}$:] Simulate until the end of this step and let $\snapshot \in \bitset^s$ be the resulting snapshot of $\widetilde V$. Use the whitebox oracle $\whiteboxoracle$ to determine the set $C \subset \set{(\vcommitstring_i, i) : i \in [\vcommitlength]}$ of size $s$ with the largest $\whiteboxoracle(\snapshot, (\vcommitstring_i, i))$. 
	        \end{enumerate}
            \iflipics\else%
        	\vspace*{-\baselineskip}   
            \tcbline
            \vspace*{-.5\baselineskip}
            \fi  
            
	        \item \label{step:pep-simulator-input} Input streaming
	        \begin{enumerate}[label={}, leftmargin=3em]
	            \item[$\bm{\widetilde V}$:] Stream $x$, computing and storing $f^x(\fingerprinttuple)$ for every $\fingerprinttuple \in \set{\vcommitstring_i : (\vcommitstring_i, i) \in C}$ while simulating the verifier.
	            
	            \item[$\bm S$:] Store $\evaluationpoint$.
            \end{enumerate}
            \iflipics\else%
        	\vspace*{-\baselineskip}   
            \tcbline
            \vspace*{-.5\baselineskip} 
            \fi
            
	        \item \label{step:simulator-pep-algebraic-commit} Algebraic commitment

	        \begin{enumerate}[label={}, leftmargin=3em]
	            \item[$\bm{\widetilde V}$:] Simulate until $\widetilde V$ sends a line $\line$, aborting if $\line(0) \neq \evaluationpoint$.
	            
	            \item[$\bm S$:] Sample a random polynomial $g: \field \to \field$ of degree at most $\degree\dimension$ subject to $g(0) = \fe1$ and $g(\fe2) = f^x\big(\line(\fe2)\big)$ for all $\fe2$ such that $(i, \line(\fe2)) \in C$ for some $i \in [\vcommitlength]$.
	            
	            Send $\pcommitstring$ followed by $\bm{\fe3} = \big(g(i) - \pcommitstring_{ik} : i \in [\degree\dimension]\big)$ and $k \sim [\pcommitlength]$.
	            
	            \item[$\bm{\widetilde V}$:] Simulate until the end of the step.
	           \end{enumerate}
            \iflipics\else%
        	\vspace*{-\baselineskip}   
            \tcbline
            \vspace*{-.5\baselineskip}
            \fi 
            
	        \item \label{step:simulator-pep-temporal-decommit} Temporal decommitment
	        \begin{enumerate}[label={}, leftmargin=3em]
	            \item[$\bm{\widetilde V}$:] Simulate until $\widetilde V$ sends $\fingerprinttuple \in \field^\dimension$ and $\ell \in [\vcommitlength]$.
	            
	            \item[$\bm S$:] Check that $\vcommitstring_\ell = \fingerprinttuple \in \line$ and $\rfe1 \in \field^\times \setminus [\degree\dimension]$, aborting if either check fails or $(\fingerprinttuple, \ell) \notin C$.
            \end{enumerate}
            \iflipics\else%
        	\vspace*{-\baselineskip}   
            \tcbline
            \vspace*{-.5\baselineskip} 
            \fi
            
            \item \label{step:pep-simulator-algebraic-decommit} Algebraic decommitment
            \begin{enumerate}[label={}, leftmargin=3em]
	            \item[$\bm{\widetilde V}$:] Simulate until $\widetilde V$ sends a line $\line': \field \to \field^\dimension$, aborting if $\line'(0) \neq k$.
	            
	            \item[$\bm S$:] Set $\rfe1 \coloneqq \line^{-1}(\fingerprinttuple)$ and send $\left(\sum_{i = 1}^{\degree\dimension} \chi_i(\rfe1) \cdot \hat \pcommitstring_i \circ \line'(j) : j \in [\degree\dimension]\right)$.
	        \end{enumerate}
	    \end{steps}
    \iflipics
    \caption{Simulator for \cref{prot:zk-pep}}
    \label{alg:pep-simulator}
    \end{figure}
    \else
    \end{algorithm}
    \fi

    Now, all that remains is to prove indistinguishability by space-$s'$ streaming algorithms between the output of $S$ and a real transcript, for some $s'$ comparable to the space complexities of the verifier and simulator. The following claim proves this with $s' = \degree\dimension^2 \polylog(\fieldsize)$ (which is larger than both).

    \begin{claim}
        \label{clm:indistinguishability}
        Fix $\fe1 \in \field$ and $f$ as in the definition of \pep, an input $(x, \evaluationpoint) \in \field^n \times \field^\dimension$, a bit string $r$ of arbitrary length and a $O(\dimension \log \fieldsize)$-space verifier algorithm $\widetilde V$. Let $D$ be a streaming algorithm with space $\degree\dimension^2 \polylog(\fieldsize)$ such that
        \begin{equation*}
            \P\left[D\left(\view_{P,\widetilde V}(x, r)\right) \text{ accepts}\right] - \P\left[D\left(S\left(\widetilde V, x, r\right)\right) \text{ accepts}\right] = \eps,
        \end{equation*}
        with $\view_{P,\widetilde V}(x, r)$ a view of \cref{prot:zk-pep} and $S\left(\widetilde V, x, r\right)$ output by \cref{alg:pep-simulator}. Then $\eps = o(1)$.
    \end{claim}
    Assume, towards contradiction, that there exist $\fe1$, $f$, an input $(x, \evaluationpoint) \in \field^n \times \field^\dimension$, a streaming verifier $\widetilde V$ with $O(\dimension \log \fieldsize)$ space and a (streaming) distinguisher $D$ with $\degree\dimension^2 \polylog(\fieldsize)$ space such that $D$ distinguishes real transcripts of $\textsf{zk-pep}(f, \fe1)$ from simulations with bias $\eps = \Omega(1)$ when the input is $(x, \evaluationpoint)$.
    
    Recall that we assume that $\widetilde V$ rejects only after receiving all messages from $P$; therefore, the algebraic commitment $(\pcommitstring, \bm{\fe3}, k)$ is always present in both real and simulated views. Our goal is to show $D$ implies a one-way protocol for \indexproblem over the binary alphabet with a small message and a large bias, using \cref{lem:string-to-bit-index}. We do so by constructing, from $D$, a one-way communication protocol that distinguishes algebraic commitments to a fixed message $\bm{\fe1} \in \field^\ell$ from algebraic commitments to a random $\bm{\fe1}' \in \field^\ell$, where $\ell \leq \degree \dimension$.
    
    As both the real and simulated transcripts are identically distributed up to (and including) the verifier's message in \cref{step:simulator-pep-algebraic-commit}, the expected distinguishing advantage and probability of a simulation failure (i.e., of an abortion in \cref{step:simulator-pep-temporal-decommit} due to $(\fingerprinttuple, \ell) \notin C$) are $\eps$ and $o(1)$, respectively (over $\vcommitstring$ and the bits of the verifier randomness $r$ used until then). Therefore, there exists a fixed prefix of the transcript that retains distinguishing advantage $\eps/2$ and whose probability of a simulation failure is $o(1)$; indeed, at least an $\eps/2$ fraction of prefixes retains advantage $\eps/2$ and at most an $o(1)$ fraction yields simulation failures with $\Omega(1)$ probability, so an $\eps/2 - o(1)$ fraction of prefixes work. We thus assume, in the one-way protocol we define next, not only $x$ and $\evaluationpoint$ to be fixed, but also the line $\line$ and $\vcommitstring$ -- and, consequently, the set $C \subset \set{(\vcommitstring_i, i) : i \in [\vcommitlength]}$ (as well as the corresponding $f^x(\vcommitstring_i)$) that captures most of the weight of correct tuples $\widetilde V$ may decommit to, as given by \cref{thm:correct-set}.

    Viewing $\line$ as the set of pairs $\set{(\line(\rfe2), \rfe2) : \rfe2 \in \field}$, define $\ell \coloneqq \degree\dimension - \abs{\line \cap C}$ and assume,\footnote{Note that when $\abs{\line \cap C} \geq \degree\dimension$ the simulator knows the entirety of $f^x_{|\line}$, in which case the distinguishing bias is $0$. Nonzero bias thus implies $\degree\dimension > \abs{\line \cap C}$.} without loss of generality, that $\line \cap C = [\degree\dimension] \setminus [\ell]$.
    Consider the following one-way communication protocol with shared randomness (for strings $w$ of length $\pcommitlength$) that distinguishes a commitment $(w, k, \bm{\eta})$ to $\big(f^x(i) : i \in [\ell]\big)$ from a commitment to a random message: Alice uses $S$ to simulate an interaction between $P$ and $\widetilde V$ with input $(x, \evaluationpoint)$ and verifier randomness $r$, executing $D$ on the (partial and fixed) transcript thus obtained, until $\widetilde V$ sends a line $\line: \field \to \field^\dimension$ in \cref{step:simulator-pep-algebraic-commit}.
    
    Alice samples $\rfe1' \sim \field^\times \setminus [\degree\dimension]$ and continues the simulation of $D$ by feeding it $\pcommitstring \in \field^{\degree\dimension \times \pcommitlength}$ defined as follows: $\pcommitstring_i \coloneqq w_i$ for $i \in [\ell]$, $\pcommitstring_i \sim \field^\pcommitlength$ for $\ell < i < \degree\dimension$ and 
    \begin{equation*}
        \pcommitstring_{\degree\dimension} \coloneqq \chi_{\degree\dimension}(\rfe1')^{-1} \cdot \left(t - \sum_{i = 1}^{\degree\dimension - 1} \chi_i(\rfe1') \pcommitstring_i\right),
    \end{equation*}
    where $\pcommitstring_{\ell + 1}, \ldots, \pcommitstring_{\degree\dimension - 1}$ and $t$ are random strings (in $\field^\pcommitlength$) shared with Bob. Note that $\rfe1' \notin \{0\} \cup [\degree\dimension]$ implies $\chi_{\degree\dimension}(\rfe1') \neq 0$, so that $\pcommitstring_{\degree\dimension}$ is well-defined. (See \cref{fig:index-reduction} for a diagram of the reduction.)
    
    After simulating $\widetilde V$, $D$ and $S$ in \cref{step:simulator-pep-algebraic-commit} with $\pcommitstring$, she sends Bob all three snapshots as well as $\line$ and $\rfe1'$ in a $\degree\dimension^2 \polylog(\fieldsize)$-bit message.\footnote{We assume Bob receives the tuple $\bm{\eta}$ and reads $C$ along with the corresponding evaluations from the simulator's snapshot; alternatively, Alice could send this information in a message that is asymptotically no larger.} (The space complexities of $\widetilde V$ and $S$ are both dominated by the distinguisher's.)
    
    \begin{figure}
	\centering
	\includegraphics[width=0.9\textwidth]{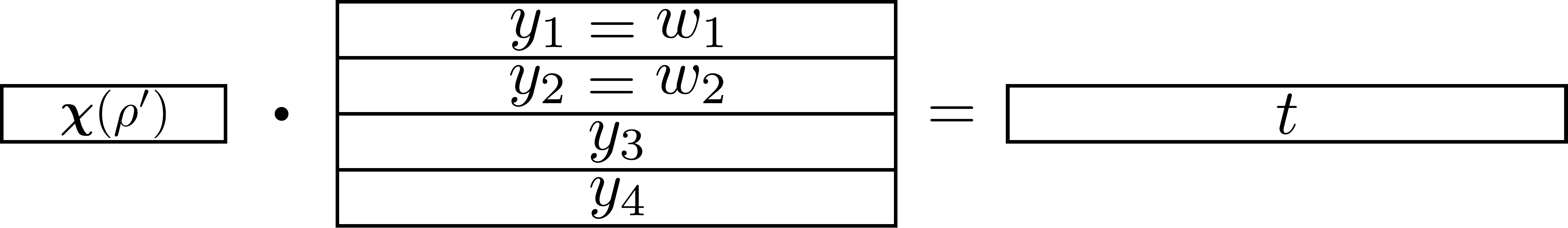}
	\caption{Reduction from \indexproblem to distinguishability of views when $\ell = 3$ and $\degree\dimension = 4$. The instance $w$ is inserted into the first $2$ rows of $\pcommitstring$, while $\pcommitstring_3$ is filled in with joint randomness and $\pcommitstring_4$ is the solution of the linear system shown in the diagram.}
    \label{fig:index-reduction}
    \end{figure}

    Bob, in turn, finishes the simulation of \cref{step:simulator-pep-algebraic-commit} with his (random) index $k \in [\pcommitlength]$ and the correction tuple $\bm{\fe3}$ defined as follows:\footnote{Recall that all $\pcommitstring_i$ for all $\ell < i < \degree\dimension$ are contained in Alice and Bob's shared randomness.}
    \begin{equation*}
        \bm{\fe3}_i = \left\{\begin{array}{ll}\bm{\eta}_i, & \text{if } i \leq \ell \\ \hat x(i) - \pcommitstring_{ik},\quad & \text{if } \ell < i < \degree\dimension\end{array}\right.
    \end{equation*}
    and
    \begin{equation*}
        \bm{\fe3}_{\degree\dimension} \coloneqq \chi_{\degree\dimension}(\rfe1')^{-1} \left(f^x_{|\line}(\rfe1') - \chi_0(\rfe1') \fe1 - t_k - \sum_{i = 1}^{\degree\dimension - 1} \chi_i(\rfe1') \bm{\fe3}_i\right).
    \end{equation*}
    Bob proceeds to simulate \cref{step:simulator-pep-temporal-decommit,step:pep-simulator-algebraic-decommit}, using $S$ to generate the remainder of the view. Note that in the former step $(\fingerprinttuple, \ell) \notin C$ is the only case in which $S$ aborts when $P$ would not, which identifies a simulated transcript with certainty; but this is a small-probability event. When $S$ fails (i.e., $(\fingerprinttuple, \ell) \notin C$) or the field element $\rfe1 = \line^{-1}(\fingerprinttuple)$ is not equal to $\rfe1'$, Bob halts the simulations and accepts or rejects uniformly; otherwise, he finishes the transcript by sending the low-degree polynomial that comprises the last round. This is possible because, while Bob does \emph{not} know all $\hat \pcommitstring_i$, he does know the required linear combination:
    \begin{align*}
        \sum_{i = 1}^{\degree\dimension} \chi_i(\rfe1) \cdot \pcommitstring_i &= \sum_{i = 1}^{\degree\dimension - 1} \chi_i(\rfe1) \cdot \pcommitstring_i + \chi_{\degree\dimension}(\rfe1) \cdot \chi_{\degree\dimension}(\rfe1)^{-1} \left(t - \sum_{i = 1}^{\degree\dimension - 1} \chi_i(\rfe1) \pcommitstring_i\right)\\
        &= t,
    \end{align*}
    and since $t$ is a (random) string known to both Alice and Bob, in particular he can compute $\hat t_{\line'}$ for any line $\line': \field^\dimension \to \field$.
    
    Finally, Bob inspects the output of $D$ and chooses his output accordingly, accepting if and only if $D$ accepts. Note that this one-way protocol succeeds
    \begin{itemize}[noitemsep]
        \item with probability $1/2$ (and thus bias $0$) either when $S$ fails or when $S$ succeeds and $\rfe1' \neq \rfe1$;
        \item with bias $\eps/2$ when $S$ succeeds and $\rfe1' = \rfe1$.
    \end{itemize}
    
    The latter follows from the fact that, if $S$ succeeds and $\rfe1' = \rfe1$, it produces a full transcript where $\bm{\fe3}$ is a correction for the (unique) degree-$\degree\dimension$ polynomial $g$ such that $g(0) = \fe1$, $g(i) = \bm{\eta}_i + \pcommitstring_{ik}$ for $i \in [\ell]$ and $g(i) = f^x(i)$ for $i \in [\degree\dimension] \setminus [\ell] = \line \cap C$.\footnote{When $\line \cap C \neq [\degree\dimension] \setminus [\ell]$, the set still fixes $\abs{\line \cap C}$ values of $g$ and leaves $\degree\dimension - \abs{\line \cap C}$ to be chosen randomly.} Therefore, if $\bm{\eta} = \big(f^x(i) - \pcommitstring_{ik} : i \in [\ell]\big)$, then $\bm{\fe3}$ is a correction to $f^x$; while if $\bm{\eta}$ is random, then $\bm{\fe3}$ is a random degree-$\degree\dimension$ polynomial that matches $f^x$ in ($0$ and) $\line \cap C$. Since $D$ distinguishes between the two cases with bias $\eps/2$, then so does the one-way protocol. Therefore,
    \begin{align*}
        \P_{\substack{w \sim \field^{\ell \times \pcommitlength}\\k \sim [\pcommitlength]}}&\left[B\left(A(w), \big(f^x(i) - w_{ik} : i \in [\ell]\big), k\right)\text{ accepts}\right]\\
        &- \P_{\substack{w \sim \field^{\ell \times \pcommitlength}\\k \sim [\pcommitlength]\\\bm{\eta} \sim \field^\ell}}\left[B\big(A(w), \bm{\eta}, k\big)\text{ accepts}\right]\\
        &= o(1) \cdot \left(\frac12 - \frac12\right) + \big(1 - o(1)\big) \cdot \left(1 - \frac1 \fieldsize\right) \cdot \left(\frac12 - \frac12\right) + \big(1 - o(1)\big) \cdot \frac1 \fieldsize \cdot \frac{\eps}{2}\\
        &\geq \frac{\eps}{3\fieldsize}~.
    \end{align*}
    
    Finally, applying \cref{lem:string-to-bit-index}, we conclude that there exists a one-way binary \indexproblem protocol for strings of length $\pcommitlength = \dimension(\degree \dimension \fieldsize)^3$ with messages of length $\frac{\degree\dimension^2 \fieldsize^2 \ell^2 \log^2 \fieldsize}{\eps^2} \polylog(\fieldsize) \leq \degree^3\dimension^4\fieldsize^{2.01}$ and constant bias, a contradiction with $\sqrt{\frac{\degree^3 \dimension^4 \fieldsize^{2.01}}{\pcommitlength}} = o(1)$.
\end{proof}

\begin{remark}
    \label{rem:indistinguishability-pep}
    Inspecting the proof of \cref{clm:indistinguishability}, we see that increasing the prover's commitment length allows us to achieve significantly stronger indistinguishability: with $\pcommitlength = \log^{\omega(1)} n$, we have $\sqrt{s' \fieldsize^2 \ell^2/\pcommitlength} = o(1)$ for any $s' = \polylog(n)$. This setting of $\pcommitlength$ increases only the communication complexity of the interactive phase (\cref{step:pep-algebraic-commit,step:pep-temporal-decommit,step:pep-algebraic-decommit}) -- which can still be bounded by $n^{o(1)}$ -- and makes the protocol secure against $\polylog(n)$-space distinguishers.
\end{remark}

\subsection{Applications: \indexproblem, \pointquery, \rangecount and \selection}
\label{sec:pep-applications}

From the general $\textsf{zk-pep}(f,\fe1)$ protocol, we immediately obtain a zero-knowledge streaming interactive proof for the $\decisionindex(\fe1)$ problem (\cref{def:decision-index}) as a corollary:
\begin{corollary}
    \label{cor:index}
    Fix $\delta \in (0,1]$. For any $\fe1 \in \field_\fieldsize$ where $\fieldsize = \Theta\left(\log^{1 + \frac{2}{\delta}} n\right)$, $\decisionindex(\fe1)$ admits a \zksip{} with space complexity $O(\log n)$ and communication complexities $O(n^{1 + \delta})$ and $\polylog(n)$ in the setup and interactive stages, respectively. The protocol is secure against $\tilde O\left(\log^{2 + \frac{2}{\delta}} n\right)$-space distinguishers.
\end{corollary}
\begin{proof}
    Set $\degree = \log^{\frac{2}{\delta}} n$ and $\dimension = \delta \log n / 2\log\log n$, so that $\degree^\dimension = n$ and $\degree\dimension/\fieldsize = o(1)$. Note, moreover, that $\decisionindex(\fe1)$ is the polynomial evaluation problem where $f_x = \hat x$, the low-degree extension of $x$ (which can be computed in $O(\dimension \log \fieldsize)$ space) and $\evaluationpoint$ is the identification of a coordinate $j \in [n]$. Thus, applying \cref{prot:zk-pep} to the mapping $x \mapsto \hat x$ with the aforementioned parameters, we obtain a protocol with verifier space complexity
    \begin{align*}
        O(\dimension \log \fieldsize) &= O\left(\frac{\log n}{\log \log n} \cdot \log \log n\right)\\
        &= O(\log n)
    \end{align*}
    and communication complexities
    \begin{align*}
        O(\fieldsize^\dimension \dimension \log^2 \fieldsize) &= \left(\log^{1 + \frac{2}{\delta}} n\right)^{\frac{\delta \log n}{2\log\log n}} \polylog(n)\\
        &= n^{1 + \frac{\delta}{2}} \polylog(n)\\
        &= O(n^{1 + \delta})
    \end{align*}
    in the setup and $O(\degree^4 \dimension^5 \fieldsize^3 \log \fieldsize) = \polylog(n)$ in the interactive stage; moreover, it is secure against distinguishers with $\degree\dimension^2 \polylog(\fieldsize) = \tilde O\left(\log^{2 + \frac{2}{\delta}} n\right)$ space.
\end{proof}

We now select a few applications of the \textsf{zk-pep} protocol to solve other streaming problems; the remainder of this section follows reductions to \pepprotocol due to \cite{CCMTV19}.

In the \pointquery problem, the input is a stream of updates $(u, i) \in \Z \times [\ell]$ to an $\ell$-dimensional vector $y$ initialised to zero, followed by an index $j$, and the task is to output $y_j$. A formal definition follows.\footnote{We remark that \pointquery is formally a promise problem: the condition that coordinatewise sums are bounded by $M$ is assumed to hold for no-instances of the language too. However, a polynomial bound is often trivially true (as in the applications that follow).}

\begin{definition}
    \label{def:point-query}
  Let $\ell, M \in \N$ and $t \in [-M, M] \cap \Z$. The language $\pointquery(t)$ is defined as
  \begin{equation*}
      \set{
          \big(u_1, k_1, \ldots, u_n, k_n, j\big)  : \begin{array}{c}
                \forall i, u_i \in [-M, M] \cap \Z \text{ and } k_i, j \in [\ell],\\
                \forall k, \abs{\sum_{i \in [n],  k_i = k} u_i} \leq M \text{ and }\\
                \sum_{i \in [n], k_i = j} u_i = t
          \end{array}
      }.
  \end{equation*}
\end{definition}

\begin{corollary}
    \label{cor:point-query}
    Fix $\delta \in (0,1]$. Let $\ell, M \in \N$ with $\ell \in [n]$, $M = \poly(n)$ and $t \in [-M, M] \cap \Z$. There exists a \zksip{} for $\pointquery(t)$ with space complexity $O(\log^2 n)$ and communication complexities $O(n^{1 + \delta})$ and $\polylog(n)$ in the setup and interactive stages, respectively. 
\end{corollary}
\begin{proof}
    We first note that, by an application of the Chinese Remainder Theorem (see, e.g., \cite{GR15}), we may assume $M = O(\log n)$ at the cost of a logarithmic blowup to the space complexity: the verifier runs \cref{prot:zk-pep} in parallel with $O(\log n)$ fields $\field_\fieldsize \supset \field_p$ for distinct primes $p = O(\log n)$, so that any integer in $[-M, M]$ can be uniquely represented by logarithmically many field elements.
    
    We set the same parameters as in \cref{cor:index}: degree $\degree = \log^{\frac{2}{\delta}} n$ and $\dimension = \delta \log n / 2\log\log n$, but also ensure $\fieldsize = \Theta\left(\log^{1 + \frac{2}{\delta}} n\right)$ is the power of a prime larger than $2M + 1$ (so that elements of $[-M,M] \cap \Z$ map to distinct field elements).
    
    Viewing integers in $[-M,M]$ as elements of $\field$, we define $y \in \field^\ell$ by
    \begin{equation*}
        y_k \coloneqq \sum_{\substack{i \in [n]\\k_i = k}} u_i,
    \end{equation*}
    and the mapping $x = \big((u_i, k_i) : i \in [n]\big) \mapsto f^x$ by $f^x \coloneqq \hat y$. Note that the verifier can compute
    \begin{equation*}
        \hat y(\fingerprinttuple) = \sum_{k \in [\ell]} \left(\sum_{\substack{i = 1\\k_i = k}}^n u_i\right) \chi_k(\fingerprinttuple)
    \end{equation*}
    by recording the running sum of $u_i \chi_{k_i}(\fingerprinttuple)$, a task for which $O(\dimension \log \fieldsize) = O(\log n)$ space suffices.
    
    Applying \cref{prot:zk-pep} with the mapping and parameters above, we obtain a zero-knowledge SIP with space complexity $O(\log^2 n)$ (due to the aforementioned logarithmic overhead), communication complexity $O(n^{1 + \delta})$ in the setup and $\polylog(n)$ in the interactive stage.
\end{proof}

With the protocol of \cref{cor:point-query}, we obtain a zero-knowledge SIP for the \rangecount problem, where the stream consists of a sequence $x$ of elements in a set $[\ell]$ followed by a subset $R \subseteq [\ell]$, and the task is to return the number of times an element of $R$ appeared in the stream. Formally,
\begin{definition}
  \label{def:range-count}
  Let $\mathcal{R} \subseteq 2^{[\ell]}$. The language $\rangecount(t)$ is defined as
  \begin{equation*}
      \set{(x, R) \in [\ell]^n \times \mathcal{R} : \abs{\set{i \in [n]: x_i \in R}} = t}.
  \end{equation*}
\end{definition}
\begin{corollary}
    \label{cor:range-count}
    Fix $\delta \in (0,1]$. For every $\mathcal{R} \subseteq 2^{[\ell]}$ of size $\poly(n)$, the language $\rangecount(t)$ admits a \zksip{} with space complexity $O(\log^2 n)$ and communication complexities $O(n^{1 + \delta})$ and $\polylog(n)$ in the setup and interactive stages, respectively. 
\end{corollary}
\begin{proof}[Proof sketch]
    We run the protocol for \pointquery (\cref{cor:point-query}) on the stream obtained by concatenating $(R' \in \mathcal{R} : x_i \in R')$ for every $i \in [n]$ (which the verifier can simulate while streaming $x$), followed by $R$ (viewed as an element of $[\abs{\mathcal{R}}]$). More precisely, we redefine the mapping $x \mapsto f^x$ as what would be obtained by processing the derived stream, which avoids the length overhead (to $n \abs{\mathcal{R}} = \poly(n)$, rather than $n$) incurred otherwise.
    
    Since $M = n$ is an upper bound for the number of points in any subset of $[\ell]$, we obtain a protocol with the complexities as claimed.
\end{proof}

We conclude with an application of the \rangecount protocol to solve \selection (and \median in particular). For $x \in [\ell]^n$ and $i \in [\ell]$, we call $\varphi(x)$ the \emph{frequency vector} of $x$, defined as $\varphi_i(x) = \abs{\set{j \in [n] : x_j = i}}$ (see, also, \cref{def:frequency-moment}). A word in the language \selection consists of $x$ along with a \emph{rank} $r \in [n]$ the integer $k \in [\ell]$ with this rank and offsets $\phi \in [n]$, $\phi' \in \set{0} \cup [n]$. (We remark that the additional parameters take into account what the verifier learns in the search version of the SIP: not only the element $k$ with rank $r$, but the values of the cumulative frequencies up to $k - 1$ and up to $k$.)

\begin{definition}
  \label{def:selection}
  For $\ell \in [n]$, the language $\selection$ is defined as
  \begin{equation*}
      \set{(x, k, r, \phi, \phi') \in [\ell]^n \times [\ell] \times [n] \times [n] \times \set{0} \cup [n] : \sum_{i = 1}^{k-1} \varphi_i(x) = r - \phi \text{ and } \sum_{i = 1}^k \varphi_i(x) = r + \phi'}.
  \end{equation*}
\end{definition}
\begin{corollary}
    \label{cor:selection}
    Fix $\delta \in (0,1]$. There exists a \zksip{} for $\selection$ with space complexity $O(\log^2 n)$ and communication complexities $O(n^{1 + \delta})$ and $\polylog(n)$ in the setup and interactive stages, respectively.
\end{corollary}
\begin{proof}[Proof sketch]
    We execute the protocol for \rangecount twice (by temporally committing and streaming $x$ only once; this can be done by saving two independent fingerprints for $f^x$, and only running \textsf{zk-pep} twice from \cref{step:simulator-pep-algebraic-commit} onwards). The class of ranges is $\mathcal{R} = \set{[n] \setminus [i] : 0 \leq i \leq n}$, of size $O(n)$, and the verifier checks that the number of hits in the ranges $[n] \setminus [k - 1]$ and $[n] \setminus [k]$ are $r - \phi$ and $r + \phi'$, respectively.
\end{proof}

\section{A zero-knowledge sumcheck SIP}
\label{sec:sumcheck}

In the previous section we showed how \cref{prot:pep}, the polynomial evaluation protocol of \cite{CTY11}, can be made zero-knowledge with the careful addition of algebraic and temporal commitment protocols. Although \pep is a foundational problem for streaming algorithms -- generalising \indexproblem, for example -- it is not immediately clear whether the same techniques enable us to construct a zero-knowledge version of the second widely used tool in SIPs: the \emph{sumcheck} protocol. In this section, we prove that they do: \cref{prot:zk-sumcheck} leys out \textsf{zk-sumcheck}, a \zksip{} for the \sumcheck problem (\cref{def:sumcheck}) with the same components, namely, the algebraic and temporal commitments that enabled \textsf{zk-pep}.

Sumcheck protocols are extremely useful building blocks for the construction of interactive proofs; indeed, some of the most celebrated results of the last two decades rely on them, most notably the GKR \cite{GKR08} and subsequent delegation-of-computation protocols (e.g., \cite{RVW13, RRR19, RR20}). Roughly speaking, they allow a verifier to check that the sum, over a subcube, of the evaluations of a polynomial yields a prescribed field element; they save \emph{exponentially} in the communication (and time) complexity as compared to sending the entire description of the polynomial. In particular, they enable the (exact) computation of frequency moments of a stream via an interactive protocol in sublinear space \cite{CCMT14}, which is impossible without interaction \cite{AMS99}.

More precisely, let $f: \field^\dimension \to \field$ be a polynomial of (individual) degree $\degree$ and $\cubeside \subset \field$ be an evaluation domain. One obvious way to check that $\sum_{\evaluationpoint \in H^\dimension} f(\evaluationpoint)$ is equal to some $\fe1 \in \field$ is via the description of $f$ (say, as a list of sufficiently many evaluations), from which the sum can be computed directly. This requires not only the entire description of $f$, which has size $(\degree + 1)^\dimension$; but also entails evaluating $f$ over $\abs{\cubeside}^\dimension$ many points, implying an even larger runtime.

The standard sumcheck protocol (\cref{prot:sumcheck}) enables a verifier $V$ to offload this costly computation to a powerful prover $P$ and check the claim by communicating $O(\degree\dimension)$ field elements in $O(\abs{\cubeside} \dimension \degree)$ time steps, with \emph{a single random evaluation of $f$}.\footnote{\cref{prot:sumcheck} is laid out in a somewhat non-standard (but equivalent) form, with checks deferred to the end, that more closely resembles the streaming version we construct.}

\iflipics
\begin{figure}
\else
\begin{protocol}[label={prot:sumcheck}]{$\textsf{sumcheck}(f, \fe1)$}
\fi
 	\textbf{Input:} Explicit access to $\field = \field_\fieldsize$, evaluation domain $H \subset \field$, degree $\degree$, dimension $\dimension$ and $\fe1 \in \field$ as well as $f(\fingerprinttuple)$ with $\fingerprinttuple \sim \field^\dimension$, where $f: \field^\dimension \to \field$ is a degree-$\degree$ polynomial.
 	
    \iflipics\else
    \tcbline
    \fi

 	Repeat, from $i = 1$ to $\dimension$:
 	\begin{enumerate}[label={}]
 	    \item[$\bm P$:] Send the polynomial $f_i(T) = \sum_{\fe2_{i+1}, \ldots, \fe2_\dimension \in \cubeside} f(\fingerprinttuple_1, \ldots, \fingerprinttuple_{i-1}, T, \fe2_{i+1}, \ldots, \fe2_\dimension)$.

 	    \item[$\bm V$:] Send $\fingerprinttuple_i$.
 	\end{enumerate}

 	$\bm V$: Check that $\sum_{\fe2_1 \in H} f_1(\fe2_1) = \fe1$, $f(\fingerprinttuple) = f_\dimension(\fingerprinttuple_\dimension)$ and the intermediate polyomials satisfy $\sum_{\fe2_i \in H} f_i(\fe2_i) = f_{i-1}(\fingerprinttuple_{i-1})$ for all $2 \leq i < \dimension$, accepting if so and rejecting otherwise.
\iflipics
\caption{Protocol $\textsf{sumcheck}(f, \fe1)$}
\label{prot:sumcheck}
\end{figure}
\else
\end{protocol}
\fi

It is well known that the protocol above (always) accepts if $\sum_{\evaluationpoint \in \cubeside^\dimension} f(\evaluationpoint) = \fe1$, and rejects with probability at least $1 - \degree \dimension / \fieldsize$ otherwise (see, e.g., \cite{AB09}). As sums of polynomials can be performed in a streaming fashion, the verifier only needs $O(\dimension \log \fieldsize)$ bits of space.

\subsection{The protocol}

We now show that the techniques of \cref{sec:commitment-protocols} enable us to construct a streaming zero-knowledge variant of $\textsf{sumcheck}(f, \fe1)$, which solves the problem defined next.

\begin{definition}
    \label{def:sumcheck}
    Let $\fe1 \in \field$, $\cubeside \subseteq \field$ and $f = \set{f^x : x \in \alphabet^n}$ be a mapping such that $f^x: \field^\dimension \to \field$ is a degree-$\degree$ polynomial. $\sumcheck(f, \fe1)$ is the language $\set{x \in \alphabet^n : \sum_{\evaluationpoint \in \cubeside^\dimension}f^x(\evaluationpoint) = \fe1}$.
\end{definition}

The techniques need to be adapted, however, with one key distinction between \textsf{zk-sumcheck} and \textsf{zk-pep}: the prover now must make many (algebraic) commitments, each of which is used in a pair of decommitments; moreover, the commitments cannot be sent in parallel anymore, owing to dependencies between messages in contiguous rounds. Intuitively, neither of these should pose too great a challenge: computing fingerprints of a set of messages whose commitment is sent sequentially should be no easier than when they are sent in parallel (indeed, for one-way communication protocols they are exactly equivalent); and if one algebraic decommitment does not leak a significant amount of information, two should not do so either.

The protocol follows. We note that (differently from \cref{sec:pep}) $\bm{\chi}(\rfe1)$ denotes the vector of Lagrange polynomials over $\field$ for degree-$\degree$ univariate polynomials with interpolating set $[\degree + 1]$, i.e., $\bm{\chi}(\rfe1) = \big(\chi_i(\rfe1) : i \in [\degree + 1]) \in \field^{\degree + 1}$.

\iflipics
\begin{figure}
\else
\begin{protocol}[label={prot:zk-sumcheck}]{$\textsf{zk-sumcheck}(f, \fe1)$}
\fi
 	\textbf{Input:} Explicit access to $\field$, element $\fe1 \in \field$, degree $\degree$, dimension $\dimension$, evaluation domain $H \subset \field$ and mapping $x \mapsto f^x$; streaming access to $x\in\alphabet^n$.\\
	
	\textbf{Parameters:}
	\begin{itemize}[label={}, nosep, itemindent=-1.5em]
	    \item Field size $\fieldsize = \abs{\field}$ satisfying $\degree \dimension = o(\fieldsize)$;
	    \item Commitment lengths $\vcommitlength = \fieldsize^\dimension (\log \dimension + \log\log \fieldsize)/96$ and $\pcommitlength = \fieldsize^{\log\log \fieldsize}$.
	\end{itemize}
    \tcbline
    \vspace*{-.5\baselineskip} 
    
    \begin{steps}[wide]
        \setcounter{stepsi}{-1}
        \item\label{step:sumcheck-temporal-commit} Temporal commitment
	    \begin{enumerate}[label={}, leftmargin=3em]
	        \item[$\bm P$:] Send a string $\vcommitstring \sim \big((\field \setminus [\degree + 1])^\dimension\big)^\vcommitlength$.
	        \item[$\bm V$:] Sample $\fingerprinttuple \sim \left(\field \setminus [\degree+1]\right)^\dimension$ and stream $\vcommitstring$. Check if $z_i = \fingerprinttuple$ for each $i$, storing $\ell \coloneqq i$ if so.
	        
	        Reject if $\fingerprinttuple \neq \vcommitstring_i$ for all $i \in [\vcommitlength]$.
	    \end{enumerate}
        \iflipics\else%
	    \vspace*{-\baselineskip}
	    \tcbline
        \vspace*{-.5\baselineskip} 
        \fi

	    \item\label{step:sumcheck-input} Input streaming
	    \begin{enumerate}[label={}, leftmargin=3em]
	        \item[$\bm V$:] Stream $x$ and compute $f^x(\fingerprinttuple) \in \field$.
        \end{enumerate}
        \iflipics\else%
	    \vspace*{-\baselineskip}
	    \tcbline
        \vspace*{-.5\baselineskip} 
        \fi

	    \item\label{step:sumcheck-algebraic-commit} Algebraic commitments
        \begin{enumerate}[label={}, leftmargin=3em]
            \item[$\bm P$:] Compute $f_1(T) = \sum_{\fe2_2, \ldots, \fe2_{\dimension} \in \cubeside} f^x(T, \fe2_2 \ldots, \fe2_{\dimension})$ and sample $k \sim [\pcommitlength]$.
            \item[$\bm V$:] Sample $\Fingerprinttuple^{(1)}, \ldots, \Fingerprinttuple^{(\dimension + 1)} \sim \field^\dimension$. Compute $\bm{\chi}(\fingerprinttuple_1) \ldots, \bm{\chi}(\fingerprinttuple_\dimension)$ and the linear coefficients $\bm{\theta}$ such that $\sum_{\fe2 \in \cubeside} g(\fe2) = \sum_i \bm{\theta}_i g(i)$ when $g$ is a degree-$\degree$ univariate polynomial.
        \end{enumerate}
        
        \begin{enumerate}[label={}, leftmargin=1.3em]
            \item Repeat, from $i = 1$ to $\dimension$:
         	\item
         	\begin{enumerate}[label={}, leftmargin=3em]
         	    \item[$\bm P$:] Send $\pcommitstring^{(i)} \sim \field^{(\degree + 1) \times \pcommitlength}$ and $\bm{\fe3}^{(i)} = \big(f_i(j) - \pcommitstring^{(i)}_{jk} : j \in [\degree + 1]\big)$.
         	    \item[$\bm V$:] Compute the fingerprints $\hat \pcommitstring^{(i)}\left(\Fingerprinttuple^{(i)}, \bm{\chi}(\fingerprinttuple_i)\right)$ and $\hat \pcommitstring^{(i)}\left(\Fingerprinttuple^{(i+1)},  \bm{\theta}\right)$, as well as the dot products $\bm{\chi}(\fingerprinttuple_i) \cdot \bm{\fe3}^{(i)}$ and $\bm{\theta} \cdot \bm{\fe3}^{(i)}$.
         	   
         	   Send $\fingerprinttuple_i$.
         	    \item[$\bm P$:] If $i < m$, compute $f_{i+1}(T) = \sum_{\fe2_{i+2}, \ldots, \fe2_{\dimension} \in \cubeside} f^x(\fingerprinttuple_1, \ldots, \fingerprinttuple_i, T, \fe2_{i+2}, \ldots, \fe2_{\dimension})$.
         	\end{enumerate}
         	\item $\bm P$: Send $k$.
        \end{enumerate}
        \iflipics\else%
	    \vspace*{-\baselineskip}
	    \tcbline
        \vspace*{-.5\baselineskip} 
        \fi
        
	    \item\label{step:sumcheck-temporal-decommit} Temporal decommitment
	    \begin{enumerate}[label={}, leftmargin=3em]
	        \item[$\bm V$:] Send $\ell$.
            \item[$\bm P$:] Check that $\vcommitstring_\ell = \fingerprinttuple \in \big(\field \setminus [\degree+1]\big)^\dimension$, aborting otherwise.
        \end{enumerate}
        \iflipics\else%
	    \vspace*{-\baselineskip}
	    \tcbline
        \vspace*{-.5\baselineskip} 
        \fi

        \item\label{step:sumcheck-algebraic-decommit} Algebraic decommitments
        \begin{enumerate}[label={}, leftmargin=3em]
            \item[$\bm V$:] For all $1 < i \leq \dimension$, run
	        \begin{equation*}
	            \textsf{decommit}\left(0, \bm{\theta} \cdot \pcommitstring^{(i)} - \bm{\chi}(\fingerprinttuple_{i-1}) \cdot \pcommitstring^{(i-1)}, k\right), \text{ with}
	       \end{equation*}
	       fingerprint $\hat \pcommitstring^{(i)}\left(\Fingerprinttuple^{(i)}, \bm{\theta}\right) - \hat \pcommitstring^{(i-1)}\left(\Fingerprinttuple^{(i)}, \bm{\chi}(\fingerprinttuple_{i-1})\right)$ and correction $\bm{\theta} \cdot \bm{\fe3}^{(i)} - \bm{\chi}(\fingerprinttuple_{i-1}) \cdot \bm{\fe3}^{(i-1)}$.
	        
	        Run $\textsf{decommit}\left(\fe1, \bm{\theta} \cdot \pcommitstring^{(1)}, k\right)$ with fingerprint $\hat \pcommitstring^{(1)}\left(\Fingerprinttuple^{(1)},  \bm{\theta}\right)$ and correction $\bm{\theta} \cdot \bm{\fe3}^{(1)}$.
	        
	        Run $\textsf{decommit}\left(f^x(\fingerprinttuple), \bm{\chi}(\fingerprinttuple_\dimension) \cdot \pcommitstring^{(\dimension)}, k\right)$ with fingerprint $\hat \pcommitstring^{(\dimension)}\left(\Fingerprinttuple^{(\dimension + 1)}, \bm{\chi}(\fingerprinttuple_\dimension)\right)$ and correction $\bm{\chi}(\fingerprinttuple_\dimension) \cdot \bm{\fe3}^{(\dimension)}$.
	        
	        Accept if all decommitments accept, and reject otherwise.
	    \end{enumerate}
    \end{steps}
\iflipics
\caption{Protocol $\textsf{zk-sumcheck}(f, \fe1)$}
\label{prot:zk-sumcheck}
\end{figure}
\else
\end{protocol}
\fi

\subsection{Analysis of the protocol}

We now show that \textsf{zk-sumcheck} is a valid (i.e., complete and sound) streaming interactive proof, and compute its space and communication complexities.
 
 \begin{theorem}
    \label{thm:sumcheck-correctness}
    Let $f$ be such that an evaluation of the $\field_\fieldsize$-polynomial $f^x$ can be computed by streaming $x$ in $O(\dimension^2 \log \fieldsize)$ space. For any $\fe1 \in \field_\fieldsize$, \cref{prot:zk-sumcheck} is an SIP for $\sumcheck(f, \fe1)$ with space complexity $s = O(\dimension^2 \log \fieldsize)$, communication complexity $O(\fieldsize^\dimension \dimension \log^2 \fieldsize)$ in the setup and $O(\fieldsize^{\log\log \fieldsize} \degree\dimension \log \fieldsize) = \fieldsize^{\log\log \fieldsize} \poly(\fieldsize)$ in the interactive phase.
\end{theorem}
\begin{proof}
    As in \cref{thm:pep-correctness}, we first show completeness and soundness, then compute the complexities.
    
    \paragraph*{Completeness\iflipics\else.\fi} Recall that $\textsf{decommit}(\fe2, w, k)$ with correction $\fe3$ accepts if (the fingerprint matches the LDE of $w$ and) $\fe3 + w_k = \fe2$. Therefore, when $P$ and $V$ are both honest, the first $\dimension - 1$ decommitments of \cref{step:sumcheck-algebraic-decommit} accept, since
    
    \begin{align*}
        \bm{\theta} \cdot \bm{\fe3}^{(i)} - \bm{\chi}(\fingerprinttuple_{i-1}) \cdot \bm{\fe3}^{(i-1)} &+ \left(\bm{\theta} \cdot \pcommitstring^{(i)} - \bm{\chi}(\fingerprinttuple_{i-1}) \cdot \pcommitstring^{(i-1)}\right)_k\\
        &= \sum_{j = 1}^{\degree + 1} \left(\bm{\theta}_j \big(\bm{\fe3}^{(i)}_j + \pcommitstring^{(i)}_{jk}\big) - \chi_j(\fingerprinttuple_{i-1}) \big(\bm{\fe3}^{(i-1)}_j + \pcommitstring^{(i-1)}_{jk}\big)\right)\\
        &= \sum_{j = 1}^{\degree + 1} \bm{\theta}_j f_i(j) - \sum_{j = 1}^{\degree + 1} \chi_j(\fingerprinttuple_{i-1})
         f_{i-1}(j)\\
        &= \left(\sum_{\fe2 \in \cubeside} f_i(\fe2)\right) - f_{i-1}(\fingerprinttuple_{i-1})\\
        &= 0.
    \end{align*}
    
    Likewise, the last two decommitments accept because
    \begin{align*}
        \bm{\theta} \cdot \bm{\fe3}^{(1)} + \left(\bm{\theta} \cdot \pcommitstring^{(1)}\right)_k &= \sum_{j = 1}^{\degree + 1} \bm{\theta}_j (\bm{\fe3}^{(1)}_j + \pcommitstring^{(1)}_{jk})\\
        &= \sum_{j = 1}^{\degree + 1} \bm{\theta}_j f_1(j)\\
        &= \sum_{\fe2 \in \cubeside} f_1(\fe2)\\
        &= \sum_{\evaluationpoint \in \cubeside^\dimension} f(\evaluationpoint)\\
        &= \fe1
    \end{align*}
    and
    \begin{align*}
        \bm{\chi}(\fingerprinttuple_\dimension) \cdot \bm{\fe3}^{(\dimension)} + \left(\bm{\chi}(\fingerprinttuple_\dimension) \cdot \pcommitstring^{(\dimension)}\right)_k &= \sum_{j = 1}^{\degree + 1} \chi_j(\fingerprinttuple_\dimension) (\bm{\fe3}^{(\dimension)}_j + \pcommitstring^{(\dimension)}_{jk})\\
        &= \sum_{j = 1}^{\degree + 1} \chi_j(\fingerprinttuple_\dimension) f_\dimension(j)\\
        &= f_\dimension(\fingerprinttuple_\dimension)\\
        &= f^x(\fingerprinttuple),
    \end{align*}
    respectively. The verifier thus accepts unless $\fingerprinttuple \neq \set{\vcommitstring_i : i \in [\vcommitlength]}$ in \cref{step:sumcheck-temporal-commit}, an event with probability
    \begin{equation*}
        \left(1 - \frac1{\fieldsize - \degree - 1}\right)^\vcommitlength \leq e^{-\vcommitlength / (\fieldsize - \degree - 1)^\dimension} \leq e^{-\vcommitlength/\fieldsize^\dimension} = o(1).
    \end{equation*}
    
    \paragraph*{Soundness\iflipics\else.\fi} We divide the behaviour of a malicious prover into three cases. The first (and simplest) is when $\widetilde P$ commits to $f_i$ for all $i$ and decommits with polynomials whose evaluations at $0$ yield the same values as the honest prover (i.e., in $\textsf{decommit}(\fe2, w, k)$ with $\fe3$ as the correction, $\widetilde P$ replies with a polynomial $g$ such that $g(0) = w_k + \fe3$). Then, since $\sum_{\evaluationpoint \in \cubeside^\dimension} f(\evaluationpoint) \neq \fe1$, the verifier rejects in $\textsf{decommit}\left(\fe1, \bm{\theta} \cdot \pcommitstring^{(1)}, k\right)$ with probability 1.
    
    The second case is when $\widetilde P$ commits to a sequence of polynomials $g_1, \ldots, g_\dimension$ such that $g_i \neq f_i$ for some $i$, and decommits honestly. Then $V$ accepts if and only if the set $\set{g_i}$ leads the verifier in the standard sumcheck protocol to accept; by the soundness of that protocol, $V$ accepts with probability at most $\degree\dimension/(\fieldsize - \degree - 1) = o(1)$.
    
    The only remaining case is when $\widetilde P$ commits to a sequence of polynomials $\set{g_i}$ (which may or may not coincide with $\set{f_i}$) and, in at least one decommitment with respect to a string $w$ where $\widetilde P$ receives the line $\line$, the prover replies with a degree-$\degree\dimension$ polynomial $g$ such that $g(0) \neq w_k = \hat w_{|\line}(0)$. Then, since $V$ has a fingerprint $\hat w(\Fingerprinttuple)$ with $\Fingerprinttuple \sim \field^\dimension$ and a field element $\rfe2 \sim \field$ such that $\line(\rfe2) = \Fingerprinttuple$, we have $g(\rfe2) \neq \hat w(\Fingerprinttuple) =  \hat w_{|\line}(\rfe2)$ with probability $\degree\dimension/\fieldsize = o(1)$ by \cref{lem:sz} (Schwartz-Zippel), and soundness follows.
   
   \paragraph*{Space and communication complexities\iflipics\else.\fi} The communication of the setup (\cref{step:sumcheck-temporal-commit}, the temporal commitment) is $\fieldsize^\dimension (\log \dimension + \log\log \fieldsize) \dimension \log \fieldsize = O(\fieldsize^\dimension \dimension \log^2 \fieldsize)$ bits. The communication of the interactive phase (\cref{step:sumcheck-algebraic-commit,step:sumcheck-temporal-decommit,step:sumcheck-algebraic-decommit}) is dominated by the $\dimension$ algebraic commitments to elements of $\field^{\degree + 1}$ with length $\pcommitlength = \fieldsize^{\log\log \fieldsize}$ each, for a total of $O(\fieldsize^{\log\log \fieldsize} \degree \dimension \log \fieldsize) \leq \fieldsize^{\log\log \fieldsize + 2}$ bits.
   
   The verifier's space complexity is dominated by computing $f^x(\fingerprinttuple)$ and storing $O(\dimension)$ elements of $\field^\dimension$ (i.e., $\fingerprinttuple$ and $\Fingerprinttuple^{(i)}$ for $i \in [\dimension + 1]$), so that it is bounded by $O(\dimension^2 \log \fieldsize)$.
\end{proof}

\subsection{Zero-knowledge}

Having shown that \textsf{zk-sumcheck} is a valid streaming interactive proof, we now show it is also zero-knowledge.

 \begin{theorem}
    \label{thm:sumcheck-zk}
    \cref{prot:zk-sumcheck} is zero-knowledge against $\poly(\fieldsize)$-space streaming distinguishers. The simulator has space complexity $\poly(\fieldsize)$.
\end{theorem}
\begin{proof}
    We shall prove indistinguishability as we have done earlier: with the simulator $S$ shown in \cref{alg:sumcheck-simulator}, we assume towards contradiction that there exists $\fe1 \in \field$, an input $x \in \field^n$, internal randomness $r$, a space-$O(\dimension^2 \log \fieldsize)$ verifier $\widetilde V$ and a $\poly(\fieldsize)$-space distinguisher $D$ that accepts $\view_{P, \widetilde V}(x, r)$ with probability $\eps = \Omega(1)$ above that with which $D$ accepts $S\big(\widetilde V, x, r\big)$. Then, via \cref{lem:string-to-bit-index}, we construct a one-way protocol for \indexproblem with impossibly large success probability.
    
    The space complexity of $S$ is dominated by its storing of $O(\dimension^2 \log \fieldsize)  = \poly(\fieldsize)$ elements of $\field^\dimension \times [\vcommitlength]$ and by the computation of the partial sums $(g_i : i \in [\dimension])$. Note that the naive strategy of sampling $g$ and computing the corresponding partial sums requires $\Omega(\degree ^\dimension)$ space; however, \cite{BCFGRS17} constructs an algorithm that can sample from the same distribution in $\poly(\fieldsize)$ time, and thus space.\footnote{More precisely, the algorithm of \cite{BCFGRS17} allows us to sample from the distributions $g_i(\fe2)$ for any $\fe2$ and $i$ under the uniform distribution of $g$ satisfying a set of constraints. To sample $(g_1, \ldots, g_\dimension)$, we begin with the set of constraints induced by $C$ and, after sampling $g_i(j)$, include the corresponding constraint before the next sample.} Note, moreover, that the alphabet over which $\vcommitstring$ is taken has size
    
    \begin{align*}
        (\fieldsize - \degree - 1)^\dimension &= \fieldsize^\dimension \left(1 - \frac{\degree + 1}{\fieldsize}\right)^\dimension\\
        &\geq \fieldsize^\dimension \left(1 - \frac1{\dimension}\right)^\dimension\\
        &\geq \frac{\fieldsize^\dimension}{3}\\
        &\geq \frac{32 \vcommitlength}{\log\log \vcommitlength},
    \end{align*}
    
    so that \cref{thm:correct-set} applies. (The conditions $(\fieldsize - \degree - 1)^\dimension = \Theta\left(\frac{\vcommitlength}{\log\log\vcommitlength}\right)$ and $\log \fieldsize \leq s = \polylog(\fieldsize)$ are also clearly satisfied.)
    
    \iflipics
    \begin{figure}
    \else
    \begin{algorithm}[label={alg:sumcheck-simulator}]{Simulator for \cref{prot:zk-sumcheck}}
    \fi
	    \textbf{Input:} Whitebox access to $\widetilde V$; oracle access to random bit string of length $\fieldsize^{\dimension + \log\log \fieldsize} \poly(\fieldsize)$ interpreted as the concatenation of $\vcommitstring \in (\field^\dimension)^\vcommitlength$ and $\pcommitstring^{(i)} \in \field^{(\degree + 1) \times \pcommitlength}$ for all $i \in [\dimension]$.\\
	    
	    \textbf{Output:} View $\left(\vcommitstring, x, \big(\pcommitstring^{(i)}, \bm{\fe3}^{(i)} : i \in [\dimension]\big), k, \big(h_i : i \in [\dimension + 1]\big)\right)$ with $\vcommitstring \in \big((\field \setminus [\degree + 1])^\degree\big)^\vcommitlength$, $\pcommitstring^{(i)} \in \field^{(\degree + 1) \times \pcommitlength}$, $\bm{\fe3}^{(i)} \in \field^{\degree + 1}$, $k \in [\pcommitlength]$ and $h_i: \field \to \field$ of degree $\degree\dimension$. 
	    
        \iflipics\else%
    	\tcbline
        \vspace*{-.5\baselineskip} 
        \fi
        
        \begin{steps}[wide]
            \setcounter{stepsi}{-1}
	        \item\label{step:sumcheck-simulator-temporal-commit} Temporal commitment
	        \begin{enumerate}[label={}, leftmargin=3em]
	            \item[$\bm S$:] Send $\vcommitstring \in \big((\field \setminus [\degree + 1])^\dimension\big)^\vcommitlength$.
	            
	            \item[$\bm{\widetilde V}$:] Simulate until the end of this step and let $\snapshot \in \bitset^s$ be the resulting snapshot of $\widetilde V$. Use the whitebox oracle $\whiteboxoracle$ to determine the set $C \subset \set{(\vcommitstring_i, i) : i \in [\vcommitlength]}$ of size $s$ with the largest $\whiteboxoracle(\snapshot, (\vcommitstring_i, i))$.
	        \end{enumerate}
            \iflipics\else%
        	\vspace*{-\baselineskip}
        	\tcbline
        	\vspace*{-.5\baselineskip}  
            \fi
        	
	        \item \label{step:sumcheck-simulator-input} Input streaming
	        \begin{enumerate}[label={}, leftmargin=3em]
	            \item[$\bm{\widetilde V}$:] Stream $x$, simulating the verifier while computing and storing $f^x(\vcommitstring_i)$ for all $(\vcommitstring_i, i) \in C$.
            \end{enumerate}
            \iflipics\else%
        	\vspace*{-\baselineskip}
        	\tcbline
        	\vspace*{-.5\baselineskip}  
            \fi
        	
	        \item \label{step:sumcheck-simulator-algebraic-commit} Algebraic commitments

	        \begin{enumerate}[label={}, leftmargin=3em]	        
	            \item[$\bm S$:] Take $g_1: \field \to \field$ of degree (at most) $\degree$ under the distribution determined by sampling $g: \field^\dimension \to \field$ subject to the constraints $\sum_{\evaluationpoint \in \cubeside^\dimension} g(\evaluationpoint) = \fe1$ and $g(\vcommitstring_i) = f^x(\vcommitstring_i)$ for all $(\vcommitstring_i, i) \in C$, then outputting $g_1(T) = \sum_{\fe2_2, \ldots, \fe2_{\dimension} \in \cubeside} g(T, \fe2_2 \ldots, \fe2_{\dimension})$.
	            
	            Sample $k \sim [\pcommitlength]$.
	            
	            \item[$\bm{\widetilde V}$:] Simulate until the end of the step.
	        \end{enumerate}

            \begin{enumerate}[label={}, leftmargin=1.3em]
                \item Repeat, from $i = 1$ to $\dimension$:
             	\item \begin{enumerate}[label={}, leftmargin=3em]
             	    \item[$\bm S$:] Send $\pcommitstring^{(i)}$ and $\bm{\fe3}^{(i)}= \left(g_i(j) - \pcommitstring^{(i)}_{jk} : j \in [\degree + 1]\right)$.
             	    \item[$\bm{\widetilde V}$:] Simulate until $\fingerprinttuple_i$ is sent (or until the end of the step when $i = \dimension$).
             	    \item[$\bm S$:] If $i < m$, sample $g_{i+1}$ under the distribution given by taking $g$ randomly and outputting $g_{i+1}(T) = \sum_{\fe2_{i+2}, \ldots, \fe2_{\dimension} \in \cubeside} g(\fingerprinttuple_1, \ldots, \fingerprinttuple_i, T, \fe2_{i+2}, \ldots, \fe2_{\dimension})$.
             	\end{enumerate}
             	
             	\item $\bm S$: Compute $\bm{\theta}$ such that $\sum_{\fe2 \in \cubeside} h(\fe2) = \sum_i \bm{\theta}_i h(i)$ when $h$ is a degree-$\degree$ univariate polynomial, and send $k$.
         	\end{enumerate}
         	
            \iflipics\else%
         	\vspace*{-\baselineskip} 
        	\tcbline	
            \vspace*{-.5\baselineskip} 
            \fi
            
	        \item \label{step:sumcheck-simulator-temporal-decommit} Temporal decommitment
	        \begin{enumerate}[label={}, leftmargin=3em]
	            \item[$\bm{\widetilde V}$:] Simulate until $\widetilde V$ sends $\ell \in [\vcommitlength]$.
	            
	            \item[$\bm S$:] Abort if $\vcommitstring_\ell \neq \fingerprinttuple$, $\fingerprinttuple \notin \big(\field \setminus [\degree + 1]\big)^\dimension$ or $(\fingerprinttuple, \ell) \notin C$.
            \end{enumerate}
            \iflipics\else%
        	\vspace*{-\baselineskip}
        	\tcbline
            \vspace*{-.5\baselineskip} 
            \fi
            
            \item \label{step:sumcheck-simulator-algebraic-decommit} Algebraic decommitments
            \begin{enumerate}[label={}, leftmargin=1.3em]
            
            \item For all $1 < i \leq \dimension$,
            \item \begin{enumerate}[label={}, leftmargin=3em]
	            \item[$\bm{\widetilde V}$]: Simulate until $\widetilde V$ sends a line $\line_i: \field \to \field^\dimension$.
	            
	            \item[$\bm S$:] Abort if $\line_i(0) \neq k$, and otherwise send
	            \begin{equation*}
	                \left(\left(\bm{\theta} \cdot \hat \pcommitstring^{(i)} - \bm{\chi}(\fingerprinttuple_i) \cdot \hat \pcommitstring^{(i-1)} \right) \circ \line_i(j) :j \in [\degree\dimension + 1]\right).
	            \end{equation*}
	        \end{enumerate}
	        \item $\bm{\widetilde V}$: Simulate until $\widetilde V$ sends a line $\line_1: \field \to \field^\dimension$.
	        
	        \item $\bm S$:  Abort if $\line_1(0) \neq k$, and otherwise send
	            \begin{equation*}
	                \left(\left(\bm{\theta} \cdot \hat \pcommitstring^{(1)}\right) \circ \line_1(j) :j \in [\degree\dimension + 1]\right).
	            \end{equation*}
	        
	        \item $\bm{\widetilde V}$: Simulate until $\widetilde V$ sends a line $\line_{\dimension + 1}: \field \to \field^\dimension$.
	        
	        \item $\bm S$: Abort if $\line_{\dimension + 1}(0) \neq k$, and otherwise send
	            \begin{equation*}
	                \left(\left(\bm{\chi}(\fingerprinttuple_\dimension) \cdot \hat \pcommitstring^{(\dimension)}\right) \circ \line_{\dimension + 1}(j) :j \in [\degree\dimension + 1]\right).
	            \end{equation*}
	       \end{enumerate}
	   \end{steps}
    \iflipics
    \caption{Simulator for \cref{prot:zk-sumcheck}}
    \label{alg:sumcheck-simulator}
    \end{figure}
    \else
    \end{algorithm}
    \fi
    
    We fix a string $\vcommitstring$ (and thus the set $C$ of the verifier's likely decommitments) along with bits of the verifier's random string $r$ that ensure distinguishing bias at least $\eps/2$ and $o(1)$ probability of simulation failure (recall that failure corresponds to the event $(\fingerprinttuple, \ell) \notin C$).
    Consider the following (linear) mapping between $\field$-vector spaces: from polynomials $g: \field^\dimension \to \field$ of degree at most $\degree$ that satisfy the $s + 1$ linear constraints of the fingerprints and subcube sum (i.e., $g(\fingerprinttuple) = \vcommitstring_i$ for all $(\vcommitstring_i, i) \in C$ and $\sum_{\evaluationpoint \in \cubeside^\dimension} g(\evaluationpoint) = \fe1$) to the sequence of univariate (partial sum) polynomials $\sum_{\fe2_{i + 1}, \ldots, \fe2_\dimension \in \cubeside} g(\fingerprinttuple_1, \ldots, \fingerprinttuple_{i - 1}, T, \fe2_{i + 1}, \ldots, \fe2_\dimension)$ for all $i \in [\dimension]$ and evaluation points $\fingerprinttuple$ in $C$.
    
    Let $\ell \leq (\degree + 1) \dimension$ be the dimension of the image of this mapping, and let $\bm{\xi} = \bm{\xi}(\fingerprinttuple) \in \field^{(\degree + 1)\dimension \times \ell}$ be the linear coefficients that map vectors in $\field^\ell$ to partial sums (given by $\degree + 1$ evaluations) with respect to $\fingerprinttuple$. We now proceed to Alice's strategy, who receives $w \in \field^{\ell \times \pcommitlength}$ as input and uses a random $\pcommitstring'^{(i)}$ shared with Bob for each commitment string $\pcommitstring^{(i)}$. She will also use $t^{(i)}$ for each pair $\pcommitstring^{(i-1)}, \pcommitstring^{(i)}$; additionally, $t^{(1)}$ and $t^{(\dimension + 1)}$ will be used for $\pcommitstring^{(1)}$ and $\pcommitstring^{(\dimension)}$, respectively. The $t^{(i)}$ will ensure Bob knows the linear combination of every algebraic decommitment.
    
    More precisely, Alice runs $S$ (with the fixed string $\vcommitstring$ and partially fixed $r$) until the end of \cref{step:sumcheck-simulator-temporal-commit}, determines the set $C$, samples $\fingerprinttuple' \sim F = \set{\vcommitstring_i : (\vcommitstring_i, i) \in C}$ and sets $\bm{\xi} = \bm{\xi}(\fingerprinttuple')$. For every $(i, j) \in [\dimension - 1] \times [\degree]$ and $(i, j) \in \set{\dimension} \times [\degree - 1]$, she sets $\pcommitstring^{(i)}_j = \pcommitstring'^{(i)}_j + \left(\bm{\xi} \cdot w\right)_{(i-1)\degree + j}$. She also sets the remaining rows (i.e., $\pcommitstring^{(i)}_{\degree + 1}$ for all $i$ as well as $\pcommitstring^{(\dimension)}_\degree$) to satisfy
    
    \begin{align*}
        \bm{\theta} \cdot \pcommitstring^{(1)} &= t^{(1)},\\
        \bm{\theta} \cdot \pcommitstring^{(i)} - \bm{\chi}(\fingerprinttuple'_{i-1}) \cdot \pcommitstring^{(i-1)} &= t^{(i)} \quad\text{ for } 1 < i \leq \dimension \text { and}\\
        \bm{\chi}(\fingerprinttuple'_\dimension) \cdot \pcommitstring^{(\dimension)} &= t^{(\dimension+1)}.
    \end{align*}
    
    Note that these are $\dimension + 1$ linear constraints on $\dimension + 1$ row vectors of dimension $\pcommitlength$, and since $\bm{\theta}_{\degree + 1}$ and $\chi_\degree(\fingerprinttuple'_\dimension)$ are nonzero, there is at least one solution.\footnote{The condition $\chi_\degree(\fingerprinttuple'_\dimension) \neq 0$ follows from choosing $\fingerprinttuple_\dimension' \notin [\degree + 1]$, and we assume the last entry of $\bm{\theta}$ is nonzero without loss of generality. Note that if $\bm{\theta}$ is the zero vector the problem trivialises: in this case the verifier does not need assistance from a prover (or even to stream $x$), accepting if and only if $\fe1 = 0$.} (If some constraint is not independent from the others, Alice replaces it with a ``canonical'' constraint to ensure a unique solution, e.g., setting the linear coefficients for $\pcommitstring^{(i)}_{\degree + 1}$ with the smallest bit representation that makes the constraint independent.) She then simulates \cref{step:sumcheck-simulator-input} and the part of \cref{step:sumcheck-simulator-algebraic-commit} until $\widetilde V$ (and $D$) finish streaming the $\pcommitstring^{(i)}$, sending the resulting snapshots of $S$, $\widetilde V$ and $D$ to Bob along with $\fingerprinttuple'$ in a $\poly(\fieldsize)$-bit message.
    
    Bob reads his input $(\bm{\eta}, k)$ and sets the correction tuples $\bm{\fe3}^{(i)} \in \field^{\degree + 1}$ so as to satisfy constraints with the same linear coefficients as $\pcommitstring^{(i)}$: he sets $\bm{\fe3}^{(i)}_j  = (\bm{\xi} \cdot \bm{\eta})_{(i - 1) \degree + j} - \pcommitstring'^{(i)}_{jk}$ for $(i, j) \in [\dimension - 1] \times [\degree]$ and $(i, j) \in \set{\dimension} \times [\degree - 1]$; then sets the coordinates $i = \degree + 1$ and $j \in [\dimension]$ as well as $(i, j) = (\degree, \dimension)$ to satisfy
    
    \begin{align*}
        \bm{\theta} \cdot \bm{\fe3}^{(1)} &= \fe1 - t_k^{(1)},\\
        \bm{\theta} \cdot \bm{\fe3}^{(i)} - \bm{\chi}(\fingerprinttuple_{i-1}') \cdot \bm{\fe3}^{(i-1)} &= -t_k^{(i)} \quad\text{ for } 1 < i \leq \dimension \text { and}\\
        \bm{\chi}(\fingerprinttuple_\dimension') \cdot \bm{\fe3}^{(\dimension)} &= f^x(\fingerprinttuple') - t_k^{(\dimension+1)}.
    \end{align*}
    
    Bob then finishes the simulation of \cref{step:sumcheck-simulator-algebraic-commit} with the coordinate $k \in [\pcommitlength]$.
    
    In \cref{step:sumcheck-temporal-decommit}, if $\fingerprinttuple' \neq \fingerprinttuple$ or the simulation fails (i.e., $\vcommitstring_\ell = \fingerprinttuple$ but $(\fingerprinttuple, \ell) \notin C$), Bob accepts or rejects uniformly at random. Otherwise, he simulates \cref{step:sumcheck-simulator-algebraic-decommit} until the protocol terminates (which his access to the shared random strings $t^{(i)}$ enables him to).
    At the end of the simulation, Bob accepts if and only if $D$ accepts.
    
    Note that, when $\bm{\eta} = \bm{\tau} - (w_{ik} : i \in [\ell])$ for a vector $\bm{\tau}$ that maps to the polynomials $(g_i : i \in [\dimension])$ via $\bm{\xi} = \bm{\xi}(\fingerprinttuple)$, then $\bm{\fe3}^{(i)}_j$ satisfies
    \begin{align*}
        \bm{\fe3}^{(i)}_j &= (\bm{\xi} \cdot \bm{\eta})_{(i - 1) \degree + j} -  \pcommitstring'^{(i)}_j\\
        &= g_i(j) - (\bm{\xi} \cdot w)_{(i - 1) \degree + j, k} - \pcommitstring'^{(i)}_{jk}\\
        &= g_i(j) - \pcommitstring^{(i)}_{jk}
    \end{align*}
    for all $i,j$ in $[\dimension - 1] \times [\degree]$ and $\set{\dimension} \times [\degree - 1]$ (equivalently, for all $i,j$ such that $\pcommitstring^{(i)}_j$ includes a linear combination of the rows of $w$). Then the linear constraints satisfied by the other $\dimension + 1$ pairs ensures the equality extends to all $(i, j)$: for $i \in [\dimension], j = \degree + 1$ and $(i, j) = (\dimension, \degree)$, we have
    
    \begin{align*}
        \bm{\theta} \cdot \bm{\fe3}^{(1)} &= \fe1 - t_k^{(1)}\\
        &= \sum_{j = 1}^{\degree + 1} \bm{\theta}_j g(j) - t_k^{(1)}\\
        &= \sum_{j = 1}^{\degree + 1} \bm{\theta}_j \left(g(j) - \pcommitstring_{jk}^{(1)}\right),\\
        \bm{\chi}(\fingerprinttuple_\dimension) \cdot \bm{\fe3}^{(\dimension)} &= f^x(\fingerprinttuple) - t_k^{(\dimension+1)}\\
        &= g_\dimension(\fingerprinttuple_\dimension) - t_k^{(\dimension+1)}\\
        &= \sum_{j = 1}^{\degree + 1} \chi_j(\fingerprinttuple_\dimension) \left(g_\dimension(j) - \pcommitstring^{(\dimension)}_{jk}\right),
    \end{align*}
    and, for $1 < i \leq \dimension$,
    \begin{align*}
        \bm{\theta} \cdot \bm{\fe3}^{(i)} - \bm{\chi}(\fingerprinttuple_{i-1}) \cdot \bm{\fe3}^{(i-1)} &= -t_k^{(i)}\\
        &= \sum_{j = 1}^{\degree + 1} \left(\chi_j(\fingerprinttuple_{i-1}) \pcommitstring^{(i-1)}_{jk} - \bm{\theta}_j \pcommitstring^{(i)}_{jk}\right)\\
         &= \sum_{j = 1}^{\degree + 1} \bm{\theta}_j \left(g_i(j)- \pcommitstring^{(i)}_{jk}\right) + \sum_{j = 1}^{\degree + 1} \chi_j(\fingerprinttuple_{i-1}) \left(g_{i-1}(j)- \pcommitstring^{(i-1)}_{jk}\right).
    \end{align*}
    
     That is, since the $\bm{\fe3}^{(i)}$ satisfy the same linear constraints as the vectors $\big(g_i(j) - \pcommitstring^{(i)}_{jk} : j \in [\degree + 1]\big)$, it follows that they are equal. Therefore the resulting view is distributed \emph{exactly} as $\view_{P, \widetilde V}(x, r)$ when $\bm{\tau}$ maps to the partial sums of $f^x$ (and thus $\bm{\xi}(\fingerprinttuple) \cdot \tau$ maps to the partial sums with respect to $\fingerprinttuple$); and if $\bm{\eta} \sim \field^{\ell}$, it is distributed as $S(\widetilde V, x, r)$ (unless the simulation fails or $\fingerprinttuple \neq \fingerprinttuple'$).

    This one-way protocol achieves bias $0$ when the simulation fails (an $o(1)$-probability event) or the verifier's temporal decommitment $\fingerprinttuple$ is in $C$ (i.e., the simulation succeeds) but $\fingerprinttuple  \neq \fingerprinttuple'$, an event with conditional probability $1 -\frac1{\abs{C}} = 1 - \frac1{s}$. Otherwise, it achieves a bias of $\eps/2$. We thus have
    \begin{align*}
        \P_{\substack{w \sim \field^{\ell \times \pcommitlength}\\k \sim [\pcommitlength]}}&\left[B\left(A(w), \big(f^x(i) - w_{ik} : i \in [\ell]\big), k\right)\text{ accepts}\right]\\
        &- \P_{\substack{w \sim \field^{\ell \times \pcommitlength}\\k \sim [\pcommitlength]\\\bm{\eta} \sim \field^\ell}}\left[B\left(A(w), \bm{\eta}, k\right)\text{ accepts}\right]\\
        &= o(1) \cdot 0 + \big(1 - o(1)\big) \cdot \left(1 - \frac1 s\right) \cdot 0 + \big(1 - o(1)\big) \cdot \frac1 s \cdot \frac{\eps}{2}\\
        &\geq \frac{\eps}{3s}.
    \end{align*}

    Applying \cref{lem:string-to-bit-index} yields a one-way binary \indexproblem protocol for strings of length $\pcommitlength = \fieldsize^{\log\log \fieldsize} $ with messages of length $\frac{s^2 \ell^2 \log^2 \fieldsize}{\eps^2} \poly(\fieldsize) = \poly(\fieldsize)$ and constant bias. But this contradicts \cref{prop:index-hardness}'s upper bound of $O\left(\sqrt{\poly(\fieldsize)/\pcommitlength}\right) = o(1)$.
\end{proof}

\subsection{Applications: \frequencymoment and \innerproductproblem}
\label{sec:sumcheck-applications}

We now proceed to applications of \textsf{zk-sumcheck}. The first is a \zksip{} that (exactly) computes frequency moments of order $k > 1$ (commonly denoted $F_k$) for a stream over an alphabet of size $\ell$, a problem known to require $\Omega(\ell)$ space without a prover \cite{AMS99}.

\begin{definition}
    \label{def:frequency-moment}
    Fix $k \in \N$. For every $\ell \in [n]$ and $t \in [n^k]$, the language $\frequencymoment_k(t)$ is $\set{x \in [\ell]^n : \sum_{i \in [\ell]} \varphi_i(x)^k = t}$, where $\varphi_i(x) \coloneqq \abs{\set{j \in [n] : x_j = i}}$.
\end{definition}

\begin{corollary}
    \label{cor:frequency-moment}
    Fix $1 < k \in \N$ and $\delta \in (0, 1]$. For every $\ell \in [n]$ and $t \in [n^k]$, there exists a zero-knowledge SIP for $\frequencymoment_k(t)$ with space complexity $O(\log^2 n / \log\log n)$. The communication complexity is $O(n^{1 + \delta})$ in the setup and $n^{o(1)}$ in the interactive phase, and the protocol is secure against $\polylog(n)$-space distinguishers.
\end{corollary}
\begin{proof}
    We set parameters analogously to \cref{cor:index}, but take into account the factor-$k$ blowup in the degree of $f^x$: set degree $\degree = k \log^{\frac{2}{\delta}} n = O\left(\log^{\frac{2}{\delta}} n\right)$, dimension $\dimension = \frac{\delta \log n}{2\log\log n}$, and take a field $\field$ of size $\abs{\field} = \fieldsize = \Theta\left(\log^{1 + \frac{2}{\delta}} n\right)$. The mapping $x \mapsto f^x$ is defined as follows: viewing $[\ell] \hookrightarrow [\degree + 1]^\dimension \hookrightarrow \field^\dimension$ and defining the frequency vector $\varphi = \varphi(x) \coloneqq \big(\varphi_i(x) : i \in [\ell]\big)$, set $f^x(\bm{\fe1}) \coloneqq \sum_{i \in [\degree + 1]} \hat \varphi(i, \bm{\fe1})^k$ for $\bm{\fe1} \in \field^{\dimension - 1}$, where $\hat \varphi$ is the degree-$d/k$ extension of $\hat \varphi$. Note that $f^x$ is a $(\dimension - 1)$-variate degree-$\degree$ polynomial.
    
    Using $O(\degree\dimension \log \fieldsize) = O(\dimension^2 \log \fieldsize)$ bits of space (recall that $k$ is constant), the verifier can compute all the low-degree extensions $\hat \varphi(i, \fingerprinttuple) \in \field$ (by adding $\chi_{x_j}(i, \fingerprinttuple)$ to each running sum upon reading $x_j$); then, after the stream, $V$ raises each LDE to the $k^\text{th}$ power and adds the results to obtain $f^x(\fingerprinttuple)$.
    
    Applying \cref{prot:zk-sumcheck}, the verifier checks whether
    \begin{equation*}
        \sum_{\bm{\fe1} \in [\degree + 1]^{\dimension - 1}} f^x(\bm{\fe1}) = \sum_{\evaluationpoint \in [\degree + 1]^\dimension} \hat \varphi (\evaluationpoint)^k = \sum_{i \in [\ell]} \varphi_i^k
    \end{equation*}
    is equal to $t$. The space complexity is $O(\dimension^2 \log \fieldsize) = O(\log^2 n / \log\log n)$; the communication complexity of the setup step is of order
    \begin{equation*}
        \fieldsize^\dimension \dimension \log^2 \fieldsize = n^{1 + \frac{\delta}{2}}\polylog(n) = O\left(n^{1 + \delta}\right),
    \end{equation*}
    and $\fieldsize^{\log\log \fieldsize} \poly(\fieldsize) = n^{o(1)}$ in the interactive phase. Lastly, the protocol is secure against distinguishers with space $\poly(\fieldsize) = \polylog(n)$.
\end{proof}

Our second and last last application is a small modification of the $F_2$ protocol that allows us to compute inner products.
\begin{definition}
    \label{def:inner-product}
    For every $\ell \in [n]$, $t \in [n^2 \ell]$ and field $\field$, the language $\innerproductproblem(t)$ is defined as $\set{(x,y) \in \field^n \times \field^n : \varphi(x) \cdot \varphi(y) = \sum_{i \in [\ell]} \varphi_i(x) \varphi_i(y) = t}$.
\end{definition}

\begin{corollary}
    \label{cor:inner-product}
    For every $\delta \in (0, 1]$, $\ell \in [n]$, $t \in [n^2 \ell]$ and field $\field_\fieldsize$ with $\fieldsize = \Theta\left(\log^{1 + \frac{2}{\delta}} n\right)$, there exists a \zksip{} for $\innerproductproblem(t)$ with space complexity $O(\log^2 n / \log\log n)$ and communication complexities $O(n^{1 + \delta})$ and $n^{o(1)}$ in the setup and communication phases, respectively.
\end{corollary}
\begin{proof}
    We use the same parameter settings as \cref{cor:frequency-moment} and define
    \begin{equation*}
        f^{x,y}(\bm{\fe1}) = \sum_{i \in [\degree + 1]} \widehat{\varphi(x)}(i, \bm{\fe1}) \widehat{\varphi(y)}(i, \bm{\fe1}),
    \end{equation*}
    a polynomial of degree $2\degree = 2\log^{\frac{2}{\delta}} n$ whose evaluation the verifier computes by saving $\widehat{\varphi(x)}(i,\fingerprinttuple)$ and $\widehat{\varphi(y)}(i,\fingerprinttuple)$ for $i \in [\degree + 1]$. \cref{prot:sumcheck} enables the verifier to check that $\sum_{i \in [\ell]} \varphi_i(x) \varphi_i(y)$ equals $t$, as desired, with complexities of the same order as in \cref{cor:frequency-moment}.
\end{proof}

We remark that while one might reduce inner product to $F_2$, by taking the difference between the second moment of $\varphi(x) + \varphi(y)$ and the second moments of $\varphi(x)$ and $\varphi(y)$, the resulting protocol leaks these values, and is therefore not zero-knowledge.

\ifblind
\else
\section*{Acknowledgements}

We thank Aditya Prakash for the proof of \cref{clm:probability-ip}, as well as Justin Thaler and Nick Spooner for fruitful discussions and careful reading of an earlier version of this manuscript. 
\fi

\iflipics
\bibliographystyle{plainurl}
\else
\bibliographystyle{alpha}
\fi
\bibliography{references}

\appendix

\section{Deferred proofs}
\label{sec:deferred-proofs}

\subsection{Proof of \texorpdfstring{\cref{prop:index-hardness}}{Theorem \ref{prop:index-hardness}}}
\label{sec:deferred-index-hardness}
\begin{proposition}[\cref{prop:index-hardness}, restated]
    Any one-way communication protocol for \searchindex with input $(x, j) \sim \alphabet^\pcommitlength \times [\pcommitlength]$ that sends an $s$-bit message succeeds with probability at most $\frac1{\abs{\alphabet}}+O\left(\sqrt{s/\pcommitlength}\right)$.
\end{proposition}
\begin{proof}
    Define, for ease of notation, $\alphabetsize = \abs{\alphabet}$. We follow the strategy used in \cite{RY20} for the binary case. First, note that by the minimax theorem we may assume Alice's and Bob's strategies are deterministic; i.e., that Alice sends $A(x) \in \bitset^s$ and Bob outputs $B(A(x),j) \in \alphabet$ for some functions $A$ and $B$.

    Let $\lambda$ be the distribution of Alice's message $A = A(x)$ induced by the (uniform) distribution of $x$, partitioning $\alphabet^\pcommitlength$ into $\set{P_a}$ where $P_a = A^{-1}(a) = \set{x \in \alphabet^\pcommitlength : A(x) = a}$. Note that the distribution of $x$ conditioned on $A = a$ is uniform over $P_a$, and that $\P_{A \sim \lambda}[A = a] = \abs{P_a}/\alphabetsize^\pcommitlength$. Then,
    \begin{align}
    \label{ineq:index}
        \P_{\subalign{x &\sim \alphabet^\pcommitlength\\j &\sim [\pcommitlength]}}[\text{Bob outputs } x_j] &= \sum_{a \in \bitset^s}\P_{x \sim \alphabet^\pcommitlength}[A(x) = a] \cdot \P_{\subalign{x &\sim \alphabet^\pcommitlength\\j &\sim [\pcommitlength]}}\big[b(a,j) = x_j ~\big|~ A(x) = a\big] \nonumber\\
        &= \sum_{a \in \bitset^s}\P_{A \sim \lambda}[A = a] \cdot \P_{\subalign{x &\sim P_a\\j &\sim [\pcommitlength]}}\big[b(a,j) = x_j\big]\nonumber\\
        &= \E_{\subalign{A &\sim \lambda\\j &\sim [\pcommitlength]}}\left[\P_{x \sim P_A}\big[b(A,j) = x_j\big]\right]\nonumber\\
        &\leq \E_{\subalign{A &\sim \lambda\\j &\sim [\pcommitlength]}}\left[\max_{\symbol1 \in \alphabet}\left\{\P_{x \sim P_A}[x_j = \symbol1]\right\}\right],
    \end{align}
    so that we only need to bound the latter expression; note that the inequality shows Bob's optimal strategy is to output the most frequent symbol at the $j^\text{th}$ coordinate in $P_A$.

    Now, define $\mu$ as the uniform distribution over $\alphabet$ and $\mu_{i, a}$ as the distribution of $x_i$ when $x \sim P_a$ (i.e., the distribution of $x_i$ when $x \sim \alphabet^\pcommitlength$ conditioned on $A(x) = a$). Then, by Pinsker's inequality (\cref{eq:pinsker}), for all $a \in \Im A$ and $i \in [\pcommitlength]$ we have
    \begin{equation*}
        \norm{\mu_{i, a} - \mu}^2 \leq \frac{\mathrm{KL}\big(\mu_{i,a} ~||~ \mu\big)}{2 \ln 2}
    \end{equation*}
    (where we use $\norm{\cdot}$ as shorthand for the $2$-norm $\norm{\cdot}_2$). Since the inequality holds for all $a$ and $i$, then it also holds for the convex combination corresponding to taking $A \sim \lambda$ and $j \sim [\pcommitlength]$ independently (i.e., whose coefficients are $\P[A = a, j = i] = \frac{\abs{P_a}}{\alphabetsize^\pcommitlength\pcommitlength}$). Therefore,
    \begin{align*}
        \E_{\subalign{A &\sim \lambda\\j &\sim [\pcommitlength]}}\left[\norm{\mu_{j,A} - \mu}^2\right] &= \frac1{\pcommitlength} \sum_{i = 1}^\pcommitlength \E_{A \sim \lambda}\left[\norm{\mu_{i,A} - \mu}^2\right]\\
        &\leq \frac1{2 \pcommitlength \ln 2} \sum_{i = 1}^\pcommitlength \E_{A \sim \lambda}\left[\mathrm{KL}(\mu_{i,A} ~||~ \mu)\right]\\
        &= \frac1{2 \pcommitlength \ln 2} \sum_{i = 1}^\pcommitlength I(A : x_i),
    \end{align*}
    where the last equality follows by the definition of mutual information (\cref{eq:mutual-information}). By convexity of $z \mapsto z^2$, we have
    \begin{align*}
        \E_{\subalign{A &\sim \lambda\\j &\sim [\pcommitlength]}}\big[\norm{\mu_{j,A} - \mu}\big]^2 &\leq \E_{\subalign{A &\sim \lambda\\j &\sim [\pcommitlength]}}\left[\norm{\mu_{j,A} - \mu}^2\right]\\
        &\leq \frac1{2 \pcommitlength \ln 2} \sum_{i = 1}^\pcommitlength I(A : x_i).
    \end{align*}

    Recall that $\mu_{i,a}(\symbol1) = \P_{x \sim P_a}[x_i = \symbol1]$. Comparing this value with the average mass $1/\alphabetsize$, we have
    \begin{align*}
        \E_{\subalign{A &\sim \lambda\\j &\sim [\pcommitlength]}}\left[\max_{\symbol1 \in \alphabet}\left\{\P_{x \sim P_A}[x_j = \symbol1]\right\}\right] - \frac1{\alphabetsize} &=  \E_{\subalign{A &\sim \lambda\\j &\sim [\pcommitlength]}}\left[\max_{\symbol1 \in \alphabet}\left\{\mu_{j,A}(\alpha) - \frac1{\alphabetsize}\right\}\right] \\
        &\leq  \E_{\subalign{A &\sim \lambda\\j &\sim [\pcommitlength]}}\left[\max_{\symbol1 \in \alphabet}\left\{\abs{\mu_{j,A}(\alpha) - \frac1{\alphabetsize}}\right\}\right] \\
        &\leq \E_{\subalign{A &\sim \lambda\\j &\sim [\pcommitlength]}}\big[\norm{\mu_{j,A} - \mu}\big]\\
        &\leq \sqrt{\frac{\sum_{i = 1}^\pcommitlength I(A : x_i)}{2 \pcommitlength \ln 2}},
    \end{align*}
    so that using \cref{ineq:index} and rearranging,
    \begin{equation*}
        \P_{\substack{x \sim \alphabet^\pcommitlength\\j \sim [\pcommitlength]}}[\text{Bob outputs } x_j] \leq \frac1{\alphabetsize} + \sqrt{\frac{\sum_{i = 1}^\pcommitlength I(A : x_i)}{2\pcommitlength \ln 2}}.
    \end{equation*}

    The theorem thus reduces to showing $\sum_{i = 1}^\pcommitlength I(A : x_i) \leq s$. By standard information-theoretic equivalences and inequalities,
    \begin{flalign*}
        &&\sum_{i = 1}^\pcommitlength I(A : x_i) &= \sum_{i = 1}^\pcommitlength \big(H(x_i) - H(x_i|A)\big) & (\text{by \cref{eq:mutual-information}})\\
        &&&= H(x) - \sum_{i = 1}^\pcommitlength H(x_j|A) & (\text{by \cref{eq:entropy-independence}})\\
        &&&\leq H(x) - \sum_{i = 1}^n H(x_i|x_1,\ldots,x_{i-1}, A) & (\text{by \cref{eq:entropy-conditioning}})\\
        &&&= H(x) - H(x|A) & (\text{by \cref{eq:chain-rule-entropy}})\\
        &&&= I(A:x) \leq H(A) & (\text{by \cref{eq:mutual-information}})\\
        &&&\leq  s & (\text{by \cref{eq:entropy-bounds}})
    \end{flalign*}
    and the result follows.
\end{proof}

\subsection{Proof of \texorpdfstring{\cref{thm:pv-algebraic-commitment}}{Theorem \ref{thm:pv-algebraic-commitment}}}
\label{sec:deferred-pv-algebraic-commit}

\begin{theorem}[\cref{thm:pv-algebraic-commitment}, restated]
    \cref{prot:algebraic-commit} (\textsf{algebraic-commit}) and \cref{prot:decommit} (\textsf{decommit}) form a streaming commitment protocol with space complexity $s = O\big((\ell + \dimension) \log \fieldsize\big)$ if $\pcommitlength = \fieldsize^{3\ell}$ and $\degree \dimension = \polylog(\fieldsize)$. The scheme is secure against $\poly(s)$-space adversaries and communicates $O(\ell \fieldsize^{3 \ell} \log \fieldsize)$ bits.
    
    Furthermore, if each linear coefficient can be computed in $O(\dimension \log \fieldsize)$ space, then $s = O(\dimension \log \fieldsize)$.
\end{theorem}

\begin{proof}
    We follow the same steps of \cref{thm:pv-commitment}, beginning with the binding property: using $\pcommitstring^{(i)}$ to denote the $i^\text{th}$ column of $\pcommitstring$, when $P$ is honest, i.e., sends the correction tuple $\bm{\fe3} = \bm{\fe1} - \pcommitstring^{(k)}$ in the \textsf{commit} stage and the polynomial $\hat z_{|\line}$ where $z = \bm{\fe2} \cdot \pcommitstring$ in the \textsf{decommit} stage, then $V$ accepts as $\hat z_{|\line}(\rfe1) = \hat z(\fingerprinttuple) = \hat \pcommitstring(\fingerprinttuple, \bm{\fe2})$ and $\hat z_{|\line}(0) + \fe3 = z_k + \bm{\fe2} \cdot \bm{\fe3} = \bm{\fe1} \cdot \bm{\fe2}$. (Recall that the line $\line$ is such that $\line(0) = k$ and $\line(\rfe1) = \fingerprinttuple$.)
    
    Now, suppose $P$ replies with a polynomial $g$ such that $g(0) \neq \sum_{i \in [\ell]} \bm{\fe2}_i
    \pcommitstring_{ik} = z_k = \hat z_{|\line}(0)$; then the Schwartz-Zippel lemma implies $g(\rfe1) \neq \hat z_{|\line}(\rfe1)$ except with probability $\degree \dimension / \fieldsize = o(1)$, in which case $V$ rejects. As the verifier only needs to store the evaluation point $\fingerprinttuple \in \field^\dimension$, the coordinate $k \in [\pcommitlength]$ and a constant number of additional field elements, its space complexity is $O(\dimension \log \fieldsize)$ as long as each $\bm{\fe2}_i$ can be computed in this space (e.g., when $\bm{\fe2}_i = \bm{\fe2}_i(\rfe1)$ is the evaluation of an $\dimension$-variate polynomial over $\field)$; if $\bm{\fe2}$ must be stored in its entirety, the complexity becomes $O\big((\ell + \dimension) \log \fieldsize\big)$.

    To show the hiding property, assume towards contradiction that there exists a streaming algorithm $D$ with space $\poly(s) = \poly(\ell, \log \fieldsize)$ that distinguishes commitments between some $\bm{\fe1} \in \field^\ell$ and $\bm{\fe1}' \in \field \setminus \set{\bm{\fe1}}$ with constant bias: that is,
    \begin{align*}
        \P_{\subalign{\pcommitstring &\sim \field^{\ell \times \pcommitlength}\\k &\sim [\pcommitlength]}}\left[D(\pcommitstring, \bm{\fe1} - \pcommitstring^{(k)}, k) \text{ accepts}\right] -  \P_{\subalign{\pcommitstring &\sim \field^\pcommitlength\\k &\sim [\pcommitlength]}}\left[D(\pcommitstring, \bm{\fe1}' - \pcommitstring^{(k)}, k) \text{ accepts}\right] \geq \eps
    \end{align*}
    for some $\eps = \Omega(1)$. Now consider the following one-way communication protocol for \searchindex over the alphabet $\field^\ell$ with input $(x,j) \in (\field^\ell)^\pcommitlength \times [\pcommitlength]$: Alice, viewing $x$ as an element of $\field^{\ell \times \pcommitlength}$, simulates $D$ on the stream $(x,\bm{\fe3})$, where $\bm{\fe3} \sim \field^\ell$, and sends the $\polylog(\pcommitlength)$-bit snapshot of $D$ to Bob, who finishes the simulation with $j$; if $D$ accepts output $\bm{\fe1} - \bm{\fe3}$, and otherwise output $\bm{\fe1}' - \bm{\fe3}$. Note that Bob outputs correctly exactly when $\bm{\fe3} = \bm{\fe1} - \pcommitstring^{(k)}$ and $D$ accepts, or $\bm{\fe3} = \bm{\fe1}' - \pcommitstring^{(k)}$ and $D$ rejects. We will now show that the protocol solves \searchindex with a bias that is too large, contradicting \cref{prop:index-hardness}.

    \begin{align*}
        \P_{\subalign{x &\sim (\field^\ell)^\pcommitlength\\j &\sim [\pcommitlength]}}&[\text{Bob outputs } x_j]\\
        &= \frac1{\fieldsize^\ell} \cdot \P_{\subalign{x &\sim (\field^\ell)^\pcommitlength\\j &\sim [\pcommitlength]}}\left[D(x,\bm{\fe1} - x_j, j) \text{ accepts}\right] + \frac1{\fieldsize^\ell} \cdot \P_{\subalign{x &\sim (\field^\ell)^\pcommitlength\\j &\sim [\pcommitlength]}}\left[D(x,\bm{\fe1}' - x_j, j) \text{ rejects}\right]\\
        &= \frac1{\fieldsize^\ell} \left(1 + \P_{\subalign{x &\sim (\field^\ell)^\pcommitlength\\j &\sim [\pcommitlength]}}\left[D(x,\bm{\fe1} - x_j, j) \text{ accepts}\right] - \P_{\subalign{x &\sim (\field^\ell)^\pcommitlength\\j &\sim [\pcommitlength]}}\left[D(x,\bm{\fe1}' - x_j, j) \text{ accepts}\right] \right)\\
        &\geq \frac{1 + \eps}{\fieldsize^\ell}\\
        &= \frac1{\fieldsize^\ell} + \Omega\left(\frac1{\fieldsize^\ell}\right).
    \end{align*}

    Since $\fieldsize^{-\ell} = \Omega\left(\sqrt{\fieldsize^\ell/\pcommitlength}\right)= \omega\left(\sqrt{\poly(s)/\pcommitlength}\right)$, owing to $s = \poly(\ell, \log \fieldsize)$, the result follows. The communication complexity of the protocols is dominated by the prover sending $\ell \pcommitlength$ field elements, for a total of $O(\ell \fieldsize^{3 \ell} \log \fieldsize)$ bits.
\end{proof}

\subsection{Proof of \texorpdfstring{\cref{clm:probability-ip}}{Claim \ref{clm:probability-ip}}}
\label{sec:deferred-probability-ip}

\begin{claim}[\cref{clm:probability-ip}, restated]
    Let $p, q \in [0,1]^\vcommitlength$ be probability vectors and $t \in [\vcommitlength]$ a positive integer. There exists a set $C \subseteq [\vcommitlength]$ of size $t$ such that $\sum_{i \in [\vcommitlength] \setminus C} p_i q_i \leq 1/t$.
\end{claim}
\begin{proof}
    We reduce the claim to proving an upper bound on a certain optimisation problem. Namely, let $\Delta = \set{x \in [0,1]^\vcommitlength : \sum_i x_i = 1}$ and $\Delta' = \Delta \cap \set{x \in [0,1]^\vcommitlength : x_1 \geq \cdots \geq x_\vcommitlength}$ be the $\vcommitlength$-dimensional simplex and the simplex with ordered coordinates, respectively. Define the function $f: \Delta' \times \Delta \to \R_+$ by $f(p,q) = \sum_{i = 1}^\vcommitlength i p_i q_i$.
    
    Under the assumption that $f(p,q) \leq 1$ for all $p \in \Delta'$ and $q \in \Delta$, we conclude as follows: since $p_1 \geq p_2 \geq \cdots \geq p_\vcommitlength$ without loss of generality (permuting the vectors to satisfy the condition does not affect the truth of the claim), for any $t \in [\vcommitlength]$
    \begin{equation*}
        1 \geq f(p,q) = \sum_{i = 1}^\vcommitlength \left(\sum_{j = i}^{\vcommitlength} p_j q_j \right) \geq \sum_{i = 1}^t \left( \sum_{j = i}^\vcommitlength p_j q_j \right)
    \end{equation*}
    implies the existence of $i \in [t]$ such that $\sum_{j = i}^\vcommitlength p_j q_j \leq 1/t$. Taking $C = [i - 1]$ completes the proof.

    We now proceed to show $f(p,q) \leq 1$. Since $f$ is continuous with compact domain, there exists a pair $(p^*, q^*)$ that maximises $f$. Let $\ell \in [\vcommitlength]$ be the largest nonzero coordinate of $p^*$. Then $q_i^* > 0$ for all $i \leq \ell$, as otherwise moving the mass $p_i^*$ onto $p_1^*$ would contradict maximality;  and $q_i^* = 0$ for all $i > \ell$, or moving $q_i^*$ onto (say) $q_1^*$ likewise leads to a contradiction.

    Now, suppose (towards contradiction) $\ell > 1$, take $1 < j \leq \ell$ and consider the pair $(p^*, q')$ with $q_1' = 0$, $q_i' = q_1^* + q_i^*$ and $q_j' = q_j^*$ otherwise. Then $f(p^*, q') \leq f(p^*, q^*)$ implies
    \begin{equation*}
      i p_i^* (q_1^* + q_i^*) \leq p_1^* q_1^* + i p_i^* q_i^*,
    \end{equation*}
    and thus $i p_i^* \leq p_1^*$ (since $q_1^* \neq 0$). But then
    \begin{equation*}
        f(p^*, q^*) = \sum_{i = 1}^\ell i p_i^* q_i^* \leq p_1^* \sum_{i = 1}^\ell q_i^* = p_1^* < 1,
    \end{equation*}
    a contradiction, as the delta distributions at $1$ achieve value $1$.
    
    We thus conclude that $\ell = 1$, so the maximisers $p^*, q^*$ are the delta distributions at $1$ and $f(p,q) \leq  f(p^*,q^*) = 1$, as desired.
\end{proof}

\end{document}